\documentclass[a4paper]{article}
\usepackage{amsfonts,amsmath,amssymb,mathtools,mathrsfs,bbm,float}
\usepackage[ruled, vlined, linesnumbered]{algorithm2e}
\usepackage{graphicx,epic,eepic,epsfig,latexsym,verbatim,color}
\usepackage[inline,shortlabels]{enumitem}
\usepackage[caption=false]{subfig}
\usepackage[dvipsnames]{xcolor}
 
\usepackage{tikz}
\usetikzlibrary{chains}
\usetikzlibrary{fit}

\usepackage{epsfig}
\usetikzlibrary{shapes.symbols,patterns} 
\usepackage{pgfplots}

\usepackage[strict]{changepage}

\usepackage[marginal]{footmisc}
\usepackage{url}
\usepackage{theorem}
\usepackage{thmtools}
\usepackage{nicematrix}
\usepackage{xfrac}

\newtheorem{definition}{Definition}
\newtheorem{proposition}{Proposition}
\newtheorem{lemma}[proposition]{Lemma}

\newtheorem{theorem}[proposition]{Theorem}
\newtheorem{corollary}[proposition]{Corollary}

\newenvironment{proof}{\noindent \textbf{{Proof~} }}{\hfill $\blacksquare$}

\def\squareforqed{\hbox{\rlap{$\sqcap$}$\sqcup$}}
\def\qed{\ifmmode\squareforqed\else{\unskip\nobreak\hfil
\penalty50\hskip1em\null\nobreak\hfil\squareforqed
\parfillskip=0pt\finalhyphendemerits=0\endgraf}\fi}
\def\endenv{\ifmmode\;\else{\unskip\nobreak\hfil
\penalty50\hskip1em\null\nobreak\hfil\;
\parfillskip=0pt\finalhyphendemerits=0\endgraf}\fi}

\newcounter{example}

\mathchardef\ordinarycolon\mathcode`\:
\mathcode`\:=\string"8000
\def\vcentcolon{\mathrel{\mathop\ordinarycolon}}
\begingroup \catcode`\:=\active
  \lowercase{\endgroup
  \let :\vcentcolon
  }

\usepackage{hyperref}
\hypersetup{colorlinks=true,citecolor=blue,linkcolor=blue,filecolor=blue,urlcolor=blue,breaklinks=true}
\usepackage{footnotebackref}
\usepackage{cleveref}
\usepackage{graphicx}

\definecolor{darkblue}{RGB}{0,76,156}
\definecolor{darkkblue}{RGB}{0,0,153}
\definecolor{blue2}{RGB}{102,178,255}
\definecolor{darkred}{RGB}{195,0,0}

\RequirePackage[framemethod=default]{mdframed}
\newmdenv[skipabove=7pt,
skipbelow=7pt,
backgroundcolor=darkblue!15,
innerleftmargin=5pt,
innerrightmargin=5pt,
innertopmargin=5pt,
leftmargin=0cm,
rightmargin=0cm,
innerbottommargin=5pt,
linewidth=1pt]{tBox}

\newmdenv[skipabove=7pt,
skipbelow=7pt,
backgroundcolor=blue2!25,
innerleftmargin=5pt,
innerrightmargin=5pt,
innertopmargin=5pt,
leftmargin=0cm,
rightmargin=0cm,
innerbottommargin=5pt,
linewidth=1pt]{dBox}

\newmdenv[skipabove=7pt,
skipbelow=7pt,
backgroundcolor=darkred!15,
innerleftmargin=5pt,
innerrightmargin=5pt,
innertopmargin=5pt,
leftmargin=0cm,
rightmargin=0cm,
innerbottommargin=5pt,
linewidth=1pt]{rBox}

\newcommand{\nc}{\newcommand}
\nc{\bra}[1]{\langle#1|}
\nc{\ket}[1]{|#1\rangle}
\nc{\ketbra}[2]{\lvert#1\rangle\!\langle#2\rvert}
\nc{\braket}[2]{\langle#1|#2\rangle}
\newcommand{\braandket}[3]{\langle #1|#2|#3\rangle}

\DeclarePairedDelimiter{\norm}{\lVert}{\rVert}
\DeclarePairedDelimiter{\abs}{\lvert}{\rvert}
\DeclarePairedDelimiter{\ceil}{\lceil}{\rceil}

\DeclarePairedDelimiterX{\infdivx}[2]{(}{)}{%
  #1\;\delimsize\|\;#2%
}

\nc{\proj}[1]{| #1\rangle\!\langle #1 |}
\nc{\avg}[1]{\langle#1\rangle}
\nc{\smfrac}[2]{\mbox{$\frac{#1}{#2}$}}
\nc{\tr}{\operatorname{Tr}}
\nc{\ox}{\otimes}
\nc{\dg}{\dagger}
\nc{\dn}{\downarrow}
\nc{\cA}{{\cal A}}
\nc{\cB}{{\cal B}}
\nc{\cC}{{\cal C}}
\nc{\cD}{{\cal D}}
\nc{\cE}{{\cal E}}
\nc{\cF}{{\cal F}}
\nc{\cG}{{\cal G}}
\nc{\cH}{{\cal H}}
\nc{\cI}{{\cal I}}
\nc{\cJ}{{\cal J}}
\nc{\cK}{{\cal K}}
\nc{\cL}{{\cal L}}
\nc{\cM}{{\cal M}}
\nc{\cN}{{\cal N}}
\nc{\cO}{{\cal O}}
\nc{\cP}{{\cal P}}
\nc{\cQ}{{\cal Q}}
\nc{\cR}{{\cal R}}
\nc{\cS}{{\cal S}}
\nc{\cT}{{\cal T}}
\nc{\cU}{{\cal U}}
\nc{\cV}{{\cal V}}
\nc{\cX}{{\cal X}}
\nc{\cY}{{\cal Y}}
\nc{\cZ}{{\cal Z}}
\nc{\cW}{{\cal W}}
\nc{\csupp}{{\operatorname{csupp}}}
\nc{\qsupp}{{\operatorname{qsupp}}}
\nc{\var}{{\operatorname{var}}}
\nc{\rar}{\rightarrow}
\nc{\lrar}{\longrightarrow}
\nc{\polylog}{{\operatorname{polylog}}}
\nc{\wt}{{\operatorname{wt}}}
\nc{\supp}{{\operatorname{supp}}}

\nc{\argmin}{{\operatorname{argmin}}}

\makeatletter
\newcommand{\tpmod}[1]{{\@displayfalse\pmod{#1}}}
\makeatother

\def\x{\xi}

\nc{\RR}{{{\mathbb R}}}
\nc{\CC}{{{\mathbb C}}}
\nc{\FF}{{{\mathbb F}}}
\nc{\NN}{{{\mathbb N}}}
\nc{\ZZ}{{{\mathbb Z}}}
\nc{\PP}{{{\mathbb P}}}
\nc{\QQ}{{{\mathbb Q}}}
\nc{\UU}{{{\mathbb U}}}
\nc{\EE}{{{\mathbb E}}}
\nc{\id}{{\operatorname{id}}}

\nc{\CHSH}{{\operatorname{CHSH}}}

\nc{\rU}{\mbox{U}}

\nc{\ob}[1]{#1}

\nc{\SEP}{{\text{\rm SEP}}}
\nc{\NS}{{\text{\rm NS}}}
\nc{\LOCC}{{\text{\rm LOCC}}}
\nc{\PPT}{{\text{\rm PPT}}}
\nc{\EXT}{{\text{\rm EXT}}}
\nc{\Sym}{{\operatorname{Sym}}}

\nc{\ERLO}{{E_{\text{r,LO}}}}
\nc{\ERLOCC}{{E_{\text{r,LOCC}}}}
\nc{\ERPPT}{{E_{\text{r,PPT}}}}
\nc{\ERLOCCinfty}{{E^{\infty}_{\text{r,LOCC}}}}
\nc{\Aram}{{\operatorname{\sf A}}}
\newtheorem{problem}{Problem}

\makeatletter
\def\grd@save@target#1{%
  \def\grd@target{#1}}
\def\grd@save@start#1{%
  \def\grd@start{#1}}
\tikzset{
  grid with coordinates/.style={
    to path={%
      \pgfextra{%
        \edef\grd@@target{(\tikztotarget)}%
        \tikz@scan@one@point\grd@save@target\grd@@target\relax
        \edef\grd@@start{(\tikztostart)}%
        \tikz@scan@one@point\grd@save@start\grd@@start\relax
        \draw[minor help lines,magenta] (\tikztostart) grid (\tikztotarget);
        \draw[major help lines] (\tikztostart) grid (\tikztotarget);
        \grd@start
        \pgfmathsetmacro{\grd@xa}{\the\pgf@x/1cm}
        \pgfmathsetmacro{\grd@ya}{\the\pgf@y/1cm}
        \grd@target
        \pgfmathsetmacro{\grd@xb}{\the\pgf@x/1cm}
        \pgfmathsetmacro{\grd@yb}{\the\pgf@y/1cm}
        \pgfmathsetmacro{\grd@xc}{\grd@xa + \pgfkeysvalueof{/tikz/grid with coordinates/major step}}
        \pgfmathsetmacro{\grd@yc}{\grd@ya + \pgfkeysvalueof{/tikz/grid with coordinates/major step}}
        \foreach \x in {\grd@xa,\grd@xc,...,\grd@xb}
        \node[anchor=north] at (\x,\grd@ya) {\pgfmathprintnumber{\x}};
        \foreach \y in {\grd@ya,\grd@yc,...,\grd@yb}
        \node[anchor=east] at (\grd@xa,\y) {\pgfmathprintnumber{\y}};
      }
    }
  },
  minor help lines/.style={
    help lines,
    step=\pgfkeysvalueof{/tikz/grid with coordinates/minor step}
  },
  major help lines/.style={
    help lines,
    line width=\pgfkeysvalueof{/tikz/grid with coordinates/major line width},
    step=\pgfkeysvalueof{/tikz/grid with coordinates/major step}
  },
  grid with coordinates/.cd,
  minor step/.initial=.2,
  major step/.initial=1,
  major line width/.initial=2pt,
}
\makeatother

\usepackage{thmtools}
\usepackage{thm-restate}
\usepackage{etoolbox}
\makeatletter
\def\problem@s{}
\newcounter{problems@cnt}

\newcommand{\allproblems}{\problem@s}
\makeatother

\usepackage{authblk} 
\usepackage[paperwidth=190.8mm,paperheight=297mm,centering,hmargin=20mm,vmargin=2cm]{geometry}

\usepackage{times}
\usepackage[utf8]{inputenc}
\usepackage[T1]{fontenc}
\pgfplotsset{compat=1.18}
\usepackage[numbers]{natbib}
\usepackage{multirow}
\usepackage{booktabs}
\usepackage{array}
\usepackage{cleveref}

\usepackage{tabularx}
\usepackage{soul}
\usepackage[export]{adjustbox}

\definecolor{colortwo}{rgb}{0.4,0.77,0.17}
\definecolor{colorthree}{rgb}{0.01,0.51,0.93}
\definecolor{darkgray}{rgb}{0.3,0.3,0.3}

\usepackage{tikz}
\usetikzlibrary{quantikz2}
\usetikzlibrary{external} 
\NiceMatrixOptions{cell-space-limits = 1pt}

\usepackage{soul}
\soulregister\cite7
\soulregister\ref7
\soulregister\pageref7

\newcommand{\trace}[2][]{\tr_{#1}\!\left[ #2 \right]}
\newcommand{\integral}[3]{\int_{#1}^{#2} {#3} \operatorname{d}\!{x}}
\NewDocumentCommand{\ddx}{o m}{%
  \IfNoValueTF{#1}
  {\frac{\operatorname{d}}{\operatorname{d}\!#2}}
  {\frac{\operatorname{d}^{\,#1}}{\operatorname{d}\!#2^{\,#1}}}%
}

\newcommand{\set}[1]{ \left\{\, #1 \,\right\} }
\newcommand{\setcond}[2]{ \left\{\, #1 :\, #2 \,\right\} }

\newcommand{\kett}[1]{|#1\rangle\!\rangle}

\newcounter{assumption}
\renewcommand{\theassumption}{\roman{assumption}}
\DeclareRobustCommand{\assumption}[2]{
  \refstepcounter{assumption}
  \label{#1}
  (\theassumption)\, #2%
}

\newcommand{\nref}[1]{\ref{#1}}

\newcommand{\Happrox}{{H^{\approx}}}
\newcommand{\ite}{{\phi(\tau)}}
\newcommand{\iteapprox}{{\phi^\approx(\tau)}}

\newcommand{\Hsys}{{H_\textrm{sys}}}

\newcommand{\poly}{\operatorname{poly}}
\newcommand{\BigO}[1]{{\mathcal{O}\!\left( #1 \right)}}
\newcommand{\BigTO}[1]{{\widetilde{\mathcal{O}}\!\left( #1 \right)}}
\newcommand{\BigOmega}[1]{{\Omega\!\left( #1 \right)}}

\newcommand{\UQPP}{V_{\theta^Y,\theta^Z}}
\newcommand{\UQPPs}[2]{V_{#1}^{#2}}
\newcommand{\aux}{{\operatorname{aux}}}
\newcommand{\UITE}[2]{\mathcal{V}_{#1}(#2)}
\newcommand{\UITEs}{\mathcal{V}}

\newcommand{\expf}{{f_{\tau, \lambda}}}
\newcommand{\expft}{{f_{t, \lambda}}}
\newcommand{\expfr}{{\xi_{\tau, \lambda}}}
\newcommand{\explb}{{B}}

\newcommand{\loss}[1]{\widehat{\omega}\!\left(#1\right)}
\newcommand{\qpploss}[1]{\widetilde{\omega}\!\left(#1\right)}
\newcommand{\measureloss}[1]{{\omega}\!\left(#1\right)}

\newcommand{\remain}[2]{R\!\left(#1; #2\right)}
\newcommand{\reldiff}{r}

\newcommand{\prob}[1]{\operatorname{Pr}\!\left(#1\right)}
\newcommand{\real}[1]{\Re\!\set{#1}}

\allowdisplaybreaks

\begin{document}
\title{Quantum Imaginary-Time Evolution with Polynomial Resources in Evolution Time}

\author[1, 2]{Lei Zhang}
\author[1]{Jizhe Lai}
\author[1, 2]{Xian Wu}
\author[1]{Xin Wang \thanks{felixxinwang@hkust-gz.edu.cn}}

\affil[1]{\small The Hong Kong University of Science and Technology (Guangzhou), Guangdong 511453, China}
\affil[2]{\small QudeLeap Research, Shanghai 200030, China}

\date{\today}
\maketitle

\begin{abstract}
Imaginary-time evolution is fundamental for analyzing quantum many-body systems, with applications spanning quantum chemistry, condensed matter physics, and quantum field theory, yet classical simulation requires exponentially growing resources in both system size and evolution time. While quantum approaches reduce the system-size scaling, existing methods rely on heuristic techniques with measurement precision or success probability that deteriorates as evolution time increases. We present a quantum algorithm that prepares normalized imaginary-time evolved states using an adaptive normalization factor to maintain a stable success probability over long imaginary-time intervals. Our algorithm approximates the target state with error polynomially small in the inverse imaginary time using a polynomial number of elementary quantum gates and a single ancilla qubit, with success probability close to one. When the initial state has reasonable overlap with the ground state, this algorithm also achieves polynomial query complexity in the system size. To our knowledge, this is the first quantum algorithm for imaginary-time evolution with provably polynomial resource scaling in evolution time. Numerical experiments validate our theoretical analysis for evolution time up to 50, demonstrating the algorithm's effectiveness for long-time evolution.
Building on this technique, we further develop imaginary-time-evolution-based algorithms for ground-state-related problems and for simulating open quantum systems. These algorithms can reduce circuit depth in certain regimes compared with existing methods, at the expense of higher total query complexity, advancing the practical feasibility of quantum simulation on early fault-tolerant devices.
\end{abstract}



\section{Introduction}

Imaginary-time evolution (ITE) provides a practical mathematical approach for analyzing complex physical systems. Propagating quantum states along sufficiently large imaginary time intervals enables the determination of ground and excited states~\cite{shi2018variational, lehtovaara2007solution}, the preparation of thermal states~\cite{motta2020determining}, and the computation of dynamical correlation functions~\cite{sakurai2022hybrid}. This concept plays a crucial role in quantum mechanics, particularly in statistical physics and quantum field theory~\cite{peskin2018introduction, lancaster2014quantum}.

Two primary computational challenges limit classical simulation of imaginary-time evolution: the Hilbert space dimension grows exponentially with particle number, and the required numerical precision increases exponentially with imaginary time. Therefore, simulating imaginary-time evolution on classical computers incurs computational costs that scale exponentially with both system size and evolution duration, severely restricting practical applications.

Quantum computing offers an alternative for preparing imaginary-time evolution states through its efficient representation of quantum many-body states. Quantum algorithms for imaginary-time evolution follow two main approaches. The first trains parameterized quantum circuits by optimizing loss functions computed from measurement outcomes~\cite{mcardle2019variational, lin2021real, mcmahon2025equating}. The second employs Trotterization to decompose the evolution into short segments, each simulated via real-time evolution algorithms~\cite{motta2020determining,gomes2020efficient, nishi2021implementation,huang2023efficient,yeter-aydeniz2020practical,yeter-aydeniz2022quantum,jouzdani2022alternative}. Numerical and experimental studies demonstrate that these approaches can approximate imaginary-time evolution with resource costs scaling polynomially in qubit number.

These existing quantum approaches remain heuristic and inadequately address the measurement precision scaling with imaginary-time duration. The total number of measurements may grow exponentially to suppress error accumulation over imaginary time. Whether quantum computing can efficiently simulate imaginary-time evolution, particularly with polynomial resource scaling in imaginary-time duration, remains theoretically unresolved, necessitating alternative quantum algorithms.

Quantum signal processing (QSP)~\cite{low2016methodology} has emerged as a fundamental framework underlying many quantum algorithms. QSP and its extensions perform polynomial transformations of input quantum data, enabling efficient data encoding and extraction. Algorithms built on QSP and generalized frameworks~\cite{gilyen2019quantum, wang2023quantum} unify and extend established quantum algorithms~\cite{martyn2021grand}, achieving rigorous complexity bounds particularly for real-time Hamiltonian evolution simulation~\cite{childs2018toward,low2019hamiltonian,martyn2023efficient}.

We address the imaginary-time scaling challenge by introducing a quantum algorithm based on the QSP framework~\cite{wang2023quantum} that prepares normalized imaginary-time evolved states with polynomial gate complexity in imaginary-time duration. Our algorithm applies polynomial approximation of $e^{\tau(x-\lambda)}$ to the system Hamiltonian and determines an adaptive normalization parameter $\lambda$ to stabilize the success probability. The success probability lower bound converges to a constant near $e^{-2}\gamma^2$, where $\gamma$ denotes the overlap between the ground state and initial system state.

Under the assumption that $\gamma$ is not exponentially small in system size $n$, our algorithm prepares the normalized imaginary-time evolved state to error $\BigTO{\poly(\tau^{-1})}$ using $\BigTO{\poly(n\tau)}$ queries to controlled-Pauli rotations and one ancilla qubit, where $\widetilde{\cO}$ suppresses Hamiltonian-dependent factors. These results establish that quantum algorithms can efficiently simulate imaginary-time evolution with resource costs polynomial in both qubit number and imaginary-time duration, providing a theoretical foundation for quantum imaginary-time evolution algorithms.

We further adapt the core idea of this approach to 
two important applications of imaginary-time evolution. We propose two algorithms: one for ground-state-related problems and one for open-system (Lindbladian) simulation. These algorithms not only fill the gap of lacking theoretical analysis for previous ITE-based algorithms~\cite{gomes2020efficient, nishi2021implementation, huang2023efficient, yeter-aydeniz2020practical, gluza2026double, zander2025role, liu2021probabilistic, kosugi2022imaginarytime, silva2023fragmented, chan2023simulating, yi2025probabilistic, kamakari2022digital} in these two applications, but also reduce circuit depth in certain regimes compared with existing methods. For ground state preparation and ground-state energy estimation, under a heuristic assumption that is attainable in practice, iterative adjustment of the evolution time $\tau$ and the normalization factor yields a query depth reduction by a factor of $\BigOmega{\gamma^{-1}}$ compared with representative Heisenberg-limited approaches~\cite{dong2022groundstate, ding2023even}, at the expense of higher total query complexity. For open-system simulation (Lindbladian simulation), our algorithm removes the polynomial dependence on the number of dissipative terms and can therefore yield shorter per-run circuit depth when the system has many local noisy channels.
Because these depth reductions target the dominant hardware bottleneck on near-term and early fault-tolerant quantum processors, our results broaden the practical reach of quantum simulation for problems in quantum chemistry, condensed matter physics, and beyond.

\section{Quantum simulation of imaginary-time evolution}~\label{sec:qite intro}

The imaginary-time Schrödinger equation $\partial_\tau \ket{\ite} = -H \ket{\ite}$ describes the imaginary-time evolution of an $n$-qubit quantum many-body system,
where $\tau$ denotes the imaginary time, $H$ is an $n$-qubit time-independent Hamiltonian, and the initial state of the system is $\ket{\phi} = \ket{\phi(0)}$.  
Quantum imaginary-time evolution {prepares} the \emph{normalized imaginary-time evolved state} (or in short, ITE state)
\begin{equation}~\label{eqn:ite def}
    \ket{\ite} = \frac{e^{-\tau H}\ket{\phi}}{\norm{e^{-\tau H}\ket{\phi}}}
\end{equation}
on a quantum {device}. The operator that maps all such $\ket{\phi}$ to $\ket{\ite}$ is called the \emph{imaginary-time evolution operator} (or in short, ITE operator), and the algorithm that prepares the state is shorthanded as the \emph{ITE algorithm}.

For large $\tau$, the ITE state converges to the lowest-energy eigenstate of $H$ within the subspace $\ketbra{\phi}{\phi}$, typically the ground state. When the initial state is maximally mixed ($I/2^n$) and $2\tau$ represents inverse temperature, the ITE state becomes the thermal state $e^{-2\tau H}/\trace{e^{-2\tau H}}$ at temperature $1/2\tau$.

Imaginary-time evolution solves the Schrödinger equation with time parameter $t$ replaced by $i\tau$, addressing problems in statistical physics and quantum field theory~\cite{peskin2018introduction, lancaster2014quantum}. Wick rotation~\cite{wick1954properties} transforms problems from Minkowski to Euclidean spacetime, converting oscillatory spacetime integrals on pseudo-Riemannian manifolds into analytically tractable forms on Riemannian manifolds. This transformation improves convergence and reveals spectral structure and stability properties of quantum systems.

\subsection{General assumptions}

Several assumptions on $H$, $\tau$ and $\ket{\phi}$ simplify our analysis without loss of generality. The Hamiltonian $H$ is assumed to be 
(\nref{assum:normalize})~\emph{normalized with negative energies}: all eigenvalues lie within the interval $[-1,1]$, and the ground-state energy $\lambda_0$ is negative.
Normalization is standard in quantum algorithms~\cite{childs2003exponential, dong2022groundstate, abrams1999quantum, peruzzo2014variational} for Hamiltonian-related problems. The negativity requirement for $\lambda_0$ can be satisfied by shifting the Hamiltonian by a multiple of identity, which does not alter the description of the ITE state.
We assume that $H$ (and any other Hamiltonian considered in this work) is provided either via its (\nref{assum:oracle})~\emph{evolution oracle}, in which error-free (control) $\exp(-iH)$ and its inverse can be queried a finite number of times, or its (\nref{assum:pauli}) \emph{Pauli form}: $H= \sum_{j} h_j \sigma_j$ is a linear combination of Pauli operators with known coefficients.
There is no assumption on the locality of $H$.

The imaginary-time evolution problem focuses on the regime that assumes
(\nref{assum:long evolution}) \emph{long evolution}: a large evolution time $\tau$, although the precise threshold for ``long'' depends on the structure of $H$. 
Such assumption is made by considering the difficulties that existing works face.

The initial state $\ket{\phi}$ is assumed to have
(\nref{assum:overlap}) \emph{non-zero overlap}: the state overlap $\gamma = \abs{\braket{\phi}{\psi_0}}$ between $\ket{\phi}$ and the ground state $\ket{\psi_0}$ is positive,
and
(\nref{assum:reprod}) \emph{reproducibility}: $\ket{\phi}$ can be accessed with finite copies.
These assumptions are common in ground-state-related problems, such as the problem of ground-state energy estimation~\cite{nielsen2010quantum, wang2023quantum, dong2022groundstate, martyn2021grand}.

More specific assumptions on the Hamiltonian and the initial state will be introduced as needed for particular problems. For clarity, we summarize all assumptions used in this work, together with the results that rely on them, in Appendix~\ref{appendix:assumption}.
{Throughout this work, \emph{query depth} denotes the maximum number of sequential oracle queries within a single circuit execution, while \emph{query complexity} denotes the expected total number of oracle queries across all circuit executions needed to obtain a successful output state (including repetitions due to post-selection, when applicable).}

\subsection{Related works}~\label{sec:ite work}

\begin{table}[t]
\setlength{\tabcolsep}{1em}
\caption{A comparison of some algorithms that prepare the normalized imaginary-time evolved state of an $n$-qubit Hamiltonian $H$ (with $L$ Pauli terms) at evolution time $\tau$, starting from an initial state with ground-state overlap $\gamma$, up to error $\BigO{\poly(\tau^{-1})}$ and overall success probability $1$. Heuristic methods are not included. Here $\BigTO{\cdot}$ omits $\log$ or $\poly\log$ factors; $k$ in Ref.~\cite{gluza2026double} denotes the number of algorithmic steps and is related to $\tau$ and $\gamma$; Theorem~\ref{thm:ite trotter} is stated under the assumption that $\gamma$ is not exponentially small in $n$. The symbol `$\backslash$' indicates that the corresponding quantity has not been fully analyzed. ``TE'' stands for Taylor expansion of the imaginary-time evolution operator; ``QSP'' for quantum signal processing; ``AN'' for adaptive normalization.}~\label{tab:ite}
\resizebox{\linewidth}{!}{
\begin{tabular}{lcccc}
\toprule
Methods & Query depth & Query complexity & Ancilla & $H$ in Pauli form?  \\
\midrule
\addlinespace
Manifold-based, $k$ steps~\cite{gluza2026double} & $\BigO{3^k}$ & $\BigO{3^k}$ & $0$ & No \\
\addlinespace
Grover-based~\cite{liu2021probabilistic} & $\BigO{\tau}$ & $\backslash$ & $\BigO{n}$ & No \\
\addlinespace
TE-based~\cite{kosugi2022imaginarytime} & $\BigO{\tau}$ & $\backslash$ & $1$ & No \\
\addlinespace
QSP-based~\cite{silva2023fragmented, chan2023simulating} & $\BigTO{\tau}$ & $\backslash$ & $1$  & No \\
\addlinespace
QSP-based with AN, Theorem~\ref{thm:ite uH} & $\BigTO{\tau}$ & $\BigTO{\gamma^{-2}\tau}$ & $1$  & No \\
\addlinespace
QSP-based with AN, Theorem~\ref{thm:ite trotter} & $\BigTO{L \poly(\tau)}$ & $\BigTO{L \poly(n\tau)}$ & $1$  & Yes \\
\addlinespace
\bottomrule
\end{tabular}
}
\end{table}

The ITE operator is a non-linear transformation since $e^{-\tau H}$ is not unitary and cannot be implemented without additional resources. Two strategies exist to simulate such operators:
\begin{enumerate}[(1),leftmargin=*]
    \item \label{enum:ite op} implement the ITE operator, with failure probability; 
    \item \label{enum:ite state} find a quantum circuit that transforms the initial state $\ket{\phi}$ to $\ket{\ite}$. This circuit executes without failure but requires reimplementation when {$\ket{\phi}$} changes. 
\end{enumerate}
Strategy~\nref{enum:ite state} is the mainstream approach  categorized into variational, Trotter  and manifold schemes. 

The variational scheme for Strategy~\nref{enum:ite state} was firstly proposed in Ref.~\cite{mcardle2019variational}, which employs McLachlan's variational principle~\cite{mclachlan1964variational} to train a parameterized quantum circuit. Gradients of circuit parameters are computed by completing a set of expectation value estimations.
As an alternative,
Ref.~\cite{lin2021real} constructs the optimization target based on oracle access to block encoding of $\exp(-\tau H/ N)$ ($N > 0$).

Two problems persist in the variational scheme. First, one must choose an ansatz whose expressivity scales with system size to cover all possible ITE states, which becomes costly in large Hilbert spaces. Second, the analysis lacks consideration of how finite measurement precision (``shot noise'') propagates through parameter updates, which may become exponentially large in $\tau$ in the worst case.

The Trotter scheme for Strategy~\nref{enum:ite state} is more common and was introduced in Ref.~\cite{motta2020determining}.
This approach finds a sequence of sliced times and unitaries $\set{(t_j, A_j)}_j$ such that $\ket{\ite} \approx (\prod_j e^{-i A_j t_j} ) \ket{\phi}$.
Variants have been developed using both variational techniques~\cite{gomes2020efficient, nishi2021implementation} and randomized approaches~\cite{huang2023efficient} to reduce the overall gate and measurement cost. This method has been applied to compute ground- and excited-state energies~\cite{yeter-aydeniz2020practical} as well as finite-temperature static and dynamical properties of one-dimensional spin systems~\cite{sun2021quantum}.

Trotter-based approaches guarantee stepwise convergence, yet {computation of} $A_j$ relies on heuristic measures. Each $A_j$ is obtained by solving a linear system $S\vec{a} = c^{-1/2}\vec{b}$, where $c$ and each element of $S, \vec{b}$ requires {estimation} of expectation values subject to shot noise.
Ref.~\cite{huang2023efficient} establishes measurement lower bounds in terms of normalization factor $c$, vector norm $\norm{\vec{b}}$, and matrix condition number $\norm{S^{-1}}$. 
However, the cumulative measurement cost remains not characterized since these parameters depend on $t_j$. This is problematic as $c$ would decay exponentially with $t_j$, potentially requiring exponentially many measurements for large evolution times.

The manifold scheme represents a recently developed approach for Strategy~\nref{enum:ite state}, first introduced in Ref.~\cite{gluza2026double}. This scheme treats the preparation of ITE states as a minimization problem of a cost function on a Riemannian manifold, providing stronger theoretical guarantees than the variational and Trotter schemes, as analyzed in Refs.~\cite{gluza2026double, zander2025role, suzuki2025doublebracket, mcmahon2025equating}. Nevertheless, for long evolution times, these methods either require large circuit depths or encounter the same limitations as the variational scheme.

Strategy~\nref{enum:ite op} includes quantum approaches~\cite{liu2021probabilistic, kosugi2022imaginarytime, silva2023fragmented, chan2023simulating, yi2025probabilistic} that prepare the ITE operator with partial theoretical guarantees. These approaches directly implement the (fragmented) ITE operator using sequences of large quantum gates interleaved with post-selections on ancilla qubits. The success probability of the post-selections, and thus the overall algorithm, decays exponentially with increasing $\tau$. 
Our work follows Strategy~\nref{enum:ite op} and solves this decay problem{. Existing algorithms that contains theorectical analysis on the imaginary time evolution problem are summarized in Table~\ref{tab:ite}.}

\subsection{Preparation of ITE state}~\label{sec:constant prob}

Quantum signal processing (QSP) was firstly introduced in Ref.~\cite{low2016methodology}, which showed that interleaving single-qubit rotation gates enables polynomial transformations of a scalar input $x$. 
Subsequent generalizations have extended QSP to multi-qubit frameworks~\cite{gilyen2019quantum, wang2023quantum, motlagh2024generalized, odake2023universal, sunderhauf2023generalized}, allowing transformations of input matrices embedded within quantum gates.
Given the Hamiltonian eigenvalues normalized to the interval $[-1,1]$, the ITE operator is an exponential transformation of the unitary $U_H=\exp(-iH)$ via a naively chosen target function $f(x)=e^{\tau x}/e^\tau$. Imaginary-time evolution can therefore be implemented by QSP-based frameworks. 

Quantum phase processing (QPP)~\cite{wang2023quantum} is {a multi-qubit QSP framework specialized for unitary transformations}. QPP enables polynomial transformations of an $n$-qubit input unitary acting on the quantum state by tuning rotation angles on a single ancilla qubit. Specifically, for a trigonometric polynomial $G \in \CC[e^{ix}, e^{-ix}]$ approximating the target function $g$ with error $\epsilon$, a quantum circuit denoted by $\UQPPs{g}{\epsilon}(U_H)$ can call the controlled input unitary $U_H = \sum_j e^{-i\lambda_j} \ketbra{\psi_j}{\psi_j}$ and its inverse $\deg(G)$ times, to implement the  transformation
\begin{equation}
   \UQPPs{g}{\epsilon}(U_H) = \begin{bNiceMatrix}
       G(U_H) & \ldots \\
       \ldots & \ldots
   \end{bNiceMatrix}
,\end{equation}
where $G(U_H) \coloneqq \sum_j G(-\lambda_j)\ketbra{\psi_j}{\psi_j}$. 
A detailed construction of such circuits is provided in Appendix~\ref{appendix:QPP}.
Post-selecting the ancilla qubit of the  circuit in the zero state yields an output state 
\begin{equation}
    \ket{\widetilde{\phi}(\tau)}
    \approx g(U_H)\ket{\phi}/\norm{g(U_H)\ket{\phi}} = \ket{\ite}
.\end{equation}
Other QSP frameworks such as quantum singular value transformation~\cite{gilyen2019quantum} would achieve similar performance, but the encoding model switches from $U_H$ to a block encoding of $H$, which is less natural and practical.

However, this naive choice of the target function $f$ leads to an exponential decay of the success probability, $\norm{f(U_H)\ket{\phi}}^2 = \BigO{e^{-2\tau}}$, as the imaginary time $\tau$ increases.
Ref.~\cite{silva2023fragmented, chan2023simulating} faced this difficulty. They employed a fragmented approach (simulating $\exp(\tau H / N)$) to mitigate the effect, yet the resource cost is not shown to be efficient.
One option is to simulate the normalized function $e^{\tau x} / e^{\tau \lambda_0}$ on the spectrum of $H$~\cite{gilyen2019quantum, an2023linear} to avoid unnecessary resource cost. Using a lower bound on $\lambda_0$ with additive error $\BigO{\tau^{-1}}$ controls the excess normalization, but obtaining such an estimate adaptively while maintaining polynomial resource scaling requires additional algorithmic design and analysis.

\subsection{Algorithm with polynomial resources}~\label{sec:ite algorithm}

Our approach addresses this problem by introducing an adaptive normalization factor $\lambda\in (0, 1]$ into the target function that stabilizes the success probability. We consider a modified function defined as
\begin{equation}\label{eqn:f def}
    \expf(x) = \begin{cases}
        {\alpha} e^{\tau(x - \lambda)}, & x \in [-1, \lambda];\\
        \expfr(x), & x \in (\lambda,1],
    \end{cases}
\end{equation}
where {$\alpha \in (e^{-1/2}, 1]$ and} $\expfr:(\lambda,1]\to \setcond{x\in\CC}{\abs{x} \leq 1}$ ensure the Fourier approximation error $\epsilon$ decays super-polynomially as the approximation degree  increases. One choice for such $\alpha$ and $\expfr$ is discussed in Appendix~\ref{appendix:QPP}. Within this construction, we obtain the following lemma. 

\begin{lemma}~\label{lem:succ prob low bound}
    Let $C \geq \tau(\lambda - \abs{\lambda_0}) \geq 0$. Under Assumptions~(\nref{assum:normalize},\nref{assum:long evolution},\nref{assum:overlap}), the output state $\ket{\widetilde{\phi}(\tau)}$ from the ITE circuit $\UQPPs{\expf}{\epsilon}(U_H)$ is obtained with success probability lower bounded by $\alpha^2 \gamma^2 e^{-2C} - \epsilon$. Moreover, the state fidelity between the output state and the ITE state is approximately lower bounded as
\begin{equation}
    \abs{\braket{\ite}{\widetilde{\phi}(\tau)}} \gtrsim 1 - \BigO{ \alpha^{-1}\epsilon \cdot e^{C}}
.\end{equation}
\end{lemma}

When $C$ in Lemma~\ref{lem:succ prob low bound} is upper bounded by {$1$}, the success probability is lower bounded by $\alpha^2 e^{-2}\gamma^2 - \epsilon$, while maintaining the fidelity of order $1 - \BigO{ \alpha^{-1} \epsilon}$. Since the magnitude of $\alpha$ is lower bounded by $e^{-1/2} > 0.6$ and $\gamma$ is fixed for the problem, the success probability bound is approximately constant, thereby solving the exponential decay problem in previous work~\cite{liu2021probabilistic, kosugi2022imaginarytime, silva2023fragmented, chan2023simulating}. 
Note that $C = 0$ corresponds to the idea discussed in Ref.~\cite{gilyen2019quantum}.

\begin{figure}[t]
    \centering
    \includegraphics[width=0.5\linewidth]{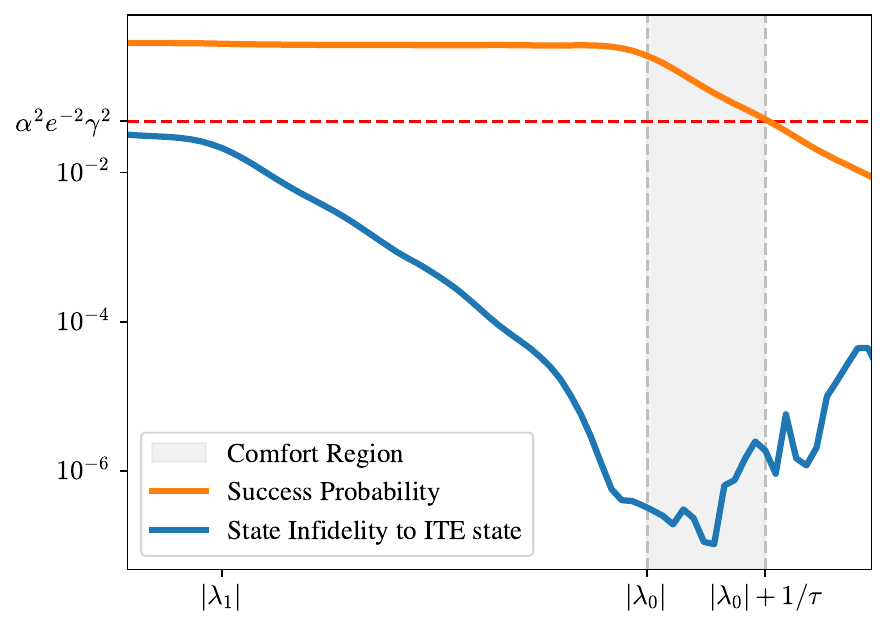}
    \caption{The performance of the ITE circuit $\UQPPs{\expf}{\epsilon}(U_H)$  with different choices of $\lambda$ (horizontal axis) and $\tau = 20$. 
    The blue line shows the state infidelity between the ITE state $\ket{\ite}$ and the output state $\ket{\widetilde{\phi}(\tau)}$. The orange line shows the success probability of obtaining the output state. The vertical axis is scaled by a logarithm of 10 for better visibility.}
    \label{fig:lambda impact}
\end{figure}

A numerical experiment demonstrates the effectiveness of Lemma~\ref{lem:succ prob low bound}. We consider the experimental setting $\tau = 20$, $\alpha = 0.85$, $\gamma^2 = 0.5$, $\epsilon = \BigO{10^{-5}}$ and $H$ is a normalized Heisenberg Hamiltonian given by Equation~\eqref{eqn:hamiltonian}. 
Figure~\ref{fig:lambda impact} shows the experiment results with the vertical axis in logarithmic scale.
When $\lambda$ is far from $\abs{\lambda_0}$ and moving towards it, the success probability (orange line) of obtaining $\ket{\widetilde{\phi}(\tau)}$ remains stable, while the state infidelity (blue line) between $\ket{\widetilde{\phi}(\tau)}$ and $\ket{\ite}$ is large. This infidelity behavior is expected, as the transformation on the ground state subspace is described by $\expfr$ instead of the exponential function. 
When $\lambda$ increases beyond $\abs{\lambda_0}$, Lemma~\ref{lem:succ prob low bound} applies: the state infidelity becomes minimal and decreases to $\epsilon$ within the machine error of the classical simulator, while the success probability decreases exponentially, illustrating the exponential decay problem introduced earlier. 
When $\lambda$ lies within $[\abs{\lambda_0}, \abs{\lambda_0}+\tau^{-1}]$ (the ``comfort region'' in Figure~\ref{fig:lambda impact}), 
one obtains the output state with high fidelity, while the success probability is lower bounded by $\alpha^2 e^{-2} \gamma^2$.

The identification of $\lambda$ can be achieved via existing ground-state energy estimation algorithms~\cite{dong2022groundstate, wang2023quantum,wang2023ground, ding2024quantum, lin2022heisenberg, ding2023even} with precision $\tau^{-1}/2$ (in case the obtained $\lambda$ is smaller than $\abs{\lambda_0}$), given access to controlled-$U_H$ and its inverse.
For example, Ref.~\cite{wang2023quantum, dong2022groundstate} provide the desired estimation using 1 ancilla qubit and $\BigTO{\gamma^{-2}\tau}$ queries of $U_H$ and its inverse.
 Here $\BigTO{\cdot}$ omits the $\log$ factors of $\gamma$ and $\tau$.
Taking the function approximation error to be $\epsilon = \BigO{\poly(\tau^{-1})}$, we obtain a quantum algorithm that prepares the ITE state with polynomial resources in time, as stated in the following theorem.

\begin{theorem}~\label{thm:ite uH}
    Under Assumptions~(\nref{assum:normalize},\nref{assum:oracle},\nref{assum:long evolution},\nref{assum:overlap},\nref{assum:reprod}), one can prepare the ITE state $\ket{\ite}$ up to fidelity $1 - \BigO{\poly(\tau^{-1})}$, with probability 1, using the following cost:
\begin{enumerate}[leftmargin=1em]
    \item [-] $\BigTO{ \gamma^{-2} \tau}$ queries to controlled-$U_H$ and its inverse,
    \item [-] $\BigO{\gamma^{-2}}$ copies of $\ket{\phi}$, 
    \item [-] $\BigTO{\tau}$ maximal query depth of $U_H$, and
    \item [-] one ancilla qubit initialized in the zero state.
\end{enumerate}
\end{theorem}

Our algorithm in Theorem~\ref{thm:ite uH} does not directly depend on the system size $n$. This stems from the fact that circuits using QPP can process unitary eigenphases simultaneously, such that the query complexity of performing the unitary transformation is independent of system size. 
When we further assume a
(\nref{assum:good overlap}) \emph{good overlap}:  $\gamma = \BigOmega{\poly(n^{-1})}$,
which  quantum algorithms assume to establish their advantages for Hamiltonian-related problems~\cite{dong2022groundstate, wang2023quantum, ding2024quantum, lin2022heisenberg, ding2023even},
the resource complexity then scales polynomially with the system size $n$. 
Note that this assumption may not hold in general. When the overlap with the true ground state is exponentially small, the analysis below applies to the lowest-energy eigenstate having non-trivial overlap with $\ket{\phi}$, which suffices for preparing the ITE state.

Apart from the ideal case, we  analyze the resource complexity when $U_H$ is not directly accessible. In this case, we consider a Trotter decomposition $U_{H^\approx}$ that approximates $U_H$. Without additional assumptions, there is no theoretical guarantee that the ground-state subspace of $\exp(-\tau H^\approx)$ matches the ground-state subspace of $\exp(-\tau H)$. Therefore, to guarantee that these two subspaces match under the effect of $\tau$, we require $H$ to be
(\nref{assum:nondegenerate}) \emph{non-degenerate}: the energy spectral gap $\Delta = \lambda_1 - \lambda_0$ between the first-excited-state energy $\lambda_1$ and the ground-state energy is non-zero,
and $\tau$ to be large enough to ensure a
(\nref{assum:long evolution gap}) \emph{distinguishable gap}:  $\Delta = \BigOmega{\tau^{-1} \log \poly(\tau)}$.
Then we can show that quantum resources remain polynomially dependent on $\tau$ and $n$, demonstrating the robustness of our algorithm.

\begin{theorem}~\label{thm:ite trotter}
    Under Assumptions~(\nref{assum:normalize},\nref{assum:pauli},\nref{assum:long evolution},\nref{assum:overlap},\nref{assum:reprod},\nref{assum:good overlap},\nref{assum:nondegenerate},\nref{assum:long evolution gap}), one can prepare the ITE state $\ket{\ite}$ up to fidelity $1 - \BigO{L^2\Lambda^2\poly(\tau^{-1})}$, using the following cost:
\begin{enumerate}[leftmargin=1em]
    \item [-] $\BigTO{L \poly(n\tau)}$ queries to controlled Pauli rotations,
    \item [-] $\BigO{\poly(n)}$ copies of $\ket{\phi}$, 
    \item [-] $\BigTO{L \poly(\tau)}$ maximal query depth, and
    \item [-] one ancilla qubit initialized in the zero state,
\end{enumerate}
    where $L$ is the number of Pauli terms and $\Lambda = \max_j\abs{h_j}$.
\end{theorem}

To the best of our knowledge, this is the first quantum imaginary-time evolution algorithm that theoretically achieves polynomial scaling of resource complexity with respect to the imaginary-time duration $\tau$, while maintaining precision of $\BigO{\poly(\tau^{-1})}$.
Compared with existing works that are either heuristic or theoretically infeasible for large $\tau$, our algorithm is efficient in terms of evolution time, and hence can be considered as an advance on the problem of imaginary-time evolution.

We also extend our discussion to imaginary-time evolutions over short time intervals, which is practically relevant when one wishes to implement many short ITE fragments interleaved with other operations, as in the setting of Section~\ref{sec:lindbladian}. This analysis is non-trivial, because existing long-time guarantees cannot be applied directly. In this regime, we show that, provided one has \emph{a priori} spectral information obtained from phase-estimation-type algorithms at some longer evolution time, it is possible to design practical algorithms for short imaginary-time evolution with controlled success probability and error bounds. For notational simplicity, we write the \emph{ITE circuit} with post-selected ancilla as
\begin{equation}~\label{eqn:qpp post-select}
    \UITE{\tau}{H}[\ket{\varphi}] \coloneqq \frac{(\bra{0} \ox I_n)\,\UQPPs{f_{\tau, \lambda}}{\epsilon}(U_H)\,(\ket{0} \ox \ket{\varphi})}{\norm{(\bra{0} \ox I_n)\,\UQPPs{f_{\tau, \lambda}}{\epsilon}(U_H)\,(\ket{0} \ox \ket{\varphi})}}
.\end{equation}

\begin{corollary}~\label{coro:ite short}
    Let $N$ be a positive integer such that $t = N\tau$ satisfies Assumption~(\nref{assum:long evolution}). Suppose $\lambda \in [\abs{\lambda_0}, \abs{\lambda_0} + t^{-1}]$. Under Assumptions~(\nref{assum:normalize},\nref{assum:overlap}),
    $\UITE{\tau}{H}$ in Equation~\eqref{eqn:qpp post-select} satisfies
\begin{equation}
    \abs{ \bra{\ite}\UITE{\tau}{H} [\ket{\phi}]} \gtrsim 1 - \BigO{\alpha^{-1} \epsilon}
.\end{equation}
    Moreover, under Assumption~(\nref{assum:oracle}), $\UITE{\tau}{H}$ can be implemented with success probability lower bounded by $\alpha^2 \left(\gamma^2 e^{-2/N} + (1 - \gamma^2)e^{-2(2t + 1)/N}\right) - \epsilon$.
\end{corollary}

Corollary~\ref{coro:ite short} theoretically ensures the stability of performing short imaginary-time evolutions. Compared with Lemma~\ref{lem:succ prob low bound}, the error bound remains of the same order, whereas the lower bound on the success probability is improved by an additional term. This reflects the fact that, for short imaginary times, contributions from excited-state subspaces are more significant than in the long-time regime. Indeed, when $N \gg t$, the two exponential factors $e^{-2/N}$ and $e^{-2(2t + 1)/N}$ both approach $1$, so the success probability is approximately lower bounded by $\alpha^2 \geq e^{-1}$. By composing $N$ such short-time evolutions, the overall success probability is again dominated by the ground-state subspace, recovering the behavior captured by Lemma~\ref{lem:succ prob low bound}.

Proofs of theorems in this section are deferred to Appendix~\ref{appendix:ite}. In the next section, we show how to apply the algorithm or idea of preparing ITE states to practical applications. 

\section{Applications}

Several works have proposed or implicitly used the idea of simulating imaginary-time evolution with digital quantum circuits as a subroutine for other tasks, including the study of ground- and excited-state properties of closed-system Hamiltonians~\cite{motta2020determining, yeter-aydeniz2020practical, sun2021quantum, wang2025computing}, open-system dynamics~\cite{kamakari2022digital}, differential equations~\cite{an2023linear}, data classification~\cite{ye2025quantum}, combinatorial optimization~\cite{wang2025imaginary, alam2023solving, bauer2024combinatorial}, etc. However, some directions have so far lacked provably efficient algorithms for practically relevant evolution times. As, to our knowledge, the first algorithm with theorectical polynomial-time guarantees for imaginary-time evolution, our method can inspire a more systematic development of these applications. To illustrate this, we apply our ITE algorithm in two applications and show that the proposed algorithms are comparable to existing approaches.

\subsection{Ground state preparation and ground-state energy estimation}

As one of the central applications of imaginary-time evolution,
ground state $\ket{\psi_0}$ preparation and ground-state energy $\lambda_0$ estimation are fundamental tasks for demonstrating quantum computational advantage.
The normalized imaginary time evolved state $\ket{\ite}$ provides a systematic approach to both problems. The amplitude of $\ket{\ite}$ on the ground-state subspace converges exponentially to 1 as $\tau$ increases, causing the expectation value $\hat{E}(\tau) = \braandket{\ite}H{\ite}$ to converge exponentially to the ground-state energy. The following lemma quantifies this convergence.

\begin{lemma}~\label{lem:overlap lower bound}
    Under Assumptions~(\nref{assum:overlap},\nref{assum:nondegenerate}), 
\begin{equation}
    \abs{\braket{\psi_0}{\ite}} \geq {\gamma}/{\sqrt{e^{-2\tau \Delta}(1 - \gamma^2) + \gamma^2}}
.\end{equation}
    Moreover, the lower bound is tight for some Hamiltonians.
\end{lemma}

When the initial state is fixed, the energy spectral gap $\Delta$ governs the convergence rate in Lemma~\ref{lem:overlap lower bound}, causing the required evolution time $\tau$ to vary significantly across different Hamiltonians. For example, the 2-qubit Hamiltonian describing the $\textrm{H}_2$ molecule achieves near-precise ground-state energy estimation with $\tau = 3$~\cite{mcardle2019variational}. In contrast, the 5-qubit Heisenberg Hamiltonian given by 
Equation~\eqref{eqn:hamiltonian} requires $\tau \geq 20$ for comparable accuracy, as shown in Figure~\ref{fig:ground}(a). Then a natural question arises: without a priori knowledge of $\Delta$, how to determine $\tau$ for ground-state-related problems?

Assumption~(\nref{assum:long evolution gap}) provides a sufficient condition on $\tau$ to guarantee adequate approximation, in which case $e^{\tau\Delta} = \BigOmega{\poly(\tau)}$. This leads to the following formal problem statement for applying imaginary-time evolution to ground state preparation and energy estimation.

\begin{problem}~\label{prob:ground}
    Under Assumptions~(\nref{assum:normalize},\nref{assum:oracle},\nref{assum:overlap},\nref{assum:reprod},\nref{assum:good overlap},\nref{assum:nondegenerate}), the goal is to obtain
\begin{enumerate}[1., leftmargin=1em]
    \item $\tau$ satisfying Assumption~(\nref{assum:long evolution gap}) \emph{(ITE state $\approx$ ground state)};
    \item $\lambda \in [\abs{\lambda_0}, \abs{\lambda_0} + \tau^{-1}]$, which enables efficient preparation of $\ket{\ite}$ via Theorem~\ref{thm:ite uH} with near-constant success probability \emph{(efficient ITE state preparation)};
    \item $E$ as an estimate of $\hat{E}(\tau)$ \emph{(ground-state energy estimation)}.
\end{enumerate}
\end{problem}

The last two tasks in Problem~\ref{prob:ground} depend on successfully locating an appropriate $\tau$. A straightforward approach would solve these tasks sequentially: apply Theorem~\ref{thm:ite uH} to prepare $\ket{\ite}$ for increasing values of $\tau$, monitor the convergence of $\hat{E}(\tau)$, and then determine $\lambda$ and estimate $\hat{E}(\tau)$. This approach becomes computationally expensive because each variation of $\tau$ requires invoking the entire algorithm, including the ground state estimation subroutine to find an appropriate $\lambda \in [\abs{\lambda_0}, \abs{\lambda_0} + \tau^{-1}]$, just to obtain sufficient samples for estimating $\hat{E}(\tau)$.

As a better alternative, we propose a unified approach that accomplishes all three tasks in Problem~\ref{prob:ground} simultaneously. The key insight is to introduce the expectation value of the Hamiltonian evolved under the unnormalized state $\expft(U_H)\ket{\phi}$:
\begin{equation}\label{eqn:ideal vqe}
    \loss{\lambda} = \bra{\phi}\expft(U_H)^\dag\,H\,\expft(U_H)\ket{\phi}
,\end{equation}
where $t > 0$ is a trial value that serves as a candidate for the evolution time $\tau$.
This quantity serves dual purposes: it detects the proximity of $\lambda$ to $\abs{\lambda_0}$ and encodes information about $\hat{E}(t)$. The rate of change of $\loss{\lambda}$ exhibits a sharp transition at the critical point $\abs{\lambda_0}$. When $\lambda$ decreases from $\abs{\lambda_0}+\tau^{-1}$ to $\abs{\lambda_0}$, the relative change in expectation value is
\begin{equation}~\label{eqn:relative}
\begin{aligned}
    r 
    &= {\left(\loss{\abs{\lambda_0}+\tau^{-1}} - \loss{\abs{\lambda_0}}\right)}/{ \loss{\abs{\lambda_0}} } \\
    &= {\left(e^2 \hat{E}(t) - \hat{E}(t)\right)}/{ \hat{E}(t) }
    =e^2 - 1
,\end{aligned}
\end{equation}
whereas further decreasing $\lambda$ slightly below $\abs{\lambda_0}$ yields negligible relative change $r \to 0$. Furthermore, estimation of $\hat{E}(t)$ emerges naturally during the evaluation of $\loss{\lambda}$. The measurement of observable $\hat{H}=\ketbra{0}{0}\otimes H$ on the output state of the QPP circuit $\UQPPs{\expft}{\epsilon}(U_H)$ yields an estimate $\measureloss{\lambda}$ of
\begin{equation}
     \bra{0,\phi}\UQPPs{\expft}{\epsilon}(U_H)^\dag\,\hat{H}\,\UQPPs{\expft}{\epsilon}(U_H)\ket{0,\phi}
.\end{equation}
Each measurement shot where $\lambda \geq \abs{\lambda_0}$ and the ancilla qubit yields 0 simultaneously contributes to the estimation of $\hat{E}(t)$. Thus, the evaluation of relative changes naturally accumulates information about $\hat{E}(\tau)$.

Our algorithm leverages these properties through an adaptive ternary search strategy. Starting with an initial guess $t$, we iteratively narrow down an interval $(\lambda_l, \lambda_r)$ containing $\abs{\lambda_0}$, where $\lambda_r \leq 1$ is determined via a logarithmic-time search algorithm. The relative change $r$ from Equation~\eqref{eqn:relative} serves as an indicator to reduce the search interval by a factor of $2/3$ per iteration. At iteration $i$, we obtain an estimate $E_i$ of $\hat{E}(t)$, where $t$ increases by a linear increment $\Delta t$. The algorithm terminates when both the interval length falls below $t^{-1}$ and the sequence $\set{E_i}_i$ satisfies a convergence criterion, thereby completing all tasks in Problem~\ref{prob:ground}. Algorithm~\ref{alg:ground} gives a sketch of this procedure.

\begin{algorithm}[t]\label{alg:ground}
\caption{Ground state preparation and energy estimation via ITE}
\SetKwInOut{Input}{Input}
\SetKwInOut{Output}{Output}
\Input{Hamiltonian $H$, initial state $\ket{\phi}$, step size $\Delta t$, lower bound $\explb$, a boolean function $\cX$ for testing convergence}
\Output{$\tau$, $\lambda$, $E$ in Problem~\ref{prob:ground}}

Guess $t \gg 0$\;

$E_0 \gets 0$, $i \gets 0$\;

$\lambda_l \gets 0$, $\lambda_r \mathrel{\overset{\approx}{\gets}}$ $\max\setcond{\lambda}{\abs{\measureloss{\lambda}} > B}$\;

\While{${\lambda_r - \lambda_l} > t^{-1}$ or $\cX(\set{E_i}_i) = \operatorname{False}$}{
    Measurement shots \# $\gets 8 L \Lambda^2 t^3 \explb^{-2}$\;

    $\delta \gets (\lambda_r - \lambda_l)/3$, $\lambda_{lm} \gets \lambda_l + \delta$, $\lambda_{rm} \gets \lambda_r - \delta$\;
    
    Estimate $\measureloss{\lambda_{lm}}$, $\measureloss{\lambda_r}$\;
    
    $\reldiff \gets \left({\measureloss{\lambda_{lm}} - \measureloss{\lambda_r}}\right)/{{\measureloss{\lambda_r}}}$ \;

    \uIf{$\abs{\reldiff-(e^{4\tau\delta}-1)}>\tau^{-1}(e^{4\tau\delta}+1)$}{
        $E_i \gets $ selected samples that estimate $\measureloss{\lambda_r}$\;
        $[\lambda_l, \lambda_r] \gets [\lambda_{lm}, \lambda_r]$\;
    }\Else{
        $E_i \gets $ selected samples that estimate  $\measureloss{\lambda_{lm}}, \measureloss{\lambda_r}$\;
        $[\lambda_l, \lambda_r] \gets [\lambda_l, \lambda_{rm}]$\;
    }
    $t \gets t + \Delta t$, $i \gets i + 1$\;
}
\Return $\tau\gets t$, $\lambda_r$, $E_{i}$\;
\end{algorithm}

To establish a theoretical analysis of Algorithm~\ref{alg:ground}, we need to bound the measurement shot number to analyze the relative change error. This analysis requires a
(\nref{assum:priori}) \emph{priori knowledge} of a quantity $B$ such that $\gamma^2\abs{\lambda_0} \geq e^2 B > 0$.
{Analogous spectral prior information is also assumed in other ground-state algorithms; for instance, Ref.~\cite{dong2022groundstate} requires access to an estimate that lies inside the spectral gap between $\lambda_0$ and $\lambda_1$.}
Although this assumption is a heuristic step, given a reasonable state overlap under Assumption~(\nref{assum:good overlap}), such knowledge is achievable by guessing a small but practically acceptable number, as we verify numerically in the next section. We establish theoretical guarantees for the estimation precision and resource requirements as follows:

\begin{theorem}~\label{thm:ground}
Suppose Assumptions~(\nref{assum:normalize},\nref{assum:oracle},\nref{assum:overlap},\nref{assum:reprod},\nref{assum:good overlap},\nref{assum:nondegenerate},\nref{assum:priori}) hold, and $H = \sum_j h_j \sigma_j$ is given in Pauli form with known coefficients.
Algorithm~\ref{alg:ground} returns a time $\tau$ that satisfies Assumption~(\nref{assum:long evolution gap}), an estimate $\lambda \in [\abs{\lambda_0}, \abs{\lambda_0}+\tau^{-1}]$, and an estimate of $\lambda_0$ within precision $\BigO{B \gamma^{-1} \tau^{-1}}$, with failure probability $\BigO{e^{-\tau}\log\tau}$.
Moreover, there are at most $\BigO{L\log\tau}$ distinct circuit constructed in Algorithm~\ref{alg:ground}, and each circuit takes at most:
\begin{enumerate}[leftmargin=1em]
    \item [-] $\BigTO{\tau}$ queries to controlled-$U_H$ and its inverse, 
    \item [-] $\BigTO{\tau}$ query depth of $U_H$, 
    \item [-] 1 ancilla qubit, and
    \item [-] $\BigO{L \Lambda^2 \explb^{-2} \tau^3 }$ measurement shots,
\end{enumerate}
    where $L$ is the number of Pauli terms in the decomposition $H = \sum_j h_j \sigma_j$, $\Lambda = \max_j\abs{h_j}$, and $B > 0$ is the prior lower bound from Assumption~(\nref{assum:priori}) satisfying $\gamma^2\abs{\lambda_0} \geq e^2 B$.
    Equivalently, in terms of the spectral gap $\Delta = \lambda_1 - \lambda_0$, the query depth scales as $\BigO{\Delta^{-1}\poly\log(\Delta^{-1})}$ and the total query complexity scales as $\BigO{\gamma^{-4}\Delta^{-4}\poly\log(\Delta^{-1})}$.
\end{theorem}

Our method achieves polynomial rather than exponential measurement scaling with respect to $t$. Finite sampling limits the precision of $\loss{\lambda}$ estimation and requires quadratically increasing measurements for higher precision. However, detecting the sharp relative change in $\loss{\lambda}$ requires only polynomially many samples in $t$. This detection becomes more reliable at larger $t$ due to the exponential amplification of energy differences.

The proof of Theorem~\ref{thm:ground} is done by analyzing the worst-case convergence of the adaptive ternary search and applies statistical guarantees for expectation-value estimation. Appendix~\ref{appendix:lambda location} presents the detailed derivation and supporting propositions.

\subsubsection{Numerical simulations}~\label{sec:experiment}

We performed two numerical experiments to validate the theoretical predictions of our algorithms for imaginary-time evolved state preparation and ground-state energy estimation. Both experiments utilized an \emph{antiferromagnetic Heisenberg Hamiltonian} on an $n$-qubit homogeneous linear chain~\cite{vandiepen2021quantum}, given as
\begin{equation}~\label{eqn:hamiltonian}
    H \propto \frac{1}{n} \sum_{j=1}^{n - 1} (X_j X_{j+1} + Y_j Y_{j+1} + Z_j Z_{j+1} - I)
,\end{equation}
where $X_j$, $Y_j$, and $Z_j$ denote Pauli matrices acting on the $j$-th qubit,
 $H$ is normalized by dividing the absolute sum of its Pauli coefficients. 

 \begin{figure}[t]
    \centering
    \includegraphics[width=\linewidth]{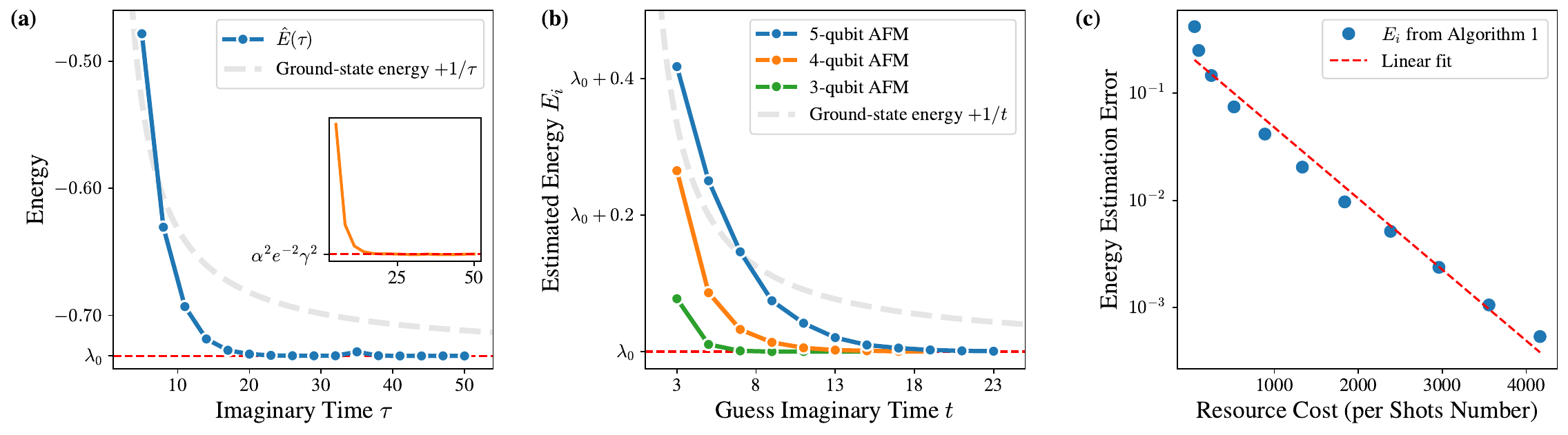}
    \caption{Experiment results for {applying our ITE-based algorithms to ground-state problems}
    for antiferromagnetic Heisenberg (AFM) chains.
    \textbf{(a)} The expectation value of output state with respect to the Hamiltonian as $\tau$ increases. The inset plot shows success probability of obtaining the imaginary-evolution state, with the red dashed line as the theoretical lower bound. The inset plot and main plot share the same x-axis label. \textbf{(b)} The list of estimated energy  recorded in numerical simulations of Algorithm~\ref{alg:ground} for 3-, 4-, and 5-qubit AFM instances. \textbf{(c)} The logarithm of the difference between the measured energy and the ground state energy for the 5-qubit AFM instance, plotted against accumulated resource consumption, where each circuit uses the same number of shots.}
    \label{fig:ground}
\end{figure}

This Hamiltonian is chosen for its computational and physical meanings in the quantum many-body system.
The Heisenberg Hamiltonian provides a prototypical setting for benchmarking quantum algorithms, as its ground state is highly entangled and because it is intimately related to quantum phase transitions~\cite{schaffer2012quantum}.
Although the 1D Heisenberg model admits exact Bethe-ansatz solutions~\cite{yang1966onedimensionalI,yang1966onedimensionalII}, this chain remains a useful benchmark for testing the algorithm performance.
Here we choose the initial state as the computational state with the smallest nonzero overlap with the ground state of $H$.
This initial state is for demonstration purpose, since its ground-state overlap is decreased exponentially with the system size~\cite{brockmann2017universal}. How to find a good initial state for this task is beyond the scope of this work.

\emph{ITE state convergence toward the ground state.---}
The first experiment assessed the efficacy of Theorem~\ref{thm:ite trotter} for preparing the ITE state $\ket{\ite}$ for a 5-qubit chain as the imaginary time $\tau$ increases from 10 to 50. For each $\tau$, we selected the normalization factor $\lambda$ within $1/\tau$ of the exact ground-state energy and set the parameter $\alpha = 0.85$ in our polynomial approximation. Although our theoretical analysis employs a Trotterized approximation for $U_H=e^{-iH}$, here we directly implemented $U_H$ as an oracle to highlight the approximation error arising from polynomial fitting.

Figure~\ref{fig:ground}\textbf{(a)} illustrates the results. The energy expectation value $\hat{E}(\tau)=\bra{\ite}H\ket{\ite}$ converged smoothly toward the exact ground-state energy (shown by the red dashed line), confirming our theoretical expectation in Lemma~\ref{lem:overlap lower bound}. The inset plot shows success probabilities for obtaining the ITE state via post-selection, consistently exceeding the theoretical lower bound in Lemma~\ref{lem:succ prob low bound}. Minor deviations observed were caused by numerical approximation errors.

\emph{Ground-state energy estimation.---}
The second experiment evaluated the performance of Algorithm~\ref{alg:ground} for ground-state energy estimation for three linear chains with $n = 3, 4, 5$ qubits. For all three system sizes, the algorithm began with an initial guess at an imaginary time $t = 3$ and incremented $t$ by $\Delta t = 2$ at each iteration. The convergence criterion $\mathcal{X}$ tested whether two consecutive energy estimates fell within the most recent binomial proportion confidence intervals~\cite{agresti1998approximate}.

To demonstrate practical feasibility, we set $B = 1/5000$ and the measurement shot number used at each iteration was fixed at $10^9$, rather than increased with guessed $t$ nor $B$. The choices of $B$ and the fixed shot count reflect a scenario where minimal prior information on detailed spectral properties of the Hamiltonian is available, yet computational resources remain manageable.

{As shown in Figure~\ref{fig:ground}\textbf{(b)}, as the guessed time $t$ increases linearly, the estimated energies for all three system sizes steadily approach the exact ground-state energy, with different convergence rates. All three curves converge significantly faster than the heuristic $1/\tau$ scaling illustrated by the grey line $\lambda_0 + \tau^{-1}$. This also indicates that the imaginary time required for convergence grows with system size: the estimates become visibly close to the exact energy once $t$ reaches approximately $15$, $19$, and $23$ for $n = 3, 4$, and $5$, respectively.}

To quantify the estimation efficiency more evidently, Figure~\ref{fig:ground}\textbf{(c)} plots the logarithmic difference between the estimated and exact ground-state energies against the cumulative number of oracle queries (query complexity) for $n = 5$. A linear regression closely aligned with data points indicated exponential convergence of energy estimation accuracy with increasing computational resources. Both figures demonstrate our algorithm's efficiency on Problem~\ref{prob:ground}.

We explicitly note the observed exponential convergence of estimation error does not violate the Heisenberg limit, which constrains precision scaling only in the regime of extremely high accuracy. Here, the achieved precision remained above this regime, and thus the observed scaling predominantly reflected the exponential convergence intrinsic to imaginary-time evolution. If higher precision is required beyond our demonstrated range, one would then encounter scaling limited by statistical measurement fluctuations, governed by Hoeffding's inequality, resulting in a square-root dependence on the number of measurement shots.

Also, the numerical Fourier approximation we use is weaker than the theoretical one argued in Appendix~\ref{appendix:QPP}. This is due to the fact that the theoretical analysis in this work does not consider the machine precision, the type of classical error that is commonly ignored in the analysis of quantum algorithms but become increasing effective in this task as $\tau$ increases.
Under this effect, the numerical Fourier approximation achieves polynomial rather than super-polynomial decay, which does not affect our claim on the polynomial resource complexity in terms of time.

\begin{table}[t]
\setlength{\tabcolsep}{1em}
\caption{A comparison of some algorithms for solving ground state preparation or ground-state energy estimation or both. Algorithms are analyzed in the regime where the prepared ground state has infidelity at most $\BigO{\poly(\tau^{-1})}$ and the estimated ground-state energy has additive error at most $B \gamma^{-1} \tau^{-1}$, with overall success probability at least $1 - e^{-\tau}$. Here $\BigTO{\cdot}$ omits logarithmic and polylogarithmic factors and Hamiltonian-dependent constant prefactors; $B^{-1} = \BigO{\gamma^{-2}}$. The symbol `$\backslash$' indicates that the corresponding quantity has not been fully analyzed. Query complexity denotes the expected total number of queries to controlled-$U_H$ and its inverse across all circuit executions; for methods with varying circuit depths, the reported scaling is a conservative upper bound obtained by multiplying the maximum query depth by the expected number of circuit runs.
``QPE'' stands for quantum phase estimation algorithms using the quantum Fourier transform; ``HT'' for the Hadamard test; ``QET'' for quantum eigenvalue transformation; ``ITE'' for imaginary-time evolution.}~\label{tab:ground}
\resizebox{\linewidth}{!}{
\begin{tabular}{lccc}
\toprule
Methods & Query depth & Query complexity & Ancilla  \\
\midrule
\addlinespace
Conventional QPE~\cite{kitaev1995quantum} & $\BigO{\gamma^{-1} \tau}$ & $\BigO{\gamma^{-3}\tau^2}$ & $\BigOmega{\log(\gamma^{-1} \tau)}$ \\
\addlinespace
Semi-classical QPE~\cite{griffiths1996semiclassical, higgins2007entanglement} & $\BigO{\gamma^{-3} \tau}$ & $\BigO{\gamma^{-5}\tau^2}$ & $1$ \\
\addlinespace
QET-based~\cite{dong2022groundstate} (state preparation) & $\BigO{\gamma^{-1}\tau}$ & $\BigTO{\gamma^{-3}\tau^2}$ & $3$ \\
\addlinespace
QET-based~\cite{dong2022groundstate} (energy estimation) & $\BigO{\gamma^{-1}\tau}$ & $\BigTO{\gamma^{-3}\tau^2}$ & $1$ \\
\addlinespace
HT-based~\cite{ding2023even} (energy estimation) & $\BigO{(1 - \gamma^2)^{1/2} \gamma^{-1}\tau}$ & $\BigTO{(1 - \gamma^2)^{1/2} \gamma^{-5}\tau^2}$ & $1$ \\
\addlinespace
ITE-based~\cite{motta2020determining, mcardle2019variational} & $\backslash$ & $\backslash$ & $0$ \\
\addlinespace
ITE-based, Theorem~\ref{thm:ground} & $\BigTO{\tau}$ & $\BigTO{\gamma^{-4}\tau^{4}}$ & $1$ \\
\addlinespace
\bottomrule
\end{tabular}
}
\end{table}

\subsubsection{Comparison with existing works}

Imaginary-time evolution has been proposed for ground-state problems since the earliest ITE algorithms~\cite{motta2020determining, mcardle2019variational}. Subsequent works on ITE~\cite{gomes2020efficient, nishi2021implementation, huang2023efficient, yeter-aydeniz2020practical, gluza2026double, zander2025role, liu2021probabilistic, kosugi2022imaginarytime, silva2023fragmented, chan2023simulating, yi2025probabilistic} have mainly demonstrated numerically how the ITE state converges to the ground state, using this as an indicator of algorithmic performance.
In parallel, there is a substantial literature on provable algorithms for ground state preparation and energy estimation~\cite{kitaev1995quantum, griffiths1996semiclassical, higgins2007entanglement, nielsen2010quantum, dong2022groundstate, wang2023quantum, ding2024quantum, wang2023ground, ding2023even, lin2022heisenberg, peruzzo2014variational, grimsley2019adaptive}. To our knowledge, however, there has not been a comparison between ITE-based methods and these more established approaches .

In this subsection we compare Algorithm~\ref{alg:ground} with representative algorithms that also query controlled time-evolution oracles, summarized in Table~\ref{tab:ground}. Throughout the comparison we use $\BigTO{\cdot}$ to suppress logarithmic and polylogarithmic factors and Hamiltonian-dependent constants. For a fair comparison, we consider the regime where the prepared ground state has infidelity $\BigO{\poly(\tau^{-1})}$ and the estimated ground-state energy has additive error $\epsilon = B\gamma^{-1}\tau^{-1}$ with overall success probability at least $1-e^{-\tau}$.

A major class of ground-state algorithms uses an ancilla register to control real-time evolution and extract eigenphase information via quantum phase estimation~\cite{kitaev1995quantum, griffiths1996semiclassical, higgins2007entanglement, nielsen2010quantum}, quantum signal processing~\cite{dong2022groundstate, wang2023quantum, ding2024quantum, wang2023ground}, or Hadamard-test-based interferometry~\cite{ding2023even, lin2022heisenberg}. Except for the last category which focuses on energy estimation, these algorithms can both prepare the ground state and estimate its energy within essentially the same circuit framework, and they achieve Heisenberg scaling in error at the cost of a query depth at least $\BigO{\epsilon^{-1}}$, which is  $\BigO{B^{-1}\gamma\tau} = \BigOmega{\gamma^{-1}\tau}$ under the same precision level.

In contrast, Algorithm~\ref{alg:ground} achieves a smaller maximum query depth in $U_H$ by using expectation-value estimation rather than Heisenberg-limited phase-estimation-type procedures. {For both ground state preparation and energy estimation, Theorem~\ref{thm:ground} shows that the query depth scales as $\BigTO{\tau}$, yielding a depth reduction by a factor of at least $\BigOmega{\gamma^{-1}}$ compared with QET-based schemes~\cite{dong2022groundstate} that achieve comparable accuracy and success probability.}
For ground-state energy estimation, the comparison with Ref.~\cite{ding2023even} depends on the overlap. When $\gamma$ is close to 1, Ref.~\cite{ding2023even} can achieve very short circuit depth; when $\gamma$ is only guaranteed to be $\BigOmega{\poly(n^{-1})}$, Algorithm~\ref{alg:ground} still improves the depth by at least a factor $\BigOmega{\gamma^{-1}}$ compared to Ref.~\cite{ding2023even}. 
Therefore, without additional assumptions, Algorithm~\ref{alg:ground} can reduce the query depth by a factor of at least $\BigOmega{\gamma^{-1}}$ for both problems when maintaining equivalent accuracy guarantees.

As a trade-off, Algorithm~\ref{alg:ground} does not improve the asymptotic quantum resource complexity for either problem.
Because our algorithm is based on expectation-value estimation, its precision scaling is limited by the standard quantum limit rather than the Heisenberg limit.
Taking $B = \BigO{\gamma^2}$ in the optimal case, the total query complexity for Algorithm~\ref{alg:ground} reaches $\BigTO{B^{-2}\tau^4} = \BigTO{\gamma^{-4}\tau^4}$, which exceeds that of existing works. 
However, this overhead primarily arises from repeated measurements of a small number of circuits, rather than from deep coherent evolutions. The algorithm requires only $\BigTO{\log\tau}$ distinct circuits to be constructed, and each circuit can be executed many times (including in parallel across multiple devices or chiplets) on hardware~\cite{arute2019quantum, zhong2020quantum}. From a practical perspective, this depth--repetition trade-off can be particularly favorable in early fault-tolerant regimes where circuit depth is a dominant constraint~\cite{katabarwa2024early}.

\subsection{Simulation of open quantum system}~\label{sec:lindbladian}

Unlike the closed-system setting discussed in previous sections, the dynamics of an $n$-qubit \emph{open} quantum system at time $t$ is governed by the \emph{Lindblad master equation}~\cite{manzano2020short}
\begin{equation}
    \ddx{t} \rho(t) = \cL[\rho(t)]
,\end{equation}
for a \emph{Lindbladian}
\begin{equation}
    \cL[\rho] \coloneqq -i[\Hsys, \rho] + \sum_{j=1}^m  D_j \rho D_j^\dag - \frac{1}{2} \{ D_j^\dag D_j, \rho \}
,\end{equation}
where $\Hsys$ is the system Hamiltonian and $D_1,\ldots, D_m$ are jump operators. The task of \emph{Lindbladian simulation} is to prepare, on a quantum computer, the solution $\rho(t) = e^{\cL t}[\rho(0)]$
for a given initial state $\rho(0)$. Efficient access to $\rho(t)$ underpins applications ranging from the analysis of decoherence and dissipation~\cite{plenio2008dephasing, mohseni2008environment, lindblad1976generators, gorini1976completely} to dissipative state engineering and computation~\cite{kliesch2011dissipative, barthel2012quasilocality, verstraete2009quantum, reiter2017dissipative, kwon2022reversing, shang2025designing}.

To separate the coherent and dissipative contributions, it is convenient to work in Liouville space. Vectorizing the density operator as $\kett{\rho(t)}$ yields
\begin{equation}
    \ddx{t} \kett{\rho(t)} = -i( H_c - iH) \kett{\rho(t)}
,\end{equation}
where $\kett{\cdot}$ denotes vectorization, and $H_c$ and $H$ are effective Hamiltonians corresponding to the coherent and dissipative parts, respectively.
Note that when all jump operators are Hermitian, the ground-state energy of $H$ is zero.
A first-order Trotter expansion gives
\begin{equation}~\label{eqn:lindbladian decomp}
    \kett{\rho(t)} = \big(e^{-iH_c \tau} e^{-H\tau} \big)^N \kett{\rho(0)} + \BigO{t^2 / N}
\end{equation}
for $\tau = t / N$. If a circuit for preparing $\rho(0)$ is available, one can prepare the normalized vectorized state $\kett{\rho(0)}$ on $2n$ qubits. Then, alternating coherent evolutions $e^{-iH_c \tau}$ and imaginary-time evolutions $e^{-H\tau}$ on the vectorized register approximates $\kett{\rho(t)}$~\cite{kamakari2022digital}, and observables can be evaluated by measuring suitable operators in Liouville space.

Ref.~\cite{kamakari2022digital} follows this idea and uses the Trotterized ITE algorithm~\cite{motta2020determining} as a subroutine to implement $e^{-H\tau}$, demonstrating on hardware that quantities of the form $\trace{O\rho(t)}$ can be estimated. However, most existing ITE methods are heuristic and lack rigorous complexity guarantees, as reviewed in Section~\ref{sec:ite work}, and Eq.~\eqref{eqn:lindbladian decomp} controls only the error for the \emph{unnormalized} vectorized state. It has therefore remained unclear whether such a method can, in general, prepare the normalized state $\kett{\rho(t)} / \norm{\kett{\rho(t)}}$ with provable accuracy.

We address this gap by replacing the heuristic ITE step with the rigorously analyzed ITE circuit $\UITE{\tau}{H}$ from Eq.~\eqref{eqn:qpp post-select}. The resulting Lindbladian-simulation routine is summarized in Algorithm~\ref{alg:lindbladian}. The jump operators are rescaled so that the dissipative Hamiltonian $H$ satisfies Assumption~(\nref{assum:normalize}). To keep the Trotter error small, we use short imaginary-time steps $\tau < 1$, and the analysis therefore builds directly on our short-time result, Corollary~\ref{coro:ite short}. After obtaining the normalization parameter $\lambda$ as in Section~\ref{sec:ite algorithm}, we obtain the following guarantee.

\begin{algorithm}[t]
\caption{Lindbladian simulation via ITE}\label{alg:lindbladian}
\SetKwInOut{Input}{Input}
\SetKwInOut{Output}{Output}
\Input{Evolution time $t$, step size $N$, normalized $\kett{\rho(0)}$, $\lambda \in [\abs{\lambda_0}, \abs{\lambda_0} + t^{-1}]$, simulation error $\epsilon$}
\Output{normalized $\kett{\rho(t)}$}

$\ket{\varphi} \gets \textrm{normalized }\kett{\rho(0)}$\;

$\tau \gets t /N$, $U_{H_c} \gets e^{-i H_c \tau}$\;

$\alpha \gets e^{-1/N}$\; 

$\UITEs \gets \UITE{\tau}{H}$ in Equation~\eqref{eqn:qpp post-select}\;

\For{$k = 1$, $\ldots$, $N$}{
    $\ket{\varphi} \gets U_{H_c} \UITEs[ \ket{\varphi} ]$\;
}

\Return $\kett{\rho(t)} / \norm{\kett{\rho(t)}} \mathrel{\overset{\approx}{\gets}} \ket{\varphi}$\;
\end{algorithm}

\begin{theorem}~\label{thm:lindbladian}
    Let $N > 0$. Under Assumptions~(\nref{assum:normalize},\nref{assum:oracle},\nref{assum:long evolution}), the state fidelity between the output state of Algorithm~\ref{alg:lindbladian} and the normalized state of $\kett{\rho(t)}$ is approximately lower bounded as
\begin{equation}
    1 - \BigO{e^{1/N} N \epsilon + t^2/\mu N}
,\end{equation}
    where $\mu = \norm{\rho(0)}_2 + \norm{\rho(t)}_2$ for $\norm{\cdot}_2 = \sqrt{\trace{(\cdot)^2}}$.
    Moreover, taking $\epsilon = \BigO{\poly(t^{-1})}$, Algorithm~\ref{alg:lindbladian} uses the following cost:
\begin{enumerate}[leftmargin=1em]
    \item [-] $N$ queries to $U_{H_c}$,
    \item [-] $\BigTO{N t}$ queries to controlled-$U_H$ and its inverse,
    \item [-] a maximal total of $\BigTO{N t}$ query depth of $U_{H}$ and $N$ query depth of $U_{H_c}$, and
    \item [-] one ancilla qubit initialized in the zero state.
\end{enumerate}
\end{theorem}

If $\rho(0)$ is pure, then $\mu > 1$, and taking $N = \BigO{t^3}$ yields an approximation to the normalized $\kett{\rho(t)}$ with error $\BigO{t^{-1}}$, while the number of queries to $U_H$ and $U_{H_c}$ in a single run scales polynomially in $t$. Furthermore, if the system Hamiltonian $\Hsys$ and the $m$ jump operators admit Pauli decompositions with at most $L$ Pauli terms each~\cite{cleve2017efficient} (the \emph{Pauli sparsity}), then $H_c$ and $H$ can be expressed in terms of at most $\BigO{L^2}$ Pauli terms using $\BigO{(m + n)L^2}$ classical operations. In this setting, a Trotter implementation of the coherent and dissipative steps leads to the following resource–accuracy trade-off.

\begin{theorem}~\label{thm:lindbladian trotter}
    Under Assumptions~(\nref{assum:normalize},\nref{assum:pauli},\nref{assum:long evolution}), if $\rho(0)$ is a pure state, Algorithm~\ref{alg:lindbladian} can prepare the normalized state of $\kett{\rho(t)}$ up to fidelity $1 - \BigO{\varepsilon}$ using the following cost:
\begin{enumerate}[leftmargin=1em]
    \item [-] $\BigTO{L^6 \Lambda^2 t^3 / \varepsilon}$ queries to (controlled) Pauli rotations, and
    \item [-] one ancilla qubit initialized in the zero state,
\end{enumerate}
    where $L$ is the number of Pauli terms and $\Lambda$ is the maximum absolute Pauli coefficient in the decompositions of $H_c$ and $H$.
\end{theorem}

The choice of $N$ and $\epsilon$ in Theorem~\ref{thm:lindbladian trotter} is deferred to Appendix~\ref{appendix:lindbladian}.
Unlike in our imaginary-time results for Hamiltonian systems, we do not obtain a general lower bound on the post-selection success probability for Algorithm~\ref{alg:lindbladian}. The difficulty is that we impose no structural constraint on $H_c$: in the worst case the coherent evolution can rotate the state almost entirely out of the ground-state subspace of $H$ after each step, preventing us from leveraging the ground-state overlap bounds used in Section~\ref{sec:ite algorithm}. One could enforce that $H_c$ implicitly preserves the ground-state subspace of $H$~\cite{shang2025fastforwardable, ding2024singleancilla}, but then $\kett{\rho(t)}$ would simply converge exponentially fast to the ground state of $H$, and taking $N = 1$ would already be sufficient once Assumption~(\nref{assum:long evolution}) holds. Identifying a less restrictive yet still meaningful condition that yields a nontrivial probability bound is an interesting open problem. As an interim step, we provide numerical evidence below that the success probability is nevertheless favorable in representative models.

\subsubsection{Numerical simulations}

We now examine the numerical performance of Algorithm~\ref{alg:lindbladian} for representative open-system models, and compare the observed errors and query complexities to the bounds in Theorem~\ref{thm:lindbladian}. Because the algorithm relies on post-selection, we define the effective query complexity as the product of a single stepwise execution cost and the expected number of step repetitions until success. The initial state $\rho(0)$ is chosen to be the zero state.

\emph{Ising model.---} We first consider a one-dimensional transverse-field Ising model (TFIM),
\begin{equation}
    \Hsys = J \sum_{i=1}^{n-1} Z_i Z_{i + 1} + h \sum_{i=1}^{n} X_i
,\end{equation}
where $J$ and $h$ are coupling and field strengths, respectively. This model is a standard benchmark in recent works on Lindbladian simulation~\cite{ding2024simulating,peng2025quantum,yu2025lindbladian,huang2025robust}. We adopt the same parameter choices as these references: $J \in \set{-1, -0.1, 1}$, $h \in \set{-1, -0.5, 0.2}$, and jump operators given either by single Pauli operators or complex linear combinations of two Paulis.

Figure~\ref{fig:lindbladian}(a) reports the average infidelity and query complexity as the simulation time $t$ is increased from $1$ to $20$, using $N = t^3$ Trotter steps. The state infidelity is defined as the non-overlap between the normalized exact vectorized state and the output of Algorithm~\ref{alg:lindbladian}. The infidelity decreases approximately as a polynomial in $t^{-1}$, in line with the $\BigO{t^{-1}}$ behavior predicted by Theorem~\ref{thm:lindbladian}. At the same time, the total query complexity for $U_H$ and $U_{H_c}$ grows polynomially in $t$, indicating that the algorithm remains practically implementable over the simulated time window.

\begin{figure}[t]
    \centering
    \includegraphics[width=\linewidth]{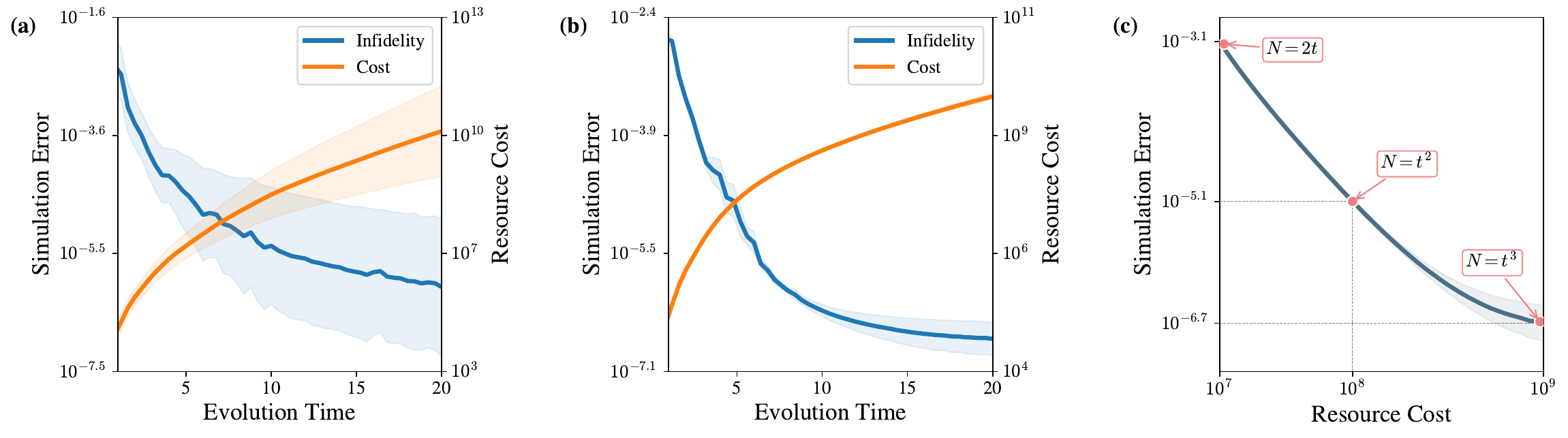}
    \caption{Numerical performance of Algorithm~\ref{alg:lindbladian} for Lindbladian simulation. \textbf{(a,b)} Infidelity and query complexity as a function of evolution time $t = 1, \ldots, 20$ with $N = t^3$ steps. Panel (a) shows results for four 4-qubit TFIM instances with jump operators from Ref.~\cite{ding2024simulating,peng2025quantum,yu2025lindbladian,huang2025robust}. Panel (b) shows an antiferromagnetic Heisenberg chain in Equation~\eqref{eqn:hamiltonian} with five random sets of jump operators. In each panel, the blue curve (log scale) is the state infidelity between the normalized $\kett{\rho(t)}$ and the algorithm output; the orange curve (log scale) is the total number of queries to $U_H$ and $U_{H_c}$ required, including repetitions due to post-selection. \textbf{(c)} Trade-off between infidelity and average query complexity for the Heisenberg model of (b) at final time $t = 20$, varying the number of steps $N$. Both axes are in logarithmic scale.}
    \label{fig:lindbladian}
\end{figure}

\emph{Heisenberg model.---} We next consider a more challenging setting where the system Hamiltonian is the antiferromagnetic Heisenberg model in Equation~\eqref{eqn:hamiltonian} on four qubits. To probe the impact of jump-operator structure, we generate five distinct sets of jump operators by taking differences of random Hermitian and anti-Hermitian matrices, thereby producing generic non-commuting dissipators.

As shown in Figure~\ref{fig:lindbladian}(b), the behavior in this model closely parallels the TFIM case: the infidelity decays polynomially with $t^{-1}$, while the total query complexity grows polynomially with $t$. Notably, the curves are less spread across different choices of jump operators than in the TFIM experiment, suggesting that the dominant contribution to both the error and cost is controlled by the underlying system Hamiltonian rather than the detailed dissipative structure.

For the Heisenberg setup, we also study the trade-off between accuracy and cost at fixed final time $t = 20$ by varying the number of Trotter steps $N$. Figure~\ref{fig:lindbladian}(c) shows that increasing $N$ by a factor of $t$ reduces the infidelity by roughly two orders of magnitude, while the query complexity grows only by about one order of magnitude. This empirically supports a scaling of the overall cost as $\BigO{\varepsilon^{-1/2}}$ in the target infidelity $\varepsilon$. Even for the relatively coarse choice $N = t$, the infidelity remains on the order of $10^{-3}$, with moderate resource overhead, indicating that Algorithm~\ref{alg:lindbladian} performs well in practically relevant parameter regimes. The observed behavior also suggests that the error bound in Theorem~\ref{thm:lindbladian} may be further sharpened.

\subsubsection{Comparison with existing works}

Quantum algorithms for Lindbladian simulation can be broadly divided into \emph{sampling-based} and \emph{deterministic} approaches. Both classes typically discretize the total evolution time $t$ into smaller steps in order to approximate the continuous dynamics and control the cumulative error.

Sampling-based schemes simulate stochastic unravelings of the Lindblad equation, such as stochastic Schr\"odinger equations, and estimate observables by averaging over many noise realizations~\cite{chen2024adaptive, huang2025robust, watad2024variational}, or by randomly sampling terms in Pauli/Kraus or Dyson-series expansions of the channel~\cite{yu2025lindbladian, peng2025quantum, patel2023wave2}. These methods avoid post-selection but introduce classical sampling overhead, and their variance typically depends on spectral properties of the Lindbladian and the observable of interest.

Deterministic schemes, in contrast, aim to implement an explicit approximation of the Lindblad channel itself on the system register. One established line of work uses an ancilla register together with a heralding measurement~\cite{cleve2017efficient, childs2017efficient, schlimgen2021quantum}, followed by oblivious amplitude amplification to drive the success probability close to one. Another line resets or traces out the ancilla after each step, thereby avoiding post-selection~\cite{hu2022general, kamakari2022digital, patel2023wave1, pocrnic2025quantum, ding2024simulating}. In both variants, the cost is often dominated by the need to coherently encode all jump operators into a single circuit; key metrics such as gate count, depth, and ancilla size can therefore grow with the number of jump operators $m$.

\begin{table}[t]
\setlength{\tabcolsep}{1em}
\caption{A comparison of some open-system simulation algorithms that directly approximate the Lindblad channel, assuming $\Lambda \leq 1$. Here the parameters $t$, $\varepsilon$, $m$, and $L$ denote the evolution time, target precision, number of jump operators, and Pauli sparsity, respectively. 
The symbol `$\backslash$' indicates that the corresponding quantity has not been fully analyzed. Query complexity denotes the total number of primitive queries (e.g., to controlled-Pauli rotations) across all circuit executions. For the ITE-based method, the reported query complexity includes the empirically observed post-selection overhead; a rigorous bound on this factor remains open.
``LCU'' stands for linear combination of unitaries; ``QAA'' for quantum amplitude amplification; “RI” for repeated interactions; ``DH'' for dilated Hamiltonian of order $p \geq 2$ (complexity analysis as in~\cite{yu2025lindbladian}); ``ITE'' for imaginary-time evolution.}~\label{tab:lindbladian}
\resizebox{\linewidth}{!}{
\begin{tabular}{lccc}
\toprule
Methods & Circuit depth & Query complexity & Ancilla  \\
\midrule
\addlinespace
LCU-based with QAA~\cite{cleve2017efficient} & $\BigO{m^2 \cdot L^2 t \poly\log(mLt / \varepsilon)}$ & $\BigO{m^2 \cdot L^2 t \poly\log(mLt / \varepsilon)}$ & $\BigO{\poly\log(mLt/\varepsilon)}$ \\
\addlinespace
RI-based~\cite{pocrnic2025quantum} & $\BigO{m^3 \cdot L t^2 / \varepsilon}$ & $\BigO{m^3 \cdot L t^2 / \varepsilon}$ & $\BigOmega{\log L}$ \\
\addlinespace
$p$-order DH~\cite{ding2024simulating} & $\BigO{m^p \cdot L^p t (t/\varepsilon)^{1/p}}$ & $\BigO{m^p \cdot L^p t (t/\varepsilon)^{1/p}}$ & $\BigO{p \log m + \log L}$ \\
\addlinespace
ITE-based~\cite{kamakari2022digital} & $\backslash$ & $\backslash$ & $n + 1$ \\
\addlinespace
ITE-based, Theorem~\ref{thm:lindbladian trotter} & $\BigTO{L^6 t^3 / \varepsilon}$ & numerically $\BigO{L^6 \poly(t) / \varepsilon}$ & $n + 1$ \\
\addlinespace
\bottomrule
\end{tabular}
}
\end{table}

The ITE-based method of Ref.~\cite{kamakari2022digital}, which approximates the Lindbladian in Liouville space, has the important feature that its gate complexity is essentially independent of $m$. However, because it relies on heuristic ITE primitives~\cite{motta2020determining}, it lacks rigorous guarantees on precision and resource scaling, especially for long simulation times.

Algorithm~\ref{alg:lindbladian} inherits the favorable structural property of Ref.~\cite{kamakari2022digital}: for fixed Pauli sparsity $L$, the per-run circuit depth implied by Theorem~\ref{thm:lindbladian trotter} does not depend on the number of jump operators $m$. Table~\ref{tab:lindbladian} compares this behavior with three representative deterministic algorithms~\cite{cleve2017efficient,pocrnic2025quantum,ding2024simulating}. In the LCU- and DH-based approaches~\cite{cleve2017efficient,ding2024simulating}, the depth grows at least quadratically in $m$, and the number of ancilla qubits scales at least logarithmically in $m$. RI-based schemes~\cite{pocrnic2025quantum} avoid $m$-dependent ancilla overhead but incur at least cubic dependence on $m$ in gate complexity. By contrast, our ITE-based method uses $n + 1$ ancilla qubits and comes with a stronger dependence on the Pauli sparsity $L$ inherited from the Liouville-space vectorization: the depth scales as $\BigO{L^6 t^3/\varepsilon}$, whereas LCU-based methods have a quadratic dependence on $L$~\cite{cleve2017efficient}. Moreover, Algorithm~\ref{alg:lindbladian} relies on post-selection, and a rigorous lower bound on its success probability is currently unavailable in full generality (see the discussion following Theorem~\ref{thm:lindbladian}); accordingly, the ITE-based query-complexity entry in Table~\ref{tab:lindbladian} is based on numerical evidence.

Our implementation does not require multi-controlled logic or oracle access to sophisticated block encodings, which simplifies hardware realization. 
Although Algorithm~\ref{alg:lindbladian} requires a normalization factor $\lambda$ as a prerequisite, such value can be stored for future usage if the setup of noise channels remain unchanged.
Overall, when the number of jump operators $m$ is large relative to the Pauli sparsity $L$ so that the $m$-dependent factors in existing deterministic methods dominate, Algorithm~\ref{alg:lindbladian} can achieve shorter per-run circuit depth at comparable accuracy. In such regimes, the method is well-suited to early fault-tolerant quantum computers, provided that the post-selection overhead remains manageable.

\section{Discussions and outlook}

In this work, we have introduced a quantum algorithm for preparing normalized imaginary-time evolved states with rigorously proven polynomial resource scaling in the imaginary-time duration. Our algorithm stabilizes the resource cost by adaptively determining an appropriate normalization factor, contrasting with previous methods that may suffer from exponentially increased costs as imaginary time grows. Under the assumption that the initial state has good overlap ($\gamma = \BigOmega{\poly(n^{-1})}$) with the target ground state, our approach also achieves polynomial scaling with respect to the number of qubits. {To our knowledge, this is the first imaginary-time-evolution algorithm with a fully rigorous polynomial resource bound in the evolution time.} Numerical experiments validate the algorithm's effectiveness and robustness for long imaginary-time evolutions.

We provide a quantum algorithm that applies imaginary-time evolution to ground-state-related problems. Our algorithm prepares the ground state and estimates the ground-state energy using circuits with reduced query depth, despite requiring a heuristic assumption. As a trade-off for higher measurement cost (and higher total query complexity), the maximum query depth in our algorithm can be reduced by a factor of $\BigOmega{\gamma^{-1}}$ compared with representative Heisenberg-limited approaches~\cite{dong2022groundstate, ding2023even}. This reduction makes the algorithm suitable for near-term and early fault-tolerant settings where circuit depth is a dominant constraint.

As a side remark, combining Theorem~\ref{thm:ite uH} with Theorem~\ref{thm:ground} suggests a close computational connection between imaginary-time evolution and ground-state-related problems under our assumptions. Theorem~\ref{thm:ite uH} shows that any efficient ground-state energy estimation algorithm with precision scaling as $\BigO{\tau^{-1}}$ can be used to implement our ITE algorithm with error polynomially small in $\tau^{-1}$, while Theorem~\ref{thm:ground} shows conversely that accurate ITE enables ground-state energy estimation with comparable resource scaling. Clarifying whether this correspondence can be formalized as a computational equivalence is an interesting direction for future work.

Following the idea in Ref.~\cite{kamakari2022digital}, we also provide a quantum algorithm that applies imaginary-time evolution to the problem of open-system (Lindbladian) simulation. Our construction removes the dependence of the per-run circuit depth on the number of dissipative terms, at the expense of a stronger dependence on the Pauli sparsity $L$ and the number of ancilla qubits ($n + 1$), both inherited from the Liouville-space formulation. Moreover, a rigorous bound on the post-selection overhead is currently unavailable in full generality. Even so, the resulting circuits can be shallower than those of existing deterministic methods when many local noise channels are present. We expect that combining our ITE-based techniques with more recent developments such as sampling schemes and trace-preserving implementations without post-selection~\cite{ding2024simulating, peng2025quantum,yu2025lindbladian,huang2025robust, pocrnic2025quantum} will lead to further improvements in open-system simulation algorithms.

The theoretical analysis relies on assumptions that are practically justified in many scenarios. When the initial state has negligible overlap with the ground state but non-trivial overlap with the first excited state, our analysis extends naturally to the excited-state scenario by replacing $\abs{\lambda_0}$ with $\abs{\lambda_1}$. For degenerate Hamiltonians, our results generalize through extending the discussion from pure states to corresponding eigenspace projectors, though this increases the complexity of theoretical analysis.

One technical challenge requiring future investigation concerns classical machine precision limitations at large evolution times. Appendix~\ref{appendix:QPP} provides a Fourier approximation of the exponential function with theoretical super-polynomial convergence, yet numerical instability arising from finite precision prevents reliable implementation of this convergence on classical devices.

More broadly, the polynomial-resource imaginary-time evolution framework established here provides a rigorous algorithmic foundation for a range of quantum simulation tasks in quantum chemistry, condensed matter physics, and quantum field theory. The circuit-depth reductions demonstrated for ground-state problems and open-system simulation are particularly relevant for near-term quantum hardware, including superconducting-qubit and trapped-ion processors, where limited coherence times impose stringent depth constraints. As fault-tolerant quantum devices become increasingly accessible, we anticipate that ITE-based algorithms will serve as practical building blocks for broader classes of quantum simulation and optimization problems.

\section*{Acknowledgement}
We would like to thank Yu-Ao Chen and Zhan Yu for their helpful comments. We also thank the anonymous reviewers of AQIS 2025 and QIP 2026 for their useful comments. Part of this work was finalized
while L. Zhang and X. Wu were at QudeLeap Research.

This work was partially supported by the National Key R\&D Program of China (Grant
No. 2024YFB4504004), the National Natural Science Foundation of China (Grant. Nos. 92576114,
12447107), the Guangdong Provincial Quantum Science Strategic Initiative (Grant Nos. GDZX2403008,
GDZX2503001, GDZX2403001), the Guangdong Natural Science Foundation (Grant No. 2025A15150
12834), the Guangdong Provincial Key Lab of Integrated Communication, Sensing and Computation
for Ubiquitous Internet of Things (Grant No. 2023B1212010007), the Quantum Science Center of
Guangdong-Hong Kong-Macao Greater Bay Area, and the Education Bureau of Guangzhou Municipality.

\section*{Data Availability}
Numerical experiments are based on an open-source Python research software for quantum computing~\cite{quairkit}. Code and data used in the numerical experiments are available on \url{https://github.com/QuAIR/QITE-codes}. 
\newcommand{\etalchar}[1]{$^{#1}$}


\clearpage
\appendix
\setcounter{subsection}{0}
\setcounter{table}{0}
\setcounter{figure}{0}

\vspace{3cm}

\begin{center}
\Large{Appendix for} \\ Quantum Imaginary-Time Evolution with Polynomial Resources in Time
\end{center}

\numberwithin{equation}{section}
\renewcommand{\theproposition}{S\arabic{proposition}}
\renewcommand{\thelemma}{S\arabic{proposition}}
\renewcommand{\thetheorem}{S\arabic{proposition}}
\renewcommand{\thedefinition}{S\arabic{definition}}
\renewcommand{\thefigure}{S\arabic{figure}}

\renewcommand{\thetable}{S\arabic{table}}
\renewcommand{\thefigure}{S\arabic{figure}}

\setcounter{equation}{0}
\setcounter{table}{0}
\setcounter{section}{0}
\setcounter{proposition}{0}
\setcounter{definition}{0}
\setcounter{figure}{0}

\tableofcontents

\clearpage
\section{Assumptions}~\label{appendix:assumption}

\begin{table}[htbp]
\setlength{\tabcolsep}{1em}
\caption{Summary of assumptions and the theoretical results they support. The `Type' column classifies each assumption as either specific to the problem (SPEC) or made without loss of generality (WLOG).}
\resizebox{\linewidth}{!}{
\begin{tabular}{lc p{22em} p{12.5em}}
\toprule
Assumption & Type & Description & Supporting Results \\
\midrule
\addlinespace
\assumption{assum:normalize}{normalized}~ & WLOG & all eigenvalues of $H$ lie within the interval $[-1,1]$, and the ground-state energy $\lambda_0$ is negative  & Lemma~\ref{lem:succ prob low bound}; Theorem~\ref{thm:ite uH},~\ref{thm:ite trotter},~\ref{thm:ground},~\ref{thm:lindbladian},~\ref{thm:lindbladian trotter} \\
\addlinespace
\assumption{assum:oracle}{evolution oracle} & WLOG & (controlled-)$U_H$ and its inverse can be accessed with finite copies & Theorem~\ref{thm:ite uH},~\ref{thm:ground},~\ref{thm:lindbladian} \\
\addlinespace
\assumption{assum:pauli}{Pauli form} & WLOG & $H$ is a linear combination of Pauli operators with known coefficients & Theorem~\ref{thm:ite trotter},~\ref{thm:lindbladian trotter} \\
\addlinespace
\assumption{assum:long evolution}{long evolution} & WLOG & Time is large enough to make the problem meaningful & Lemma~\ref{lem:succ prob low bound}; Theorem~\ref{thm:ite uH},~\ref{thm:ite trotter},~\ref{thm:lindbladian},~\ref{thm:lindbladian trotter} \\
\addlinespace
\assumption{assum:overlap}{non-zero overlap} & WLOG & the state overlap $\gamma$ between initial state and the ground state is positive & Lemma~\ref{lem:succ prob low bound},~\ref{lem:overlap lower bound}; Theorem~\ref{thm:ite uH},~\ref{thm:ite trotter},~\ref{thm:ground} \\
\addlinespace
\assumption{assum:reprod}{reproducibility} & WLOG & initial state can be accessed with finite copies & Theorem~\ref{thm:ite uH},~\ref{thm:ite trotter},~\ref{thm:ground} \\
\addlinespace
\assumption{assum:good overlap}{good overlap} & SPEC & $\gamma = \BigOmega{\poly(n^{-1})}$ & Theorem~\ref{thm:ite trotter},~\ref{thm:ground} \\
\addlinespace
\assumption{assum:nondegenerate}{non-degenerate} & SPEC & the energy spectral gap $\Delta = \lambda_1 - \lambda_0$ is non-zero & Lemma~\ref{lem:overlap lower bound}; Theorem~\ref{thm:ite trotter},~\ref{thm:ground} \\
\addlinespace
\assumption{assum:long evolution gap}{distinguishable gap} & SPEC & $\Delta = \BigOmega{\tau^{-1} \log \poly(\tau)}$ & Theorem~\ref{thm:ite trotter} \\
\addlinespace
\assumption{assum:priori}{priori knowledge} & SPEC &  knowing a quantity $B$ that satisfies $\gamma^2 \abs{\lambda_0} \geq e^2 B > 0$ & Theorem~\ref{thm:ground} \\
\addlinespace
\bottomrule
\end{tabular}
}
\end{table}

\section{Polynomial transformations of unitaries}~\label{appendix:QPP}

Let $f$ be a function mapping from $\RR$ to $\CC$. $f$ is a degree-$L$ \emph{polynomial} if $f(x) = \sum_{j=0}^L c_j x^j$ for some vector $c \in \CC^{L + 1}$. $f$ is a degree-$L$ \emph{Laurent polynomial} in $\CC[X, X^{-1}]$ if $f(x) = \sum_{j=-L}^L c_j X^j$ for some vector $c \in \CC^{2L + 1}$.  $f$ is a \emph{trignometric polynomial} if $f \in \CC[e^{ix}, e^{-ix}]$.
The $j$-th derivative of a univariate function $f$ is denoted as $f^{(j)}$.

Let $p$ satisfy $1 \leq p \leq \infty$. The $L^p$-\emph{norm} of $f$ within interval $[a, b]$ is defined as $\norm{f}_{p, [a, b]} = (\integral{a}{b}{\abs{f}^p})^{1/p}$. $f$ is \emph{square integrable} on $[a, b]$ if $\norm{f}_{2, [a, b]} < \infty$. Such norm is called the \emph{supremum norm} when $p = \infty$, in which case $\norm{f}_{\infty, [a, b]} = \max_{x \in [a, b]} \abs{f(x)}$. The $L^p$-\emph{distance} between $f$ and $f'$ within interval $[a, b]$ is $\norm{f - g}_{p, [a, b]}$. Without further assumption, we denote $\norm{\cdot}_p = \norm{\cdot}_{p, [-\pi, \pi]}$ for convenience.

Let $f: [-\pi, \pi] \to \setcond{x\in \CC}{\abs{x} \leq 1}$ be a square-integrable function. We can extend the domain of $f$ to the unitary group by applying $f$ on the eigenphases of these unitaries.
Such extension is defined as follows:

\begin{definition}[Eigenphase transformation]~\label{def:f(u)}
    Let $U$ be a unitary operator with spectral decomposition $U = \sum_j e^{i \tau_j} \ketbra{\chi_j}{\chi_j}$, with $\tau_j \in [-\pi, \pi]$. The \emph{eigenphase transformation} of $U$ under $f$, denoted as $f(U)$, is defined as
\begin{equation}
    f(U) = \sum_j f(\tau_j) \ketbra{\chi_j}{\chi_j}
.\end{equation}
\end{definition}

When $f(x) = \sum_j c_j e^{ijx}$, $f(U) = \sum_j c_j U^{j}$ is simply a polynomial of $U$ and $U^{-1} = U^\dag$. This is where \emph{quantum phase processing} (QPP)~\cite{wang2023quantum} comes into play.
Equivalent up to a global phase, the QPP circuit for simulating degree-$L$ trigonometric polynomial $F \in \CC[e^{ix}, e^{-ix}]$ is constructed as
\begin{equation}
    \UQPP(U) \coloneqq A\left(\theta^Y_0, \theta^Z_0\right)_\aux \left[
    \prod_{l=1}^{L}
    \begin{bNiceMatrix}
        U^\dag & 0 \\
        0 & I^{\ox n}
    \end{bNiceMatrix} A\left(\theta^Y_{2l - 1}, \theta^Z_{2l - 1}\right)_\aux \begin{bNiceMatrix}
        I^{\ox n} & 0 \\
        0 & U
    \end{bNiceMatrix} A\left(\theta^Y_{2l}, \theta^Z_{2l}\right)_\aux
    \right]
,\end{equation}
where $A\left(\theta^Y_j, \theta^Z_j\right) = R_y(\theta^Y_j) R_z(\theta^Z_j)$ is applied on the ancilla qubit.

From a more general perspective, we can view $\UQPP$ as a quantum comb~\cite{chiribella2008quantum} (or a quantum circuit architecture). Under this prospective, one can treat controlled-$U$ and its dagger as inputs of $\UQPP$, and outputs a quantum process that applies $F(U)$ to an input state $\ket{\phi}$ with probability $\norm{F(U) \ket{\phi}}^2$. Note that the angles $\theta^Y, \theta^Z$ and degree $L$ depend on the choice of $F$, one can thereby extend the definition of above structure to simulate more general functions.

\begin{definition}~\label{def:qpp comb}
    Let $\bar{\mathbb{D}} \coloneqq \setcond{x\in \CC}{\abs{x} \leq 1}$ and $f: [a, b] \to \bar{\mathbb{D}}$ be a square-integrable function for $[a, b] \subseteq [-\pi, \pi]$, and let $\epsilon > 0$. A sequential quantum comb is said to be a \emph{QPP comb} $\UQPPs{f}{\epsilon}$ that approximates $f$ within error $\epsilon$, if the comb uses one ancilla qubit initialized in the state $\ket{0}$ and inputs controlled-$U$ and its inverse to simulate the operator $f(U)$ within error $\epsilon$. Formally, the comb satisfies for any input unitary $U$ with eigenphases (modulo $2\pi$) in $[a, b]$,
\begin{equation}~\label{eqn:qpp approx}
    \left(\bra{0} \ox I_n\right) \UQPPs{f}{\epsilon}(U) \left(\ket{0} \ox I_n\right) = F(U) 
    \textrm{, where } \norm{F}_\infty \leq 1 \textrm{ and } \norm{f - F}_{\infty, [a, b]} \leq \epsilon
.\end{equation}
    Moreover, $\UQPPs{f}{\epsilon}$ is said to be an \emph{$L$-slot} QPP comb if the total number of queries to controlled-$U$ and its inverse in $\UQPPs{f}{\epsilon}(U)$ is $L$.
\end{definition}

\begin{theorem}[Theorem 1 in~\cite{wang2023quantum}]~\label{thm:qpp existance}
There exists a $2L$-slot QPP comb $\UQPPs{F}{0}$ for any degree-$L$ trigonometric polynomial $F \in \CC[e^{ix}, e^{-ix}]$ satisfying $\norm{F}_\infty \leq 1$. 
\end{theorem}

Since square-integrable functions can be approximated by its Fourier expansions, the above theory gurantees the existance of $\UQPPs{f}{\epsilon}$ for any such $f, \epsilon$.
A natural question then arises regarding the practical implementation: how many slots are required to realize this quantum comb?
The required number of slots depends directly on how accurately the function $f$ can be approximated by its Fourier series.

\subsection{Exponential transformation}

We will show that in our case, i.e., $f(x) = e^{\tau(x - \lambda)}$ defined in $[-1, \lambda]$, 
there exists a trigonometric polynomial that converges to $e^{\tau(x - \lambda- \mu)} = e^{-\tau \mu} f(x)$ for some constant shift $\mu \in [0, 1/\tau)$, with error decays superpolynomially as the approximation degree increases. We first need to introduce the Jackson's theorem, where we change $[0, 2\pi]$ in the original statement to $[-\pi, \pi]$ without loss of generality.

\begin{theorem}[Jackson's theorem for smooth function~\cite{jackson1911genauigkeit}]~\label{thm:jackson}
    Suppose $g: [-\pi, \pi] \to \bar{\mathbb{D}}$ is a smooth (i.e., infinitely differentiable) periodic function. Let $p$ be a positive integer. Then there exists a positive constant $K_{p}$ depending only on $p$, for every positive integer $L$, there exists a trigonometric polynomial $G \in \CC[e^{ix}, e^{-ix}]$ of degree at most $L$ such that for all $x \in [-\pi, \pi]$,
\begin{equation}
    \abs{g(x) - G(x)} \leq K_{p} \norm{g^{(p)}}_{\infty} \cdot (L + 1)^{-p}
.\end{equation}
\end{theorem}

Note that $f$ is smooth in $[-1, \lambda]$. To apply Theorem~\ref{thm:jackson}, one can extend $f$ to a smooth function $g$ up to a constant, while maintaining the relation with $\tau$. When $g$ is defined in $[-\pi, \pi]$, $g$ will be naturally periodic as long as the behaviors of $g$ at $x = \pm \pi$ coincides.
One example of such $g$ can be a multiplication between $f$ and a ``bump function'' $\rho$. Here $\rho:[-\pi, \pi] \to [0, 1]$ is defined as
\begin{equation}~\label{eqn:bump def}
    \rho(x) = \begin{cases}
        1, & x \in [-1, \lambda]; \\
        \beta\left((x + 1 + \mu) / \mu\right), & x \in (-1 - \mu, -1); \\
        \beta\left((\lambda + \mu - x) / \mu\right), & x \in (\lambda, \lambda + \mu); \\
        0, & x \in [-\pi, -1 - \mu) \cup (\lambda + \mu, \pi], \\
    \end{cases}
\end{equation}
with $\beta$ given as
\begin{equation}
    \beta(z) = \frac{\varphi(z)}{\varphi(z) + \varphi(1 - z)} \textrm{, \quad where\,\, }
    \varphi(z) = \begin{cases}
        e^{-1/z}, & z > 0; \\
        0, & z \leq 0.
    \end{cases}
\end{equation}

\begin{lemma}~\label{lem:f approx construction}
    Let $\tau > 0$, $\lambda \in (0, 1]$, $\mu \in (0, 1/\tau]$ and $\rho$ be as defined in Equation~\eqref{eqn:bump def}. Then $g(x) = \rho(x) \cdot e^{\tau(x - \lambda- \mu)}$ satisfies
\begin{enumerate}
    \item $g(x) = e^{\tau(x - \lambda- \mu)}$ for all $x \in [-1, \lambda]$ ;
    \item $\abs{g(x)} \leq 1$ for all $x \in [-\pi, \pi]$ ;
    \item $g$ is smooth on $[-\pi, \pi]$.
\end{enumerate}
\end{lemma}
\begin{proof}
    The first and second conditions holds by the construction of $g$. Sine the product of smooth functions are smooth, the rest of the proof is to show $\rho$ in Equation~\eqref{eqn:bump def} is smooth on $[-\pi, \pi]$.

    Observe that $\varphi(z)$ is a smooth function as 
\begin{align}
    \lim_{z\to 0^+} \frac{\operatorname{d}^p}{\operatorname{d}\!z^p} \varphi(z) 
    &= \lim_{z\to 0^+} e^{-1/z} \eta(z) \textrm{, \, with \,} \eta(z) = \BigO{\poly(1/z)} \\
    &= 0 = \lim_{z\to 0^-} \frac{\operatorname{d}^p}{\operatorname{d}\!z^p} \varphi(z)
\end{align}
    and $\varphi(z) + \varphi(1 - z) > 0$ for all $z \in \RR$. Then $\beta$ is a smooth function. The only thing left are the smoothness on the boundary of intervals $x = -1 - \mu, -1, \lambda, \lambda + \mu$. Similar to above reasoning, one can check that
\begin{equation}
    \lim_{z\to 0^+} \frac{\operatorname{d}^p}{\operatorname{d}\!z^p}\beta(z) 
    = \lim_{z\to 1^-} \frac{\operatorname{d}^p}{\operatorname{d}\!z^p}\beta(z)  = 0
\end{equation}
    and hence $\rho$ is smooth on interval boundaries.
\end{proof}
\vspace{1em}

Subsequently, $\alpha, \expfr$ mentioned in Equation~\eqref{eqn:f def} can be construted as
\begin{equation}~\label{eqn:f def complete}
    \alpha = e^{-\tau \mu} \textrm{\,\, and \,\,}
    \expfr(x) = g(x) \textrm{\,\, for all \,} x \in [\lambda, 1]
.\end{equation}
This summarizes to the following result:

\begin{theorem}~\label{thm:qpp converge}
    Let $\tau > 0$, $\lambda \in (0, 1]$, $\mu = \Theta(1/\tau)$ and $\epsilon = \BigO{\poly (\tau^{-1})}$. Suppose $f$ is defined as $f(x) = e^{\tau(x - \lambda - \mu)}$ for all $x \in [-1, \lambda]$.
    There exists an $2L$-slot QPP comb $\UQPPs{\expf}{\epsilon}$ with $L = \BigO{\tau \cdot \poly\log \tau}$, such that for all input unitary $U$ with eigenphases in $[-1, \lambda]$,
\begin{equation}
    \left(\bra{0} \ox I_n\right) \UQPPs{f}{\epsilon}(U) \left(\ket{0} \ox I_n\right) = F(U) \textrm{, where } 
    \norm{f - F}_{\infty, [-1, \lambda]} \leq \epsilon
.\end{equation}
\end{theorem}
\begin{proof}

    Let $g(x)=\rho(x)e^{\tau(x-\lambda-\mu)}$ be the smooth extension constructed in Lemma~\ref{lem:f approx construction}. 
    We bound $\norm{g^{(l)}}_{\infty}$ for each positive integer $l$. By the Leibniz rule $(uv)^{(l)}=\sum_{k=0}^l\binom{l}{k}u^{(k)}v^{(l-k)}$,
\begin{align}
    \norm{g^{(l)}}_{\infty}
    &= \norm{\sum_{k=0}^{l} \binom{l}{k}
    \rho^{(k)}(x) \cdot
    \ddx[l-k]{x} e^{\tau(x-\lambda-\mu)}}_{\infty} \\
    &= \norm{\sum_{k=0}^{l} \binom{l}{k}
    \rho^{(k)}(x)\tau^{l-k}e^{\tau(x-\lambda-\mu)}}_{\infty}
    \leq \sum_{k=0}^{l} \binom{l}{k}
    \norm{\rho^{(k)}}_{\infty}\tau^{l-k}
.\end{align}
    On each interval where $\rho$ changes between $0$ and $1$, it has the form $\beta((x-a)/\mu)$ for some constant $a$ and for the fixed smooth profile $\beta$. Hence, for $B_k=\norm{\beta^{(k)}}_{\infty}$, $\norm{\rho^{(k)}}_{\infty}\leq B_k\mu^{-k}$.
    It follows that there exists a constant $D_l>0$, independent of $\tau$, $\mu$, and $L$, such that $\norm{g^{(l)}}_{\infty}\leq D_l(\tau+\mu^{-1})^l$.
    By Theorem~\ref{thm:jackson}, for every positive integer $L$ there exists a trigonometric polynomial $F \in \CC[e^{ix},e^{-ix}]$ of degree at most $L$ such that
\begin{equation}
    \norm{f- F}_{\infty, [-1, \lambda]} 
    \leq \norm{g-F}_{\infty}
    \leq K_l \norm{g^{(l)}}_{\infty} (L+1)^{-l}
    \leq A_l \left(\frac{\tau+\mu^{-1}}{L+1}\right)^l
,\end{equation}
    where $A_l = K_l D_l$ is independent of $\tau$, $\mu$, and $L$.
    Since $\mu=\Theta(1/\tau)$, $\tau+\mu^{-1}=\Theta(\tau)$. 
    Choosing $l=\ceil{\log(1/\epsilon)}$ and $L=\ceil{e A_l^{1/l}(\tau+\mu^{-1})}$ gives the above error at most $e^{-l}\leq \epsilon$.
    The factor $A_l^{1/l}$ contributes only polylogarithmic factors when $l=\BigO{\log(1/\epsilon)}$, and hence $L = \tau \cdot \poly\log(1/\epsilon)$.
    Since $\epsilon=\BigO{\poly(\tau^{-1})}$, this gives $L = \BigO{\tau \cdot \poly\log \tau}$.
Then Theorem~\ref{thm:qpp existance} implies that there exists a $2L$-slot QPP comb $\UQPPs{F}{0}$. By Definition~\ref{def:qpp comb}, above inequalities implies $\UQPPs{F}{0}$ is equivalent to $\UQPPs{f}{\epsilon}$, as required.
\end{proof}
\vspace{1em}

As a side note, there is an inherent lower bound on the constant $\alpha$, given by the following inequality:
\begin{equation}~\label{eqn:alpha lowerbound}
    1 \geq \alpha > \sqrt{\frac{(1 + \tau^{-1})\,e^1}{(1 - \tau^{-1})\,e^2 - 2\tau^{-1}}}
.\end{equation}
In the limit as $\tau \to \infty$, the RHS reduces to $e^{-1/2} \approx 0.6065$. 
This lower bound arises to ensure that the stopping criteria defined in Proposition~\ref{prop:vqe stop condition} can be properly triggered during Algorithm~\ref{alg:ground}.
In practical numerical implementations, for $\tau \geq 5$, we may set $\alpha = 0.85$ i.e., $\mu \leq 1/6.153\tau$.

\subsection{Quantum phase estimation}

Given an eigenstate $\ket{\psi}$ of a unitary $U$ and its evolution operator $U$, the problem of quantum phase estimation is to estimate the corresponding eigenvalue $x$ such that $U \ket{\psi} = e^{ix} \ket{\psi}$. Similar to Ref.~\cite{martyn2021grand, dong2022groundstate}, QPP can simulate the STEP function
\begin{equation}
    f(x - a)= \begin{cases}
        0, & \textrm{if } x < a; \\
        1, & \textrm{otherwise }
    \end{cases}
\end{equation}
to allocate such $x$. We summarize the results in Ref.~\cite{wang2023quantum} as follows:

\begin{theorem}[Algorithm 1, Lemma 3, Theorem 3 in \cite{wang2023quantum}]~\label{thm:qpe}
    Suppose $\ket{\phi}$ be an input state. Then under Assumptions~(\nref{assum:normalize}, \nref{assum:overlap}, \nref{assum:reprod}), we can obtain an estimation of the ground-state energy $\lambda_0$ up to $\varepsilon$ precision with failure probability $\eta$, using
\begin{enumerate}[leftmargin=1em]
    \item [-] $\BigO{\gamma^{-2}\varepsilon^{-1} \log \left(\varepsilon^{-1}\log(\gamma^{-2}\eta^{-1}) \right)}$ queries to controlled-$U$ and its inverse,
    \item [-] $\BigO{\gamma^{-2}}$ copies of $\ket{\phi}$, 
    \item [-] $\BigO{\varepsilon^{-1} \log \left(\varepsilon^{-1}\log(\gamma^{-2}\eta^{-1}) \right)}$ maximal query depth of $U$, and
    \item [-] one ancilla qubit initialized in the zero state.
\end{enumerate}
\end{theorem}
\begin{proof}
    Theorem 3 in \cite{wang2023quantum} states that Algorithm 1 can use 1 ancilla qubit and $\BigO{\varepsilon^{-1} \log \left(\varepsilon^{-1}\log\eta'^{-1} \right)}$ queries to controlled-$U$ and its inverse to obtain an eigenvalue $x$ with precision $\varepsilon$ and failure probability $\eta'$, while the probability that $x$ is the ground-state energy is $\gamma^2$. Then one can repetitively apply Algorithm 1 sufficiently many (around $\BigO{\gamma^{-2}}$) times such that an estimation of $\lambda_0$ is obtained. The overall failure probability would be $\eta = 1-  (1 - \eta')^{\gamma^{-2}} \approx \gamma^{-2}\eta'$. Then the overall resource cost includes $\BigO{\gamma^{-2}\varepsilon^{-1} \log \left(\varepsilon^{-1}\log(\gamma^{-2}\eta^{-1}) \right)}$ queries to controlled-$U$ and its inverse and $\BigO{\gamma^{-2}}$ copies of $\ket{\phi}$.
    As for the query depth, note that Algorithm 1 is completed by one quantum circuit, so the query depth equals the number of oracle queries used within that circuit, as required.
\end{proof}

\section{Theories in imaginary-time evolution}~\label{appendix:ite}

\begin{lemma}~\label{lem:qpp norm}
    Let $\epsilon \in (0, 1)$, $\tau > 0$, $\lambda \in [\abs{\lambda_0}, 1]$ and $\expf$ be as defined in Equation~\eqref{eqn:f def}. Then under Assumptions~(\nref{assum:normalize},\nref{assum:overlap}), $\UQPPs{\expf}{\epsilon}$ in Theorem~\ref{thm:qpp converge} satisfies for all input evolution $U_H = e^{-iH}$,
\begin{equation}
    \gamma^2 \alpha^2 e^{-2\tau(\lambda_0 + \lambda)} - 2\epsilon 
    \leq \norm{V\ket{\phi}}^2 
    \leq \alpha^2 \left(e^{-\tau\lambda} \norm{e^{-\tau H}\ket{\phi}} \right)^2 +  2\alpha \epsilon \left(e^{-\tau\lambda / 2} \norm{e^{-\tau H / 2}\ket{\phi}} \right)^2 + \epsilon^2
,\end{equation}
    where $V = \left(\bra{0} \ox I_n\right) \UQPPs{\expf}{\epsilon}(U_H) \left(\ket{0} \ox I_n\right)$.
\end{lemma}
\begin{proof}
    Assumptions~(\nref{assum:normalize},\nref{assum:overlap}) is here to guarantee non-trivial existences for $V$ and $\gamma$. We have
\begin{align}
    V \ket{\phi} &= \sum_j F(-\lambda_j) \ket{\psi_j} \textrm{, with } \norm{V \ket{\phi}}^2 = \sum_{j} \abs{c_j}^2  F(-\lambda_j)^2
.\end{align}
    Equation~\eqref{eqn:qpp approx} provides $\abs{\expf(x)} - \epsilon \leq \real{F(x)}$ and $\abs{F(x)} \leq \min\set{\abs{\expf(x)} + \epsilon, 1}$ for all $x \in [-1, 1]$. One can derive 
\begin{equation}
    \norm{V \ket{\phi}}^2 
    = \sum_j \abs{c_j}^2 \abs{F(-\lambda_j)}^2
    \leq \sum_{j: -\lambda_j \leq \lambda} \abs{c_j}^2 \left(\abs{\expf(-\lambda_j)} + \epsilon\right)^2 + \sum_{j: -\lambda_j > \lambda} \abs{c_j}^2
\end{equation}
Since $\lambda \geq -\lambda_0$, this inequality becomes 
\begin{align}
    \norm{V \ket{\phi}}^2 
    &\leq \sum_{j} \abs{c_j}^2 \left(\abs{\expf(-\lambda_j)} + \epsilon\right)^2 \\
    &= \sum_{j} \abs{c_j}^2 \left( \abs{\expf(-\lambda_j)}^2 + 2 \epsilon \abs{\expf(-\lambda_j)} + \epsilon^2 \right) \\
    &= \alpha^2 e^{-2\tau\lambda} \sum_{j} \abs{c_j}^2 e^{2\tau\lambda_j} + 2 \alpha  \epsilon \sum_{j} \abs{c_j}^2 e^{\tau\lambda_j} + \epsilon^2 \\
    &= \alpha^2 \left(e^{-\tau\lambda} \norm{e^{-\tau H}\ket{\phi}} \right)^2 + 2 \alpha \epsilon \left(e^{-\tau\lambda / 2} \norm{e^{-\tau H / 2}\ket{\phi}} \right)^2 + \epsilon^2
.\end{align}
Similarly, we have
\begin{align}~\label{eqn:prob lower bound}
    \norm{V \ket{\phi}}^2
    &\geq \sum_{j} \abs{c_j}^2 \left(\abs{\expf(-\lambda_j)} - \epsilon\right)^2 
    \geq \abs{c_0}^2 \left(\abs{\expf(-\lambda_0)} - \epsilon\right)^2 \\
    &\geq \gamma^2 \left(\abs{\expf(-\lambda_0)}^2 - 2\epsilon \abs{\expf(-\lambda_0)}\right) \\
    &\geq \gamma^2 \abs{\expf(-\lambda_0)}^2 - 2\epsilon
    = \gamma^2 \alpha^2 e^{-2\tau(\lambda_0 + \lambda)} - 2\epsilon
.\end{align}
\end{proof}

\renewcommand\thelemma{\ref{lem:overlap lower bound}}
\begin{lemma}
    Under Assumptions~(\nref{assum:overlap},\nref{assum:nondegenerate}), 
\begin{equation}
    \abs{\braket{\psi_0}{\ite}} \geq {\gamma}/{\sqrt{e^{-2\tau \Delta}(1 - \gamma^2) + \gamma^2}}
.\end{equation}
    Moreover, the lower bound is tight for some Hamiltonians $H$.
\end{lemma}
\renewcommand{\thelemma}{S\arabic{proposition}}
\begin{proof}
    Suppose $c_0$ is a positive real number without loss of generality. Then $c_0 = \gamma$.
    One can observe that
\begin{align}
    e^{-\tau H}\ket{\phi} &= \gamma e^{-\tau \lambda_{0}} \ket{\psi_0} + \sum_{j > 0} c_j e^{-\tau \lambda_j} \ket{\psi_j}, \\
    \norm{e^{-\tau H}\ket{\phi}}^2 &= \gamma^2 e^{-2\tau \lambda_0} + \sum_{j > 0} \abs{c_j}^2 e^{-2\tau \lambda_j} 
.\end{align}
    Under Assumption~(\nref{assum:nondegenerate}), $\Delta > 0$.
    Since $\lambda_j \geq \lambda_0 + \Delta$ for all $j > 0$, we have $e^{-2\tau \lambda_j} \leq e^{-2\tau (\lambda_0 + \Delta)}$ and hence
\begin{align}~\label{eqn:ite norm upper bound}
    \norm{e^{-\tau H}\ket{\phi}}^2 &\leq \gamma^2 e^{-2\tau \lambda_0} + e^{-2\tau (\lambda_0 + \Delta)} \sum_{j > 0} \abs{c_j}^2 \\
    &= \gamma^2 e^{-2\tau \lambda_0} \left( 1 + e^{-2\tau \Delta} (1 - \gamma^2) / \gamma^2 \right)
.\end{align}
    Substituting Equation~\eqref{eqn:ite norm upper bound} back into $\braket{\psi_0}{\ite}$ gives
\begin{align}
    \abs{\braket{\psi_0}{\ite}} &= \frac{1}{\norm{e^{-\tau H}\ket{\phi}}} \cdot \abs{\braandket{\psi_0}{e^{-\tau H}}{\phi}} 
    = \frac{1}{\norm{e^{-\tau H}\ket{\phi}}} \cdot \gamma e^{-\tau \lambda_0} \\
    &\geq \left( 1 + e^{-2\tau \Delta} \frac{1 - \gamma^2}{\gamma^2} \right)^{-1} 
    = \frac{\gamma}{\sqrt{e^{-2\tau \Delta} (1 - \gamma^2) + \gamma^2}}
.\end{align}
    To make the lower bound as tight, simply choose some $H$ satisfying all eigenvalues are equal except for $\lambda_0$.
\end{proof}

\subsection{Proof of Lemma~\ref{lem:succ prob low bound}, Theorem~\ref{thm:ite uH} and Corollary~\ref{coro:ite short}}

\renewcommand\thelemma{\ref{lem:succ prob low bound}}
\begin{lemma}
    Let $C \geq \tau(\lambda - \abs{\lambda_0}) \geq 0$. Under Assumptions~(\nref{assum:normalize},\nref{assum:long evolution},\nref{assum:overlap}), the output state $\ket{\widetilde{\phi}(\tau)}$ from the ITE circuit $\UQPPs{\expf}{\epsilon}(U_H)$ is obtained with success probability lower bounded by $\alpha^2 \gamma^2 e^{-2C} - \epsilon$. Moreover, the state fidelity between the output state and the ITE state is approximately lower bounded as
\begin{equation}
    \abs{\braket{\ite}{\widetilde{\phi}(\tau)}} \gtrsim 1 - \BigO{ \alpha^{-1}\epsilon \cdot e^{C}}
.\end{equation}
\end{lemma}
\renewcommand{\thelemma}{S\arabic{proposition}}
\begin{proof}
    Under Assumptions~(\nref{assum:normalize},\nref{assum:overlap}), the statement for probability lower bound is a direct implication of Lemma~\ref{lem:qpp norm}, as $\norm{V\ket{\phi}}^2$ is the success probability of post selection. The rest of the problem is to prove the fidelity lower bound. 
    
    For convenience, denote $V = F(U_H)$ and the output state $\ket{\widetilde{\phi}(\tau)}$ by calling the input state $\ket{0} \ox \ket{\phi}$ to $\UQPPs{\expf}{\epsilon}(U_H)$ and making the post-selection of ancilla qubit to be 0. Recall $\norm{e^{-\tau H}\ket{\phi}}^2 = \sum_j \abs{c_j^2} e^{-2\tau \lambda_j}$, By Lemma~\ref{lem:qpp norm}, we have
\begin{align}
    \norm{V \ket{\phi}} &\leq \sqrt{\alpha^2 \left(e^{-\tau\lambda} \norm{e^{-\tau H}\ket{\phi}} \right)^2 +  2\alpha \epsilon \left(e^{-\tau\lambda / 2} \norm{e^{-\tau H / 2}\ket{\phi}} \right)^2 + \epsilon^2}
.\end{align}
    Similarly, we have
\begin{align}
    \abs{\braandket{\phi}{V e^{-\tau H}}{\phi}} 
    &= \abs{\sum_j \abs{c_j}^2 F(-\lambda_j) e^{-\tau \lambda_j}}
    \geq \sum_j \abs{c_j}^2 \left( \expf(-\lambda_j) - \epsilon \right) e^{-\tau \lambda_j} \\
    &= \sum_j \abs{c_j}^2 \left( \alpha e^{\tau(-\lambda_j - \lambda)} - \epsilon \right) e^{-\tau \lambda_j} \\
    &= \alpha e^{-\tau \lambda} \sum_j \abs{c_j}^2 e^{-2\tau\lambda_j} - \epsilon \sum_j \abs{c_j}^2 e^{-\tau \lambda_j} \\
    &= \alpha e^{-\tau\lambda} \norm{e^{-\tau H}\ket{\phi}}^2 - \epsilon \norm{e^{-\tau H/2}\ket{\phi}}^2 \\
    \implies \frac{\abs{\braandket{\phi}{V e^{-\tau H}}{\phi}}}{\norm{e^{-\tau H} \ket{\phi}}} &= \alpha e^{-\tau\lambda} \norm{e^{-\tau H}\ket{\phi}} - \epsilon \norm{e^{-\tau H/2}\ket{\phi}}^2 / \norm{e^{-\tau H} \ket{\phi}}
.\end{align}
    These results together imply
\begin{align}
     \abs{\braket{\ite}{\widetilde{\phi}(\tau)}} &= \frac{\abs{\braandket{\phi}{V e^{-\tau H}}{\phi}}}{\norm{V \ket{\phi}}\cdot\norm{e^{-\tau H} \ket{\phi}}} \\
    &\geq \frac{ \alpha e^{-\tau\lambda} \norm{e^{-\tau H}\ket{\phi}} - \epsilon \norm{e^{-\tau H/2}\ket{\phi}}^2 / \norm{e^{-\tau H} \ket{\phi}} }{ \sqrt{\alpha^2 \left(e^{-\tau\lambda} \norm{e^{-\tau H}\ket{\phi}} \right)^2 + 2 \alpha \epsilon \left(e^{-\tau\lambda / 2} \norm{e^{-\tau H / 2}\ket{\phi}} \right)^2 + \epsilon^2} } \\
    &= \frac{ 1 - \epsilon \cdot e^{\tau \lambda} \norm{e^{-\tau H/2}\ket{\phi}}^2 / \norm{e^{-\tau H} \ket{\phi}}^2 }{ \sqrt{1 + 2 \epsilon / \alpha \cdot e^{\tau\lambda} \norm{e^{-\tau H/2}\ket{\phi}}^2 / \norm{e^{-\tau H} \ket{\phi}}^2 + \epsilon^2 / \alpha^2 \cdot e^{2\tau\lambda} / \norm{e^{-\tau H} \ket{\phi}}^2 } } \\
    &= \frac{1 - a(\tau) / \alpha}{\sqrt{1 + 2 a(\tau) / \alpha + b(\tau) / \alpha^2}}
,\end{align}
    where $a(\tau) = \epsilon \cdot e^{\tau\lambda} \norm{e^{-\tau H/2}\ket{\phi}}^2 / \norm{e^{-\tau H} \ket{\phi}}^2$ 
    and $b(\tau) = \epsilon^2 \cdot e^{2\tau\lambda} / \norm{e^{-\tau H} \ket{\phi}}^2$. Here, by Assumption~(\nref{assum:long evolution}), we consider $\norm{e^{-\tau H/2}\ket{\phi}}^2 = \BigO{\gamma^2 e^{-\tau\lambda_0}}$ and $\norm{e^{-\tau H}\ket{\phi}}^2 =\BigO{ \gamma^2 e^{-2\tau\lambda_0}}$. Then one can derive
\begin{equation}
    a(\tau) = \epsilon \cdot \BigO{e^{\tau(\lambda + \lambda_0)}}
    \textrm{, \,\,} b(\tau) = \epsilon^2 \cdot \BigO{e^{2\tau(\lambda + \lambda_0)}} = \BigO{a(\tau)^2}
,\end{equation}
    which gives
\begin{equation}
    \abs{\braket{\ite}{\widetilde{\phi}(\tau)}} 
    \gtrsim \frac{1 - a(\tau) / \alpha}{\sqrt{1 + 2 a(\tau) / \alpha + a(\tau)^2 / \alpha^2}} 
    = \frac{1 - a(\tau) / \alpha}{1 + a(\tau) / \alpha}
    = 1 - \BigO{a(\tau) / \alpha}
.\end{equation}
    By Assumption~(\nref{assum:normalize}), $\lambda + \lambda_0 = \lambda - \abs{\lambda_0}$. Since $C \geq \tau (\lambda - \abs{\lambda_0})$ and hence $e^{C} \geq e^{\tau(\lambda + \lambda_0)}$, substituting $a(\tau) \leq \BigO{\epsilon \cdot e^{C}}$ gives the desired result.
\end{proof}

\renewcommand\thetheorem{\ref{thm:ite uH}}
\begin{theorem}
    Under Assumptions~(\nref{assum:normalize},\nref{assum:oracle},\nref{assum:long evolution},\nref{assum:overlap},\nref{assum:reprod}), one can prepare the ITE state $\ket{\ite}$  up to fidelity $1 - \BigO{\poly(\tau^{-1})}$, with probability 1, using the following cost:
\begin{enumerate}[leftmargin=1em]
    \item [-] $\BigTO{ \gamma^{-2} \tau}$ queries to controlled-$U_H$ and its inverse,
    \item [-] $\BigO{\gamma^{-2}}$ copies of $\ket{\phi}$, 
    \item [-] $\BigTO{\tau}$ maximal query depth of $U_H$, and
    \item [-] one ancilla qubit initialized in the zero state.
\end{enumerate}
\end{theorem}
\renewcommand{\thetheorem}{S\arabic{proposition}}
\begin{proof}
    Such state preparation can be done by two parts: a rough estimation of $\abs{\lambda_0}$ (QPE part) and a simulation of the exponential function (ITE part).
    
    On the one hand, by Theorem~\ref{thm:qpe}, one can obtain a value in interval $[\abs{\lambda_0} - \tau^{-1} / 2, \abs{\lambda_0} + \tau^{-1} / 2]$ with failure probability $e^{-\tau}$ using $\BigO{\gamma^{-2}\tau \log \left(\tau\log(\gamma^{-2}e^{\tau}) \right)}$ queries to controlled-$U_H$ and its inverse, $\BigO{\gamma^{-2}}$ copies of $\ket{\phi}$ and $\BigO{\tau \log \left(\tau \log(\gamma^{-2}e^{\tau}) \right)}$ maximal query depth of $U$. By adding this value by $\tau^{-1} / 2$, we obtain an estimation $\lambda \in [\abs{\lambda_0}, \abs{\lambda_0} + \tau^{-1}]$.

    On the other hand, such $\lambda$ gives $C$ in Lemma~\ref{lem:succ prob low bound} can be $1$. Taking $\epsilon = \BigO{\poly(\tau^{-1})}$, under Assumptions~(\nref{assum:normalize}, \nref{assum:long evolution}, \nref{assum:overlap}), Lemma~\ref{lem:succ prob low bound} guarantees that $\UQPPs{\expf}{\epsilon}(U_H)$ output the ITE state $\ket{\ite}$ with fidelity $1 - \BigO{\poly(\tau^{-1})}$, while the probability of post-selection is lower bounded by $\Omega(\gamma^2) - \BigO{\poly(\tau^{-1})}$, where $\alpha \geq e^{-1/2}$ is considered as a constant. Then one needs to execute the circuit $\UQPPs{\expf}{\epsilon}(U_H)$ in $\BigO{\gamma^{-2}}$ times to obtain an approximated ITE state. 

    Theorem~\ref{thm:qpp converge} states that with 1 ancilla qubit, the circuit $\UQPPs{\expf}{\epsilon}(U_H)$ can be constructed by querying $\BigTO{\tau}$ times of controlled-$U_H$ and its inverse, and so is the query depth of $U_H$. Combining the QPE part and the ITE part, the total resource cost is summarized as follows:
\begin{enumerate}[leftmargin=1em]
    \item [-] queries to controlled-$U_H$ and its inverse: $\BigO{\gamma^{-2}\tau \log \left(\tau\log(\gamma^{-2}e^{\tau}) \right)} + \BigTO{\gamma^{-2} \tau} = \BigTO{\gamma^{-2}\tau}$
    \item [-] copies of $\ket{\phi}$: $\BigO{\gamma^{-2}} + \BigO{\gamma^{-2}} = \BigO{\gamma^{-2}}$
    \item [-] maximal query depth of $U_H$: $\BigO{\tau \log \left(\tau \log(\gamma^{-2}e^{\tau}) \right)} + \BigTO{\tau} = \BigTO{\tau}$
\end{enumerate}
\end{proof}

We also analyze the fragmented case for simulating a long evolution composed of relatively short imaginary-time evolutions (so that Assumption~(\nref{assum:long evolution}) does not hold for such short time). We first start with a short lemma.

\begin{lemma}~\label{lem:error amplify}
    Suppose $\tau, \lambda > 0$ and $N$ is a positive integer. Let $f_1$, $f_2$ be two functions mapping from $[-1, \lambda]$ to $\bar{\mathbb{D}}$ such that $\norm{f_1 - e^{\tau (x - \lambda)}}_{[-1, \lambda]} \leq \epsilon_1$ and $\norm{f_2 - e^{N \tau (x - \lambda)}}_{[-1, \lambda]} \leq \epsilon_2$. Then
\begin{equation}
    \norm{f_1^N - f_2}_{[-1, \lambda]} \leq N\epsilon_1 + \epsilon_2
.\end{equation}
\end{lemma}
\begin{proof}
Define $g(x)=e^{\tau(x-\lambda)}$ on $[-1,\lambda]$. Since $x-\lambda \leq 0$ and $\tau>0$, we have $\abs{g(x)} \leq 1$. Also $\abs{f_1(x)} \leq 1$ because $f_1$ maps into $\overline{\mathbb{D}}$.
For each $x\in[-1,\lambda]$, 
$f_1(x)^N-g(x)^N = (f_1(x)-g(x))\sum_{k=0}^{N-1} f_1(x)^{N-1-k} g(x)^k$, hence
\begin{equation}
    \abs{ f_1(x)^N-g(x)^N }
    \leq \abs{f_1(x)-g(x)} \sum_{k=0}^{N-1} \abs{f_1(x)}^{N-1-k} \abs{g(x)}^{k}
    \leq N \abs{f_1(x)-g(x)}
.\end{equation}
Taking the supremum over $x \in [-1,\lambda]$ gives
$\norm{f_1^N-g^N}_{[-1,\lambda]} \leq N \norm{f_1-g}_{[-1,\lambda]} \leq N\epsilon_1$.
Since $g(x)^N = e^{N\tau(x-\lambda)}$, the triangle inequality yields
$\norm{f_1^N-f_2}_{[-1,\lambda]}
\leq \norm{f_1^N-g^N}_{[-1,\lambda]} + \norm{g^N-f_2}_{[-1,\lambda]}
\leq N\epsilon_1 + \epsilon_2$, as required.
\end{proof}

\renewcommand\thecorollary{\ref{coro:ite short}}
\begin{corollary}
    Let $N$ be a positive integer such that $t = N\tau$ satisfies Assumption~(\nref{assum:long evolution}). Suppose $\lambda \in [\abs{\lambda_0}, \abs{\lambda_0} + t^{-1}]$. Under Assumptions~(\nref{assum:normalize},\nref{assum:overlap}),
    $\UITE{\tau}{H}$ in Equation~\eqref{eqn:qpp post-select} satisfies
\begin{equation}
    \abs{ \bra{\ite}\UITE{\tau}{H} [\ket{\phi}]} \gtrsim 1 - \BigO{\alpha^{-1} \epsilon}
.\end{equation}
    Moreover, under Assumption~(\nref{assum:oracle}), $\UITE{\tau}{H}$ can be implemented with success probability lower bounded by $\alpha^2 \left(\gamma^2 e^{-2/N} + (1 - \gamma^2)e^{-2(2t + 1)/N}\right) - \epsilon$.
\end{corollary}
\renewcommand{\thecorollary}{S\arabic{proposition}}
\begin{proof}
    The proof is done by comparing the norm difference among $\UITE{\tau}{H}^{\circ N} [\ket{\phi}]$, $\UITE{t}{H}[\ket{\phi}]$ and $\ket{\phi(t)}$. Suppose the approximation error in $\UITE{t}{H}$ is $N\epsilon$. By Lemma~\ref{lem:error amplify}, $f_{\tau, \lambda}$ and $f_{t, \lambda}$ satisfy $\norm{f_{\tau, \lambda}^N - f_{t, \lambda}}_{[-1, \lambda]} \leq 2N\epsilon$ and hence
\begin{equation}
    \norm{\UITE{\tau}{H}^{\circ N} [\ket{\phi}] - \UITE{t}{H}[\ket{\phi}]} \leq 2N\epsilon
.\end{equation}
    Since $t$ satisfies Assumption~(\nref{assum:long evolution}) and $\lambda \in [\abs{\lambda_0}, \abs{\lambda_0} + t^{-1}]$, Lemma~\ref{lem:succ prob low bound} implies $\norm{\UITE{t}{H}[\ket{\phi}] - \ket{\phi(t)}} \leq \BigO{N \alpha^{-1} \epsilon}$. Then by triangle inequality, 
\begin{equation}
    \norm{\UITE{\tau}{H}^{\circ N} [\ket{\phi}] - \ket{\phi(t)}} 
    \leq \norm{\UITE{\tau}{H}^{\circ N} [\ket{\phi}] - \UITE{t}{H}[\ket{\phi}]} + \norm{\UITE{t}{H}[\ket{\phi}] - \ket{\phi(t)}]} 
    = \BigO{N \alpha^{-1} \epsilon}
\end{equation}
    or equivalently, as $\ket{\phi(t)}$ is driven by $e^{-Ht}$ and $\UITE{\tau}{H}$ is simulating a trigonometric polynomial $F$,
\begin{equation}
    \norm{\frac{ F(U_H)^N \ket{\phi} }{ \norm{F(U_H)^N \ket{\phi}} }  - \frac{ (e^{-H\tau})^N \ket{\phi} }{ \norm{(e^{-H\tau})^N \ket{\phi}} }} \leq \BigO{N\epsilon}
.\end{equation}
    Together with $\norm{F(U_H) - \alpha e^{\tau(-H - \lambda I)}}_\infty \leq \epsilon$, this gives the desired error statement.
    
    As for probability lower bound, we continue from Equation~\eqref{eqn:prob lower bound} in Lemma~\ref{lem:qpp norm}, such that
\begin{align}
    \norm{V \ket{\phi}}^2
    &\geq \sum_{j} \abs{c_j}^2 \abs{\expf(-\lambda_j)}^2 - \epsilon \\
    &= \alpha^2 \Big( \gamma^2 e^{-2\tau(\lambda_0 + \lambda)} + \sum_{j \neq 0} \abs{c_j}^2 e^{-2\tau(\lambda_j + \lambda)} \Big)  - \epsilon \\
    &\geq \alpha^2 \Big( \gamma^2 e^{-2\tau(\lambda_0 + \lambda)} + \sum_{j \neq 0} \abs{c_j}^2 e^{-2\tau(1 + (1 + 1/t))} \Big)  - \epsilon \\
    &= \alpha^2 \left( \gamma^2 e^{-2\tau/t} + (1 - \gamma^2) e^{-2\tau(2 + 1/t)} \right)  - \epsilon \\
    &= \alpha^2 \left( \gamma^2 e^{-2/N} + (1 - \gamma^2) e^{-2(2t + 1)/N} \right)  - \epsilon
,\end{align}
    as required.
\end{proof}

\subsection{Resource analysis for Trotter case}~\label{appendix:trotter}

In this section, we analyze the resource complexity when $U_H$ is now realized by its Trotter decomposition.
It is hard to implement $U(t)$ directly, so generally Hamiltonians of interest will be written as the sum of $L$ Pauli matrices:
\begin{equation}
    U(t)=\exp(t H)=\exp\left(t\sum_{j=1}^{L}h_j\sigma_j\right)
.\end{equation}

Consider a system Hamiltonian $H$ that is decomposed into a sum of polynomially many Hermitian terms $\sigma_j$, each of which is a tensor product of Pauli operators. Specifically, we have $H = \sum_{j=1}^{L} h_j \sigma_j$, where the $\sigma_j$ are constructed as tensor products of Pauli operators.
Its time evolution can be described by the unitary $U = e^{it\sum_{j=1}^Lh_j \sigma_j}$. The goal of the Hamiltonian simulation is to find an efficient circuit construction for this unitary.

One of the leading approaches is the product formula of the Trotter formula,
\begin{align}
    V(t)&= \prod_{j=1}^{L}e^{ith_j \sigma_j}
\end{align}
and each individual operator $V_j (t) = e^{it h_j \sigma_j }$ can be efficiently implemented by a quantum circuit. $\prod_{j=1}^LV_j(t) =V(t)=U$ if all the terms are commute, but in most cases, this condition does not hold. $(V(t/N))^N$ approximates $U$ for large $N$ even if some terms are not commute. This algorithm is referred as the first-order approximation. $V(t/N)$ is called one \emph{Trotterization} step and the circuit has $N$ such repetitions.

$V(t/N)$ is the first-order Suzuki formula and can be written as $S_1(t/N)=V(t/N)$. The complexity of quantum simulation can be improved by using higher order Suzuki formula. The $2k$th-order Suzuki formula $S_{2k}$ is defined as below:
\begin{equation}
    S_2(t) = \prod_{j=1}^{L}\exp({\frac{t}{2}ih_j \sigma_j})\prod_{j=L}^{1}\exp({\frac{t}{2}ih_j \sigma_j})
\end{equation}

\begin{equation}
    S_{2k}(t)=\lceil s_{2k-2}(p_kt)\rceil^2 s_{2k-2}((1-4p_k)t) \lceil s_{2k-2}(p_kt)\rceil^2
\end{equation}
with $p_k=1/(4-4^{1/(2k-1)})$ for $k>1$. Although higher-order Suzuki formulas can achieve smaller errors, the first-order formula already performs sufficiently well and is intuitive and easy to understand. In practical applications, the first-order or second-order forms are primarily used.

\begin{theorem}[1st-order analytic bound, \cite{product2018general}]~\label{thm:1st order product formula}
Let $N \in \NN$ and $t \in \RR$. Let $H$ be the Hamiltonian, and $\Lambda \coloneqq \max\set{\abs{h_j}}$. The first-order formula's upper bound is:
\begin{align}
    \left\| \exp\left(-it\sum_{j=1}^{L} h_j\sigma_j\right) - 
    \left[\prod_{j=1}^{L}\exp\left(-\frac{it}{N}h_j\sigma_j\right)\right]^N \right\|_\infty
    \leq \frac{(L\Lambda t)^2}{N}\exp{(\frac{L\Lambda t}{N})}
.\end{align}
\end{theorem}

Now we can proceed to the Proof of Theorem~\ref{thm:ite trotter}.

\begin{theorem}[Davis–Kahan Theorem~\cite{davis1970rotation} for ground states]~\label{thm:davis kahan}
   Let $H, \Happrox$ be Hermitian matrices acting on a finite-dimensional Hilbert space, and $\Delta$ be the spectral gap between the smallest eigenvalue $\lambda_0$ and the second smallest eigenvalue. If $\norm{H - \Happrox}_\infty < \epsilon$ for some $\epsilon < \Delta$, then
\begin{equation}
    \norm{P - P'}_\infty < \epsilon / \Delta
,\end{equation}
    where $P$ is the spectral projector onto the $\lambda_0$-eigenspace of $H$, and $P'$ is the spectral projector of $\Happrox$ onto the cluster of eigenvalues that lie in $[\lambda_0 - \epsilon, \lambda_0 + \epsilon]$.
\end{theorem}

\begin{proposition}~\label{prop:evolution error}
Let $\epsilon < \Delta/2$. Suppose $\Happrox$ is a Hamiltonian satisfying $\norm{\exp(-iH) - \exp(-i\Happrox)}_\infty < \epsilon$, and $\ket{\iteapprox}$ is the normalized imaginary-time evolution under $\Happrox$. Under Assumptions~(\nref{assum:normalize}, \nref{assum:long evolution}, \nref{assum:overlap}, \nref{assum:nondegenerate}, \nref{assum:long evolution gap}), we have
\begin{equation}\label{eqn:main_bound}
    \norm{\ket{\ite} - \ket{\iteapprox}} 
    \leq \sqrt{2}\frac{\epsilon}{\Delta} + \BigO{e^{-\tau \Delta}}.
\end{equation}
\end{proposition}
\begin{proof}
Let $\ket{\psi_0}$ and $\ket{\psi^\approx_0}$ be the ground states of $H$ and $\Happrox$, respectively. By the triangle inequality,
\begin{equation}\label{eqn:tri_bound}
    \norm{\ket{\ite} - \ket{\iteapprox}}
    \leq \norm{\ket{\ite} - \ket{\psi_0}}
    + \norm{\ket{\psi_0} - \ket{\psi^\approx_0}}
    + \norm{\ket{\psi^\approx_0} - \ket{\iteapprox}}.
\end{equation}
Under the assumptions on the spectral gaps and the overlap of the initial state, Lemma~\ref{lem:overlap lower bound} implies both 
\begin{equation}
\norm{\ket{\ite} - \ket{\psi_0}} \quad\text{and}\quad \norm{\ket{\psi^\approx_0} - \ket{\iteapprox}}
\end{equation}
decay exponentially in $\tau$ (i.e., they are bounded by $\BigO{e^{-\tau\Delta}}$. Hence, we may write
\begin{equation}\label{eqn:phi_bound}
    \norm{\ket{\ite} - \ket{\iteapprox}} \leq \norm{\ket{\psi_0} - \ket{\psi^\approx_0}} + \BigO{e^{-\tau\Delta}}.
\end{equation}
Next, note that the eigenvalues of $H$ lie in $[-1,1]$, so the exponential map is invertible in this region. Consequently, the operator norm of the difference of the real-time evolutions implies that $\norm{H - \Happrox}_\infty < \epsilon$.
By applying Theorem~\ref{thm:davis kahan} and using the inequality
\begin{equation}\label{eqn:state_diff}
    \norm{\ket{\psi_0} - \ket{\psi^\approx_0}} \leq \sqrt{2}\,\norm{\ketbra{\psi_0}{\psi_0} - \ketbra{\psi^\approx_0}{\psi^\approx_0}}_\infty
,\end{equation}
we bound the difference between the ground states by
\begin{equation}\label{eqn:ground_bound}
    \norm{\ket{\psi_0} - \ket{\psi^\approx_0}} \leq \sqrt{2}\,\frac{\norm{H - \Happrox}_\infty}{\Delta} \leq \sqrt{2}\frac{\epsilon}{\Delta}
.\end{equation}
Then substituting Equation~\eqref{eqn:ground_bound} into Equation~\eqref{eqn:phi_bound} yields the desired result.
\end{proof}

\renewcommand\thetheorem{\ref{thm:ite trotter}}
\begin{theorem}
    Under Assumptions~(\nref{assum:normalize},\nref{assum:pauli},\nref{assum:long evolution},\nref{assum:overlap},\nref{assum:reprod},\nref{assum:good overlap},\nref{assum:nondegenerate},\nref{assum:long evolution gap}), one can prepare the ITE state $\ket{\ite}$ up to fidelity $1 - \BigO{L^2\Lambda^2\poly(\tau^{-1})}$, using the following cost:
\begin{enumerate}[leftmargin=1em]
    \item [-] $\BigTO{L \poly(n\tau)}$ queries to controlled Pauli rotations,
    \item [-] $\BigO{\poly(n)}$ copies of $\ket{\phi}$, 
    \item [-] $\BigTO{L \poly(\tau)}$ maximal query depth, and
    \item [-] one ancilla qubit initialized in the zero state,
\end{enumerate}
    where $L$ is the number of Pauli terms and $\Lambda = \max_j\abs{h_j}$.
\end{theorem}
\renewcommand{\thetheorem}{S\arabic{proposition}}
\begin{proof}
    We first consider the first-order Trotter-Suzuki decomposition of $\exp(-iH)$ with number of Trotter steps $N$, such that by Theorem~\ref{thm:1st order product formula}, $U = \left[ \prod_{j=1}^{L}\exp\!\left(-i h_j\sigma_j / N\right) \right]^N$ satisfies
\begin{equation}
    \norm{U - \exp(-iH)}_\infty \leq \exp({L\Lambda}/{N}) \cdot (L\Lambda)^2 / N
.\end{equation}
    Let $\epsilon$ be the simulation error and $\Happrox$ be the Hamiltonian such that $U = \exp(-i\Happrox)$. 
    On the one hand, by Theorem~\ref{thm:ite uH}, there exists a quantum algorithm that can prepare the state $\ket{\widetilde{\phi}(\tau)}$ up to fidelity $\BigO{\poly(\tau^{-1})}$ such that
\begin{equation}
    \norm{\ket{\widetilde{\phi}(\tau)} - \ket{\iteapprox}} 
    = 2 - 2\real{\braket{\ite}{\iteapprox}} \leq \BigO{\poly(\tau^{-1})}
;\end{equation}
    on the other hand, Proposition~\ref{prop:evolution error} implies that, the norm difference between $\ket{\ite}$ and $\ket{\iteapprox}$ is bounded as $\norm{\ket{\ite} - \ket{\iteapprox}} \leq \sqrt{2}\epsilon / \Delta + \BigO{e^{-\tau \Delta}}$. These two inequalities together gives
\begin{align}
    \norm{\ket{\widetilde{\phi}(\tau)} - \ket{\ite}} 
    &\leq \norm{\ket{\widetilde{\phi}(\tau)} - \ket{\iteapprox}}  + \norm{\ket{\iteapprox} - \ket{\ite}} \\
    &\leq \BigO{\epsilon / \Delta + e^{-\tau \Delta} + \poly(\tau^{-1})} \\
    &= \BigO{\exp({L\Lambda}/{N}) \cdot (L\Lambda)^2 / N \Delta  + e^{-\tau \Delta} + \poly(\tau^{-1})} \\
    &= \BigO{\exp({L\Lambda}/{N}) \cdot (L\Lambda)^2 / N \Delta  + \poly(\tau^{-1})}
,\end{align}
    where $e^{\tau \Delta} = \BigOmega{\poly(\tau)}$ by Assumption~(\nref{assum:long evolution gap}).
    As for the resource cost, by Assumption~(\nref{assum:good overlap}), we will use $\BigTO{\poly(n) \tau}$ queries to controlled-$U$ and its inverse, $\BigO{\poly(n)}$ copies of $\ket{\phi}$, and one ancilla qubit initialized in the zero state.
    Note that each controlled-$U$ (or controlled-$U^\dag$) requires $LN$ calls of controlled-Pauli gates. Therefore, the total query number of controlled-Pauli rotations is $\BigTO{LN \cdot \poly(n) \tau }$.
    Finally, choosing $N = \BigO{\poly(\tau)}$ gives the statement that it requires 
    $\BigTO{L \poly(n\tau)}$ number of controlled-Pauli gates and 
    $\BigO{\poly(n)}$ copies of $\ket{\phi}$ to realize $\ket{\ite}$ up to norm distance $\BigO{L^2\Lambda^2  \cdot\poly(\tau^{-1})}$, and hence so is the state infidelity.
\end{proof}

\section{Details and proofs of Algorithm~\ref{alg:ground}}~\label{appendix:lambda location}

The complete version of Algorithm~\ref{alg:ground} is given by Algorithm~\ref{alg:ground complete}.

\begin{algorithm}[ht]\label{alg:ground complete}
\caption{Adaptive Ternary Search (full version)}
\SetKwInOut{Input}{Input}
\SetKwInOut{Output}{Output}
\Input{Hamiltonian $H$, initial state $\ket{\phi}$, step size $\Delta t$, lower bound $\explb$, a boolean function $\cX$ for testing convergence}
\Output{$\tau, \lambda, E$ in Problem~\ref{prob:ground}}

Take an initial guess $t \gg 0$; initialize $E_0 = 0$, $i \gets 0$\;

Initialize search interval endpoints: $\lambda_l \gets 0$, $\lambda_r \gets$ initial upper bound via Algorithm~\ref{alg:binary search}\;

\While{${\lambda_r - \lambda_l} > t^{-1}$ or $\cX(\set{E_i}_i) = \operatorname{False}$}{
    Set measurement shot number $8 L \Lambda^2 t^3 \cdot \explb^{-2}$ for each estimation of $\measureloss{\cdot}$\;

    Set interval width $\delta = (\lambda_r - \lambda_l)/3$ and two trisection points: $\lambda_{lm} \gets \lambda_l + \delta$, $\lambda_{rm} \gets \lambda_r - \delta$\;
    
    Evaluate the estimations $\measureloss{\lambda_{lm}}$, $\measureloss{\lambda_r}$ for $\loss{\lambda_{lm}}$, $\loss{\lambda_r}$ given in Equation~\eqref{eqn:ideal vqe}, respectively\;
    
    Compute relative difference $\reldiff \gets \left({\measureloss{\lambda_{lm}} - \measureloss{\lambda_r}}\right)/{{\measureloss{\lambda_r}}}$ \;

    \uIf{$\abs{\reldiff-(e^{4t\delta}-1)}>\tau^{-1}(e^{4t\delta}+1)$}{
        Obtain $E_i$ from selected samples that estimate  $\measureloss{\lambda_r}$, when the ancilla qubit is measured to be 0\;
        Set $[\lambda_l, \lambda_r] \gets [\lambda_{lm}, \lambda_r]$\;
    }\Else{
        Obtain $E_i$ from selected samples that estimate  $\measureloss{\lambda_{lm}}, \measureloss{\lambda_r}$, when the ancilla qubit is measured to be 0\;
        Set $[\lambda_l, \lambda_r] \gets [\lambda_l, \lambda_{rm}]$\;
    }
    Update $t \gets t + \Delta t$, $i \gets i + 1$\;
}
\Return $\tau \gets t$, $\lambda_r$, $E_i$\;
\end{algorithm}

\renewcommand\thetheorem{\ref{thm:ground}}
\begin{theorem}
Suppose Assumptions~(\nref{assum:normalize},\nref{assum:oracle},\nref{assum:overlap},\nref{assum:reprod},\nref{assum:good overlap},\nref{assum:nondegenerate},\nref{assum:priori}) hold. 
Algorithm~\ref{alg:ground} returns a time $\tau$ that satisfies Assumption~(\nref{assum:long evolution gap}), an estimate $\lambda \in [\abs{\lambda_0}, \abs{\lambda_0}+\tau^{-1}]$, and an estimate of $\lambda_0$ within precision $\BigO{B \gamma^{-1}  \tau^{-1}}$, with failure probability $\BigO{e^{-\tau}\log\tau}$.
Moreover, there are at most $\BigO{L\log\tau}$ distinct circuit constructed in Algorithm~\ref{alg:ground}, and each circuit takes at most:
\begin{enumerate}[leftmargin=1em]
    \item [-] $\BigTO{\tau}$ queries to controlled-$U_H$ and its inverse, 
    \item [-] $\BigTO{\tau}$ query depth of $U_H$, 
    \item [-] 1 ancilla qubit, and
    \item [-] $8 L \Lambda^2 \explb^{-2} \tau^3 $ measurement shots.
\end{enumerate}
\end{theorem}
\renewcommand{\thetheorem}{S\arabic{proposition}}
\begin{proof}
For the number of total iterations, observe that the search interval will decrease by a factor of $2/3$ in each iteration, while parameter $t$ increases linearly within the loop. Also, observe that $\set{E_i}_i$ will converge as long as $\ket{\ite}$ converges to the ground state, i.e., $t$ satisfies Assumption~(\nref{assum:long evolution gap}). Then the number of iterations is at most $\BigO{\log \tau}$, where $\tau$ is the final value of $t$ in Algorithm~\ref{alg:ground}.

For the resource cost analysis, we consider a stricter variant of Algorithm~\ref{alg:ground complete}, presented as Algorithm~\ref{alg:ground strict}, in which the parameter $t$ is fixed to the value $\tau$, i.e., $\Delta t = 0$, and $\tau$ is provided as input. This algorithm inputs $\tau$ that is the output of Algorithm~\ref{alg:ground complete}, but outputs the exact same $\lambda$ and $E$ that Algorithm~\ref{alg:ground complete} would output.
Then an upper bound of the resource cost is obtained as the resource cost of Algorithm~\ref{alg:ground strict}, given by Theorem~\ref{thm:lambda location}. 

\begin{algorithm}[t]\label{alg:ground strict}
\caption{Adaptive Ternary Search (strict version)}
\SetKwInOut{Input}{Input}
\SetKwInOut{Output}{Output}
\Input{Hamiltonian $H$, initial state $\ket{\phi}$, time $\tau$, lower bound $\explb$}
\Output{$\lambda, E$ in Problem~\ref{prob:ground}}

Initialize search interval endpoints: $\lambda_l \gets 0$, $\lambda_r \gets$ initial upper bound via Algorithm~\ref{alg:binary search}\;

Set measurement shot number $8 L \Lambda^2 \tau^3 \cdot \explb^{-2}$ for each estimation of $\measureloss{\cdot}$\;

\While{${\lambda_r - \lambda_l} > \tau^{-1}$}{

    Set interval width $\delta = (\lambda_r - \lambda_l)/3$ and two trisection points: $\lambda_{lm} \gets \lambda_l + \delta$, $\lambda_{rm} \gets \lambda_r - \delta$\;
    
    Evaluate the estimations $\measureloss{\lambda_{lm}}$, $\measureloss{\lambda_r}$ for $\loss{\lambda_{lm}}$, $\loss{\lambda_r}$ given in Equation~\eqref{eqn:ideal vqe}, respectively\;
    
    Compute relative difference $\reldiff \gets \left({\measureloss{\lambda_{lm}} - \measureloss{\lambda_r}}\right)/{{\measureloss{\lambda_r}}}$ \;

    \uIf{$\abs{\reldiff-(e^{4\tau\delta}-1)}>\tau^{-1}(e^{4\tau\delta}+1)$}{
        Obtain $E$ from selected samples that estimate $\measureloss{\lambda_r}$, when the ancilla qubit is measured to be 0\;
        Set $[\lambda_l, \lambda_r] \gets [\lambda_{lm}, \lambda_r]$\;
    }\Else{
        Obtain $E$ from selected samples that estimate  $\measureloss{\lambda_{lm}}, \measureloss{\lambda_r}$, when the ancilla qubit is measured to be 0\;
        Set $[\lambda_l, \lambda_r] \gets [\lambda_l, \lambda_{rm}]$\;
    }
    Update $i \gets i + 1$\;
}
\Return $\lambda_r$, $E$\;
\end{algorithm}

The statement for $\lambda$ follows by Theorem~\ref{thm:lambda location}. The statement for $\tau$ follow by the design of Algorithm~\ref{alg:ground complete}. In the last iteration, $\tau$ already satisfies Assumption~(\nref{assum:long evolution gap}). As for the estimation for $\lambda_0$, Lemma~\ref{lem:succ prob low bound} applies that around $\BigO{\gamma^2}$ proportion of samples for estimating $\loss{\lambda}$ can be used to estimate $\hat{E}(\tau)$. Note that each sample is either $+\Lambda$ or $-\Lambda$, $L$ stands for the number of Pauli terms, and total number of samples for $\hat{E}(\tau)$ is $8 L \Lambda^2 \gamma^2 \tau^3 \explb^{-2}$. Then by Theorem~\ref{thm:hoeffding} (Hoeffding's inequality), an estimation of $\hat{E}(\tau)$ (which is $e^{-\tau \Delta}$-close to $\lambda_0$) is obtained up to measurement additive error $\BigO{B \gamma^{-1}  \tau^{-1}}$ and failure probability $e^{-\tau}$.

The rest of subsections in this section give the proof of Theorem~\ref{thm:lambda location}.
\end{proof}

\subsection{Performance guarantee of sampling measurements}

\begin{lemma}~\label{lem:ideal vqe bound}
    Let $\qpploss{\lambda}$ be the expectation value of the quantum state $\UQPPs{\expf}{\epsilon}(U_H)(\ket{0} \ox \ket{\phi})$ with respect to $\widehat{H} = \ketbra{0}{0} \ox H$,
\begin{equation}~\label{eqn:qpp vqe}
    \qpploss{\lambda} = (\bra{0} \ox \bra{\phi}) \UQPPs{\expf}{\epsilon}(U_H)^\dag \cdot \widehat{H} \cdot \UQPPs{\expf}{\epsilon}(U_H) (\ket{0} \ox \ket{\phi})
.\end{equation} 
    Then under Assumption~(\nref{assum:normalize}), the estimation error of $\loss{\lambda}$ is bounded as $\abs{\loss{\lambda} - \qpploss{\lambda}} \leq 2\epsilon$.
\end{lemma}
\begin{proof}
    By Theorem~\ref{thm:qpp converge}, there exists a trigonometric polynomial $F \in \CC[e^{ix}, e^{-ix}]$ such that $\norm{F}_\infty \leq 1$, $\norm{F - \expf}_{\infty, [-1, 1]} \leq \epsilon$, and hence
\begin{equation}
    \qpploss{\lambda} = \sum_{j} \abs{c_j}^2 \abs{F(-\lambda_j)}^2 \lambda_j
.\end{equation}
    Then we have
\begin{align}
    \abs{\loss{\lambda} - \qpploss{\lambda}} 
    &\leq \sum_j \abs{c_j}^2 \abs{\lambda_j} \cdot \Big| \expf(-\lambda_j)^2 - \abs{F(-\lambda_j)}^2 \Big| \\
    &\leq \max_j \Big| \expf(-\lambda_j)^2 - \abs{F(-\lambda_j)}^2 \Big| \\
    &= \max_j \left( \expf(-\lambda_j) + \abs{F(-\lambda_j)}\right) \cdot \Big| \expf(-\lambda_j) - \abs{F(-\lambda_j)} \Big|
    \leq 2\epsilon
.\end{align}
\end{proof}

\begin{theorem}[Hoeffding's inequality for iid variables, \cite{hoeffding1994probability}]~\label{thm:hoeffding}
    Let $S_n$ be the emperical mean of sampling the random variable $X$ for $n$ times. Then for any $\epsilon > 0$,
\begin{equation}
    \prob{\abs{S_n - \EE[X]} \geq \epsilon} \leq 2 \exp\left( - \frac{2n\epsilon^2}{(x_{\max} - x_{\min})^2} \right)
,\end{equation}
    where $x_{\max}$ and $x_{\min}$ are the maximal and minimal values of $X$, respectively.
\end{theorem}

\begin{algorithm}[H]\label{alg:expectation sample}
\SetKwInOut{Input}{Input}
\SetKwInOut{Output}{Output}
\Input{$M$ copies of the QPP circuit $\UQPPs{\expf}{\epsilon}(U_H)$, input state $\ket{\phi}$}
\Output{Estimates of $\qpploss{\lambda}$}
$S \gets \sum_l \abs{h_l}$. Take $M$ samples from $l \in \{1, \ldots, L\}$ with probability weight $\abs{h_l} / S$\;
$M_l \gets$ number of samples with outcome $l$\;

$i \gets 1$\;
\For{$l$ from $1$ to $L$}{
    Determine $T$ such that $T \sigma_l T^\dag$ is a tensor product of $Z$ and $I$\;  
    Prepare $\ket{\psi} \gets (I \ox T) \cdot \UQPPs{\expf}{\epsilon}(U_H)(\ket{0} \ox \ket{\phi})$\;
    
    \For{$k$\, from $1$ to $M_l$}{
        Measure $\ket{\psi}$ in computational basis to get a bitstring $b_0 b \in \set{0, 1}^{n + 1}$\;
        $X_i \gets (1 - b_0) \cdot \textrm{sign}(h_l) S \cdot \braandket{b}{T \sigma_l T^\dag}{b}$\;
        $i \gets i + 1$\;
    }
}

\KwRet{$\sum_{i} X_i / M$}\;
\caption{Expectation Estimation Protocol for $\loss{\lambda}$}
\end{algorithm}

\begin{lemma}~\label{lem:sample upper bound}
    Algorithm~\ref{alg:expectation sample} needs to prepare $L$ quantum circuits, with each circuit uses $2 L \Lambda^2 \tau / \epsilon^2$ measurement shots in average, to obtain an estimation of the quantity $\qpploss{\lambda}$, up to measurement additive error $\epsilon$ and failure probability $e^{-\tau}$.
\end{lemma}
\begin{proof}
We begin by noting that in Algorithm~\ref{alg:expectation sample}, each recorded value $X_i$ can be seen as an independent sample of a random variable $X$. 
More precisely, when the algorithm selects an index $l$ with probability proportional to $\abs{h_l}/S$ with $S = \sum_l \abs{h_l}$, the corresponding random variable takes the value $(1 - b_0) \textrm{sign}(h_l) S \braandket{b}{T \sigma_l T^\dag}{b}$
with probability $\abs{\braket{\psi}{b_0, b}}^2$. Here, both the operator $T$ and the state $\ket{\psi}$ are determined by the chosen $\sigma_l$.
Using Equation~\eqref{eqn:qpp vqe}, we compute the expectation value as
\begin{align}
    \EE_{l, b_0 b}[X] &= \sum_l \frac{\abs{h_l}}{S} \EE_{b_0 b}[X| l] \\
    &= \sum_l \frac{\abs{h_l}}{S} \sum_{b_0, b} (1 - b_0) \textrm{sign}(h_l) S \braandket{b}{T \sigma_l T^\dag}{b} \cdot \abs{\braket{\psi}{b_0, b}}^2 \\
    &= \sum_l h_l \sum_{b} \braandket{b}{T \sigma_l T^\dag}{b} \cdot \abs{\braket{\psi}{0, b}}^2
    = \sum_l h_l \sum_{b} \bra{\psi} (\ketbra{0}{0} \ox \ketbra{b}{b} {T \sigma_l T^\dag} \ketbra{b}{b}) \ket{\psi} \\
    &= \sum_l h_l \braandket{\psi}{(\ketbra{0}{0} \ox T \sigma_l T^\dag)}{\psi}
    = \sum_l h_l \braandket{0, \phi}{\UQPPs{\expf}{\epsilon}(U_H)^\dag (\ketbra{0}{0} \ox \sigma_l) \UQPPs{\expf}{\epsilon}(U_H)}{0, \phi} \\
    &= \braandket{0, \phi}{\UQPPs{\expf}{\epsilon}(U_H)^\dag (\ketbra{0}{0} \ox H) \UQPPs{\expf}{\epsilon}(U_H)}{0, \phi}
    = \qpploss{\lambda}
.\end{align}
This shows that the random variable $X$ is an unbiased estimator of $\qpploss{\lambda}$.
Further, since every sample satisfies {$\abs{X_i} \leq S \leq L\Lambda$},
the Hoeffding's inequality (Theorem~\ref{thm:hoeffding}) gives
\begin{equation}
    \prob{\abs{\frac{1}{M}\sum_i X_i - \qpploss{\lambda}} \geq \epsilon} \leq 2 \exp\left( - \frac{M\,\epsilon^2}{2 L^2 \Lambda^2} \right)
.\end{equation}
By setting the failure probability to be $e^{-\tau}$, we have $2\exp\left( - \frac{M\,\epsilon^2}{2 L^2 \Lambda^2} \right) \geq e^{-\tau}$.
Taking logarithms and rearranging the terms yields $M \leq {2L^2\Lambda^2\tau}/{\epsilon^2}$.
The total number of quantum circuits created in Algorithm~\ref{alg:expectation sample} is $L$, so each circuit takes $M/L = {2L\Lambda^2\tau}/{\epsilon^2}$ shots in average.
\end{proof}

\begin{theorem}~\label{thm:sample upper bound}
    Under Assumptions~(\nref{assum:normalize},\nref{assum:reprod}), Algorithm~\ref{alg:expectation sample} needs to prepare $L$ quantum circuits, with each circuit uses $8 L \Lambda^2 \explb^{-2} \tau^3 $ measurement shots in average and $\BigTO{\tau}$ queries of controlled-$U_H$ and its inverse, to obtain an estimation of the quantity $\loss{\lambda}$, up to additive error $\tau^{-1} \explb$ and failure probability $e^{-\tau}$.
\end{theorem}
\begin{proof}
    Consider the estimation of $\loss{\lambda}$ with QPP simulation error $\tau^{-1} \explb / 4$ and measurement error $\tau^{-1} \explb / 2$.
    Since $\explb = \BigO{ \poly(\tau^{-1}) }$ as assumed,
    Theorem~\ref{thm:qpp converge} implies that such QPP circuit $\UQPPs{\expf}{\epsilon}(U_H)$ would require $\BigTO{\tau}$ queries of controlled-$U_H$ and controlled-$U_H^\dag$, and can obtain estimation of $\loss{\lambda}$ up to additive error $\tau^{-1} \explb / 2$ by Lemma~\ref{lem:ideal vqe bound}.
    Then Lemma~\ref{lem:sample upper bound} implies the output $\measureloss{\lambda}$ of Algorithm~\ref{alg:expectation sample} satisfies
\begin{equation}
    \abs{\measureloss{\lambda} - \loss{\lambda}}
    \leq \abs{\measureloss{\lambda} - \qpploss{\lambda}} + \abs{\qpploss{\lambda} - \loss{\lambda}} 
    \leq \tau^{-1} \explb
.\end{equation}
\end{proof}

\subsection{Location of the starting point}

The overall idea is to use binary search to locate the region where $\abs{\measureloss{\lambda}} > \explb$, and then use ternary search combined with the Algorithm~\ref{alg:ground} to determine braking.

\begin{algorithm}[H]\label{alg:binary search}
\caption{Binary Search}
\SetKwInOut{Input}{Input}
\SetKwInOut{Output}{Output}
\Input{$\tau, \ket{\phi}, H$ as defined in Section~\ref{sec:qite intro}, lower bound $\explb$ in Assumption~(\nref{assum:priori})}
\Output{A $\lambda$ such that $\measureloss{\lambda}\leq -\explb$ and $\measureloss{\lambda+1/2\tau} > -\explb$.}

Initialize $i \gets 0$, initial guess $[\lambda_l,\lambda_r] \gets [1/\tau, 1 + 1/\tau]$\;
Set the base measurement count $M \gets 8 L \Lambda^2 \tau^3 \cdot \explb^{-2}$ in the estimation of $\measureloss{\cdot}$\;
Estimate $\measureloss{\lambda_r}$\;
\uIf{$\measureloss{\lambda_r} \leq -\explb$}{
    \Return $\lambda_r$\;
}
\uElseIf{$\measureloss{\lambda_l} > -\explb$}{
    \Return $\lambda_l$\;
}

\While{$i \leq \ceil{\log_2\tau}$}{
    Update $i \gets i + 1$\;
    Estimate $\measureloss{\lambda_l}$, update
    \begin{align}
        [\lambda_l, \lambda_r] \gets \begin{cases}
            [\lambda_l - 1/2^{i},\lambda_l], & \textrm{if $\measureloss{\lambda_l} > -\explb$}\\
            [\lambda_r - 1/2^{i},\lambda_r], & \textrm{otherwise}
        \end{cases}\;
    \end{align}
}

\Return $\lambda_r-1/2\tau$
\end{algorithm}

\begin{proposition}
    Under Assumptions~(\nref{assum:normalize},\nref{assum:reprod}), Algorithm~\ref{alg:binary search} requires $L\ceil{1+\log_2 \tau}$ quantum circuits, with each circuit uses
\begin{enumerate}[leftmargin=1em]
    \item [-] one ancilla qubit initialized in the zero state,
    \item [-] $\BigTO{\tau}$ queries to controlled-$U_H$ and its inverse, and
    \item [-] $8 L\Lambda^2\tau^3 \explb^{-2}$ measurement shots in average
,\end{enumerate}
    to produce a value $\lambda$ such that $\measureloss{\lambda}\leq -\explb$ and $\measureloss{\lambda+1/2\tau} > -\explb$ with failure probability $\ceil{\log_2 \tau} e^{-\tau}$.
\end{proposition}
\begin{proof}
    By the construction of Algorithm~\ref{alg:binary search}, the binary search halves the search interval during each iteration. It is updated in each iteration to $ \delta = 1 / 2^i $ until the condition $ \delta < 2 / \tau $ is satisfied. The number of iterations required to achieve the target precision is determined by the condition yielding the number of iteration as $\ceil{1 + \log_2(\tau)}$.

    The resource analysis proceeds as follows. By Theorem~\ref{thm:sample upper bound}, each estimation $\measureloss{\lambda^{(j)}}$ requires $L$ circuits that uses one ancilla qubit initialized in the zero state, $\BigTO{\tau}$ controlled-$U_H$ queries per circuit, and $8 L\Lambda^2\tau^3 \explb^{-2}$ measurement shots in average, with individual failure probability bounded by $e^{-\tau}$. Given the iteration count of $\ceil{1 + \log_2(\tau)}$ for the while loop, we require $L\ceil{1+\log_2 \tau}$ circuits. The overall success probability follows from the union bound:
\begin{equation}
    (1 - e^{-\tau})^{\ceil{1+\log_2 \tau}} \approx 1 - \ceil{\log_2 \tau} e^{-\tau}
,\end{equation}
    where the approximation holds via first-order Taylor expansion when $\tau \gg 0$.
\end{proof}

\subsection{Resource cost of Algorithm~\ref{alg:ground strict}}

\begin{proposition}~\label{prop:vqe step}
    Let $\delta \geq 0$, $k \geq 0$ and $\lambda \geq -\lambda_0$. Under Assumption~(\nref{assum:overlap},\nref{assum:nondegenerate}), when $\lambda - \delta \geq -\lambda_0$,
\begin{equation}
    \loss{\lambda-\delta} = e^{2\tau \delta} \loss{\lambda}
;\end{equation}
    when $-\lambda_0 > \lambda- \delta$,
\begin{equation}
    \loss{\lambda-\delta} = \left(e^{2\tau \delta} - \remain{\lambda}{\delta}\right) \loss{\lambda}
,\end{equation}
    where the remain term $\remain{\lambda}{\delta}$ is given as
\begin{equation}~\label{eqn:remain ratio}
    \remain{\lambda}{\delta} = \sum_{j: -\lambda_j > \lambda - \delta} \abs{c_j}^2 \left(\alpha^2 e^{-2\tau (\lambda_j +\lambda - \delta)} - \abs{\expfr(-\lambda_j)}^2 \right) \lambda_j / \loss{\lambda}
.\end{equation}
\end{proposition}
\begin{proof}
    This is proved by directly substituting Equation~\eqref{eqn:ideal vqe}. 
    Denote $\lambda' = \lambda - \delta$. When $\lambda' \geq -\lambda_{0}$,
\begin{equation}
    \loss{\lambda'} = \sum_{j} \abs{c_j}^2 \alpha^2 e^{-2\tau(\lambda_j + \lambda')} \lambda_j
    = \sum_{j} \abs{c_j}^2 \alpha^2  e^{-2\tau(\lambda_j + \lambda - \delta)} \lambda_j 
    = e^{2\tau \delta} \loss{\lambda}
.\end{equation}
    When $-\lambda_0 > \lambda'$,
\begin{align}
    \loss{\lambda'} &= \sum_{j: -\lambda_j \leq \lambda'} \abs{c_j}^2 \alpha^2 e^{-2\tau(\lambda_j + \lambda')} \lambda_j + \sum_{j: -\lambda_j > \lambda'} \abs{c_j}^2 \abs{\expfr(-\lambda_j)}^2 \lambda_j\\
    &= e^{2\tau \delta} \Big[ \loss{\lambda} - \sum_{j: -\lambda_j > \lambda'} \abs{c_j}^2 \alpha^2 e^{-2\tau(\lambda_j + \lambda)} \lambda_j \Big] + \sum_{j: -\lambda_j > \lambda'} \abs{c_j}^2 \abs{\expfr(-\lambda_j)}^2 \lambda_j\\
    &= e^{2\tau \delta} \loss{\lambda} + \sum_{j: -\lambda_j > \lambda'} \abs{c_j}^2 \left(\abs{\expfr(-\lambda_j)}^2 - \alpha^2 e^{-2\tau (\lambda_j +\lambda')}\right) \lambda_j
.\end{align}
\end{proof}

\begin{lemma}~\label{lem:ratio scaling}
    Let $\hat{x}, \hat{y} \in [-1, 0)$. Suppose $y$ is an estimation of $\hat{y}$ up to additive error $\eta \abs{y}$. If $x$ is an estimation of $\hat{x}$ with additive error at most $\eta \abs{y}$, then
\begin{equation}
    \abs{\frac{x - y}{y} - \frac{\hat{x} - \hat{y}}{\hat{y}}} 
    \leq \eta \left( 1 + \frac{\abs{\hat{x}}}{\abs{\hat{y}}} \right)
;\end{equation}
    if $x$ is an estimation of $\hat{x}$ with additive error at least $\eta' \abs{y}$, then
\begin{equation}
    \abs{\frac{x - y}{y} - \frac{\hat{x} - \hat{y}}{\hat{y}}} 
    \geq \max\set{0, \eta' - \eta \abs{\hat{x}}/\abs{\hat{y}}}
.\end{equation}
\end{lemma}
\begin{proof}
    Observe that
\begin{align}
    \abs{\frac{x - y}{y} - \frac{\hat{x} - \hat{y}}{\hat{y}}} 
    &= \abs{\frac{(\hat{x} - \hat{y})y - (x - y)\hat{y}}{y\hat{y}}} 
    = \abs{\frac{\hat{x}y - x\hat{y}}{y\hat{y}}} \\
    &= \abs{\frac{\hat{x}y - x\hat{y} + \hat{x}\hat{y} -\hat{x}\hat{y}}{y\hat{y}}}
    = \abs{\frac{(\hat{x} - x)\hat{y} + \hat{x}(y - \hat{y})}{y\hat{y}}}
.\end{align}
    When $\abs{\hat{x} - x} \leq \eta \abs{y}$, we have
\begin{equation}
    \abs{\frac{x - y}{y} - \frac{\hat{x} - \hat{y}}{\hat{y}}}
    \leq \abs{\frac{\hat{x} - x}{y}} + \abs{\frac{\hat{x}(y - \hat{y})}{y\hat{y}}}
    \leq \eta \left( 1 + \frac{\abs{\hat{x}}}{\abs{\hat{y}}} \right)
;\end{equation}
    when $\abs{\hat{x} - x} \geq \eta' \abs{y}$, we have
\begin{align}
    \abs{\frac{x - y}{y} - \frac{\hat{x} - \hat{y}}{\hat{y}}}
    &\geq \Big| \frac{\abs{\hat{x} - x}\cdot\abs{\hat{y}} - \abs{\hat{x}}\cdot\abs{y - \hat{y}}}{y\hat{y}} \Big|
    \geq \Big| \frac{\abs{\hat{x} - x} - \abs{\hat{x}}/\abs{\hat{y}}\cdot\abs{y - \hat{y}}}{y} \Big| \\
    &\geq \begin{cases}
        \frac{\abs{\hat{x} - x} - \abs{\hat{x}}/\abs{\hat{y}}\cdot\abs{y - \hat{y}}}{y} & \textrm{when } \abs{\hat{x} - x} > \abs{\hat{x}}/\abs{\hat{y}}\cdot\abs{y - \hat{y}}; \\
        0, & \textrm{otherwise};
    \end{cases} \\
    &\geq \begin{cases}
        \eta' - \eta \abs{\hat{x}}/\abs{\hat{y}} & \textrm{when } \eta' > \eta \abs{\hat{x}}/\abs{\hat{y}}; \\
        0, & \textrm{otherwise};
    \end{cases} \\
    &= \max\set{0, \eta' - \eta \abs{\hat{x}}/\abs{\hat{y}}}
.\end{align}
\end{proof}

\begin{proposition}~\label{prop:vqe ratio gap}
   Let $\delta, \eta \geq 0$ and $\lambda \geq -\lambda_0$. 
    Suppose Assumptions~(\nref{assum:overlap},\nref{assum:nondegenerate}) hold, and $\measureloss{\lambda - \delta}, \measureloss{\lambda}$ are estimations of $\loss{\lambda - \delta}, \loss{\lambda}$ up to additive error $\eta \abs{\measureloss{\lambda}}$, respectively. Denote $\reldiff = \left({\measureloss{\lambda'} - \measureloss{\lambda}}\right)/{\measureloss{\lambda}}$.
    When $\lambda - \delta \geq -\lambda_0$,
\begin{equation}
    \abs{\reldiff - (e^{2\tau\delta} - 1)} \leq \eta (e^{2\tau \delta}+ 1)
.\end{equation}
    When $-\lambda_0 > \lambda -\delta\geq 0$,
\begin{equation}
    \abs{\reldiff - (e^{2\tau\delta} - 1)} \geq \abs{\remain{\lambda}{\delta}} - \eta \left(1 + \abs{e^{2\tau\delta} - \remain{\lambda}{\delta}}\right)
,\end{equation}
    and if Assumptions~(\nref{assum:long evolution},\nref{assum:good overlap}) hold and $\remain{\lambda}{\delta} \geq 0$,
\begin{equation}
    \abs{\reldiff - (e^{2\tau\delta} - 1)} 
    \gtrsim e^{2\tau\delta} - (1 + \eta)\alpha^{-2} e^{2\tau(\lambda_0 + \lambda)} - \eta
.\end{equation}
\end{proposition}
\begin{proof}
    Denote $\lambda' = \lambda - \delta$. Proposition~\ref{prop:vqe step} implies
\begin{equation}
    \loss{\lambda'} / \loss{\lambda} = \begin{cases}
        e^{2\tau \delta}, & \textrm{when } \lambda' \geq -\lambda_0; \\
        e^{2\tau \delta} - \remain{\lambda}{\delta}, & \textrm{when } -\lambda_0 > \lambda'.
    \end{cases}
\end{equation}
    When $\lambda' \geq -\lambda_0$, $\loss{\lambda'} = e^{2\tau\delta}\loss{\lambda}$. Then by Lemma~\ref{lem:ratio scaling}, the conditions that $\abs{\measureloss{\lambda'} - e^{2\tau\delta}\loss{\lambda}} \leq \eta \abs{\measureloss{\lambda}}$ and $ \abs{\measureloss{\lambda} - \loss{\lambda}} \leq \eta \abs{\measureloss{\lambda}}$ give
\begin{align}
    \abs{\reldiff - (e^{2\tau\delta} - 1)} 
    &= \abs{\frac{\measureloss{\lambda'} - \measureloss{\lambda}}{\measureloss{\lambda}} - \frac{e^{2\tau\delta}\loss{\lambda} - \loss{\lambda}}{\measureloss{\lambda}}} \\
    &\leq \frac{\eta \abs{\measureloss{\lambda}}}{\abs{\measureloss{\lambda}}} \left( 1 + \frac{\abs{\loss{\lambda'}}}{\abs{\loss{\lambda}}} \right)
    = \eta (e^{2\tau \delta}+ 1)
.\end{align}
    When $-\lambda_0 > \lambda' \geq 0$, given that $\abs{\measureloss{\lambda'} - \loss{\lambda'}} \leq \eta \abs{\measureloss{\lambda}}$, the difference between $\loss{\lambda'}$ and $e^{2\tau \delta} \loss{\lambda}$ is lower bounded as
\begin{align}
    \abs{\measureloss{\lambda'} - e^{2\tau \delta}\loss{\lambda}} 
    &\geq \Big| \abs{\measureloss{\lambda'} - e^{2\tau \delta} \loss{\lambda} + \remain{\lambda}{\delta} \loss{\lambda}} - \abs{\remain{\lambda}{\delta} \loss{\lambda}} \Big| \\
    &= \Big| \abs{\measureloss{\lambda'} - \loss{\lambda'}} - \abs{\remain{\lambda}{\delta} \loss{\lambda}} \Big| \\
    &\geq \begin{cases}
        \left(\abs{\remain{\lambda}{\delta}} - \eta \right) \cdot \abs{\loss{\lambda}}, & \textrm{when } \abs{\remain{\lambda}{\delta}} > \eta; \\
        0, & \textrm{otherwise}
    \end{cases} \\
    &= \max\set{0, \abs{\remain{\lambda}{\delta}} - \eta} \cdot \abs{\loss{\lambda}}
.\end{align}
    Then by Lemma~\ref{lem:ratio scaling}, the conditions that $\abs{\measureloss{\lambda'} - e^{2\tau \delta}\loss{\lambda}} \geq \max\set{0, \abs{\remain{\lambda}{\delta}} - \eta} \abs{\loss{\lambda}}$ and $ \abs{\measureloss{\lambda} - \loss{\lambda}} \leq \eta \abs{\measureloss{\lambda}}$ give
\begin{align}
    \abs{\reldiff - (e^{2\tau\delta} - 1)}
    &\geq \max\set{0, \max\set{0, \abs{\remain{\lambda}{\delta}} - \eta} - \eta \abs{\loss{\lambda'}}/\abs{\loss{\lambda}}} \\
    &\geq \abs{\remain{\lambda}{\delta}} - \eta \left(1 + \abs{\loss{\lambda'}/\loss{\lambda}}\right) \\
    &= \abs{\remain{\lambda}{\delta}} - \eta \left(1 + \abs{e^{2\tau\delta} - \remain{\lambda}{\delta}}\right)
.\end{align}
    For the last statement, suppose $\tau \gg 0$ and $\abs{c_0}$ is not too small. For all $j > 0$, since $0 < \lambda_0 + \lambda < \lambda_j + \lambda$, one can assume $e^{-2\tau(\lambda_0 + \lambda)} \gg e^{-2\tau(\lambda_j + \lambda)} \to 0^+$ and hence
\begin{equation}
    \loss{\lambda} = \sum_j \abs{c_j^2} \alpha^2 e^{-2\tau(\lambda_j + \lambda)} \lambda_j 
    \approx \abs{c_0}^2 \alpha^2 e^{-2\tau(\lambda_0 + \lambda)}  \lambda_0
.\end{equation}
    Note that $\loss{\lambda}$ at this time is negative as $\lambda_0 < 0$. Then $\remain{\lambda}{\delta}$ can be approximately lower bounded as
\begin{align}
    \remain{\lambda}{\delta} 
    &= \sum_{j: \delta > \lambda + \lambda_j } \abs{c_j}^2 \left(\alpha^2 e^{-2\tau (\lambda_j +\lambda - \delta)} - \abs{\expfr(-\lambda_j)}^2\right) \lambda_j / \loss{\lambda} \\
    &\geq \sum_{j: \delta > \lambda + \lambda_j } \abs{c_j}^2 \left(\alpha^2 e^{-2\tau (\lambda_j +\lambda - \delta)} - 1\right) \lambda_j / \loss{\lambda} \\
    &\geq \abs{c_0}^2 \alpha^2 \left(e^{-2\tau (\lambda_0 +\lambda - \delta)} - \alpha^{-2}\right) \lambda_0 / \loss{\lambda} \\
    &\gtrsim \left(e^{-2\tau (\lambda_0 +\lambda - \delta)} - \alpha^{-2}\right)/e^{-2\tau(\lambda_0 + \lambda)} \\
    &= e^{2\tau\delta} - \alpha^{-2} e^{2\tau(\lambda_0 + \lambda)} \label{eqn:remain lower bound}
.\end{align}
Then the assumption $\remain{\lambda}{\delta} \geq 0$ gives the approximated lower bound of $\abs{\reldiff - (e^{2\tau\delta} - 1)}$ as
\begin{align}
    &\quad\quad\, \abs{\remain{\lambda}{\delta}} - \eta \left(1 + \abs{e^{2\tau\delta} - \remain{\lambda}{\delta}}\right) \\
    &\geq \begin{cases}
        \remain{\lambda}{\delta} - \eta \left(1 + \remain{\lambda}{\delta} - e^{2\tau\delta}\right), & \textrm{when } \remain{\lambda}{\delta} \geq e^{2\tau\delta}; \\
        \remain{\lambda}{\delta} - \eta \left(1 + e^{2\tau\delta} - \remain{\lambda}{\delta}\right), & \textrm{when } \remain{\lambda}{\delta} < e^{2\tau\delta}
    \end{cases} \\
    &= \begin{cases}
        (1 - \eta)\remain{\lambda}{\delta} - \eta \left(1 - e^{2\tau\delta}\right), & \textrm{when } \remain{\lambda}{\delta} \geq e^{2\tau\delta}; \\
        (1 + \eta)\remain{\lambda}{\delta} - \eta \left(1 + e^{2\tau\delta}\right), & \textrm{when } \remain{\lambda}{\delta} < e^{2\tau\delta};
    \end{cases} \\
    &\gtrsim \begin{cases}
        (1 - \eta)\left(e^{2\tau\delta} - \alpha^{-2} e^{2\tau(\lambda_0 + \lambda)}\right) - \eta \left(1 - e^{2\tau\delta}\right), & \textrm{when } \remain{\lambda}{\delta} \geq e^{2\tau\delta}; \\
        (1 + \eta)\left(e^{2\tau\delta} - \alpha^{-2} e^{2\tau(\lambda_0 + \lambda)}\right) - \eta \left(1 + e^{2\tau\delta}\right), & \textrm{when } \remain{\lambda}{\delta} < e^{2\tau\delta};
    \end{cases} \\
    &= \begin{cases}
        e^{2\tau\delta} - (1 - \eta)\alpha^{-2} e^{2\tau(\lambda_0 + \lambda)} - \eta, & \textrm{when } \remain{\lambda}{\delta} \geq e^{2\tau\delta}; \\
        e^{2\tau\delta} - (1 + \eta)\alpha^{-2} e^{2\tau(\lambda_0 + \lambda)} - \eta, & \textrm{when } \remain{\lambda}{\delta} < e^{2\tau\delta};
    \end{cases} \\
    &\geq e^{2\tau\delta} - (1 + \eta)\alpha^{-2} e^{2\tau(\lambda_0 + \lambda)} - \eta
.\end{align}
\end{proof}

\begin{proposition}~\label{prop:vqe stop condition}
    Suppose Assumptions~(\nref{assum:normalize},\nref{assum:long evolution},\nref{assum:overlap},\nref{assum:reprod},\nref{assum:good overlap},\nref{assum:nondegenerate},\nref{assum:priori}) hold. Let $\reldiff, \lambda_{r}, \lambda_{r}, \lambda_{rm}, \lambda_{lm}$ be as defined in each iteration of the while loop in Algorithm~\ref{alg:ground strict}. Let $\delta = \lambda_{rm} -\lambda_{lm}$. When $\abs{\reldiff - (e^{4\tau\delta} - 1)} \leq \tau^{-1}(e^{4\tau\delta} + 1)$,
\begin{equation}
    \lambda_{rm}  > - \lambda_0
;\end{equation}
    When $\abs{\reldiff - (e^{4\tau\delta} - 1)} > \tau^{-1}(e^{4\tau\delta} + 1)$,
\begin{equation}
    \lambda_{lm} < -\lambda_0
.\end{equation}
\end{proposition}
\begin{proof}
Proposition~\ref{prop:vqe ratio gap} is the main theory used to prove these two statements. 
We need to firstly show that the prerequisite for Proposition~\ref{prop:vqe ratio gap} are satisfied. By Theorem~\ref{thm:sample upper bound}, Algorithm~\ref{alg:expectation sample} can obtain the estimations up to additive error $\tau^{-1} \cdot \explb\leq\tau^{-1} \cdot \abs{\measureloss{\lambda_r}}$, with $\tau^{-1}$ will be used as $\eta$ in Proposition~\ref{prop:vqe ratio gap}.

Suppose $\lambda_{lm} = \lambda_r - 2\delta \geq - \lambda_0$. By Proposition~\ref{prop:vqe ratio gap}, $\abs{\reldiff - (e^{2\tau \cdot 2\delta} - 1)} \leq \tau^{-1}(e^{2\tau \cdot 2\delta} + 1)$.

Suppose $- \lambda_0 \leq \lambda_r \leq - \lambda_0 + \delta$. Then we have $\lambda_{lm} = \lambda_r - 2\delta \leq - \lambda_0 - \delta$.
Since the while loop ends only if $\lambda_r - \lambda_l \leq \tau$, during every iteration $\delta > 1 / 2 \tau$.
We first prove that $\remain{\lambda_{r}}{2\delta} \geq 0$ for such $\delta$.
Note that $\alpha$ in Equation~\eqref{eqn:alpha lowerbound} satisfies
\begin{equation}
    \alpha^{-2} < \frac{(1 - \tau^{-1})\,e^2 - 2\tau^{-1}}{(1 + \tau^{-1})\,e^1} \leq e^{1}
.\end{equation}
    Continued from Equation~\eqref{eqn:remain lower bound}, we hvae
\begin{equation}
    \remain{\lambda_{r}}{2\delta} 
    \geq e^{4\tau\delta} - \alpha^{-2} e^{2\tau(\lambda_0 + \lambda_r)}
    \geq e^{2} - \alpha^{-2} e^1 > 0
.\end{equation}
    Then by Proposition~\ref{prop:vqe ratio gap},
\begin{align}
    \abs{\reldiff - (e^{4\tau\delta} - 1)} 
    &\gtrsim e^{2\tau\cdot 2\delta} - (1 + \tau^{-1}) \alpha^{-2} e^{2\tau(\lambda_0 + \lambda_r)} - \tau^{-1} \\
    &\geq \tau^{-1}(e^{4\tau\delta} + 1) + \left[ (1 - \tau^{-1}) e^{4\tau\delta} - (1 + \tau^{-1})\alpha^{-2} e^{2\tau\delta} - 2\tau^{-1} \right]
.\end{align}
    For the last term surrounded by brackets, observe that
\begin{align}
    (1 - \tau^{-1}) e^{4\tau\delta} - (1 + \tau^{-1})\alpha^{-2} e^{2\tau\delta} - 2\tau^{-1} 
    &\geq (1 - \tau^{-1}) e^{2} - (1 + \tau^{-1})\alpha^{-2} e^{1} - 2\tau^{-1} \\
    &>  (1 - \tau^{-1}) e^{2} - \left((1 - \tau^{-1})\,e^2 - 2\tau^{-1}\right) - 2\tau^{-1} = 0
.\end{align}
    Then we have $\abs{\reldiff - (e^{4\tau\delta} - 1)} > \tau^{-1}(e^{4\tau\delta} + 1)$, as required.
    
    These two statements imply two contrapositives
\begin{align}
    \abs{\reldiff - (e^{4\tau\delta} - 1)} > \tau^{-1}(e^{4\tau\delta} + 1) &\implies \lambda_{lm}  < - \lambda_0, \\
    \abs{\reldiff - (e^{4\tau\delta} - 1)} \leq \tau^{-1}(e^{4\tau\delta} + 1) &\implies \lambda_{rm} > -\lambda_0 \textrm{ or } \lambda_r < -\lambda_0.
\end{align}
    According to the update rule of the Algorithm~\ref{alg:ground strict}, we have that if initially $\lambda_r > -\lambda_0$, then in later iterations $\lambda_r$ remains bigger than $-\lambda_0$. By the construction of Algorithm~\ref{alg:binary search}, we have $\measureloss{\lambda}\leq -\explb$ and $\measureloss{\lambda+1/2\tau}> -\explb$. By assumption of $B$, this implies $\lambda > -\lambda_0$. Thus we have
\begin{align}
    \abs{\reldiff - (e^{4\tau\delta} - 1)} \leq \tau^{-1}(e^{4\tau\delta} + 1) &\implies \lambda_{rm} > -\lambda_0.
\end{align}
\end{proof}

\begin{theorem}~\label{thm:lambda location}
Under Assumptions~(\nref{assum:normalize},\nref{assum:long evolution},\nref{assum:overlap},\nref{assum:reprod},\nref{assum:good overlap},\nref{assum:nondegenerate},\nref{assum:priori}), Algorithm~\ref{alg:ground strict} computes an estimate $\lambda \in [\abs{\lambda_0}, \abs{\lambda_0}+\tau^{-1}]$ with failure probability $\BigO{e^{-\tau}\log\tau}$, requiring at most $\BigO{L\log\tau}$ distinct quantum circuits. Each quantum circuit takes:
\begin{enumerate}[leftmargin=1em]
    \item[-] one ancilla qubit initialized in the zero state,
    \item[-] $\BigTO{\tau}$ queries to controlled-$U_H$ and its inverse, and
    \item[-] $8 L \Lambda^2\tau^3\explb^{-2}$ measurement shots in average.
\end{enumerate}
\end{theorem}
\begin{proof}
    Algorithm~\ref{alg:binary search} implies that one can query $\BigO{L\log \tau}$ quantum circuits to get a value $\lambda$ such that $\measureloss{\lambda}\leq -\explb$ with failure probability $e^{-\tau} \log\tau$. Then following steps divide current interval into three equal parts during each iteration and narrows the search range based on the key quantity, the relative difference $\reldiff$ using the two statements of Proposition~\ref{prop:vqe stop condition}
\begin{align}
    \abs{\reldiff - (e^{4\tau\delta} - 1)} > \tau^{-1}(e^{4\tau\delta} + 1) &\implies \lambda_{lm}  < - \lambda_0, \\
    \abs{\reldiff - (e^{4\tau\delta} - 1)} \leq \tau^{-1}(e^{4\tau\delta} + 1) &\implies \lambda_{rm} > -\lambda_0
.\end{align}
    The value of $\reldiff$ determines the behavior of the system under certain conditions. It indicates where $\lambda_0$ lies through comparison between $\abs{\reldiff - (e^{4\tau\delta} - 1)}$ and $\tau^{-1}(e^{4\tau\delta} + 1)$. Here, $\lambda_{lm}$ and $\lambda_{rm}$ are the left and right trisection points in the ternary search. 
    
    Due to the overlapping regions in the output of the discriminant, although it is not possible to determine which trisection interval $-\lambda_0$ lies in, the above discriminant can tell us which trisection interval that $-\lambda_0$ does not belong to. 
    When $\abs{\reldiff - (e^{4\tau\delta} - 1)} > \tau^{-1}(e^{4\tau\delta} + 1)$, we can conclude that $-\lambda_0$ is not in the leftmost trisection interval, and thus the left endpoint $\lambda_l$ can be contracted to the left trisection point $\lambda_{lm}$. When $\abs{\reldiff - (e^{4\tau\delta} - 1)} \leq \tau^{-1}(e^{4\tau\delta} + 1)$, it indicates that $-\lambda_0$ is not in the rightmost trisection interval, allowing the right endpoint $\lambda_r$ to be contracted to the right trisection point $\lambda_{rm}$. This allows us to update the search interval accordingly. Detailed proof of these two statements is deferred to Proposition~\ref{prop:vqe stop condition}. 
    
    The initial step size of the ternary search Algorithm~\ref{alg:ground complete} is $ \delta = \lambda / 3 $. By the update rule
\begin{align}
    [\lambda_l,\lambda_r] \gets \begin{cases}
        [\lambda_{lm}, \lambda_r], & \textrm{if $\abs{\reldiff - (e^{4\tau\delta} - 1)} > \tau^{-1}(e^{4\tau\delta} + 1)$} \\
        [\lambda_l, \lambda_{rm}], & \textrm{otherwise},
    \end{cases}
\end{align}
    the interval length decreases to $2/3$ of its previous value after each iteration until $\delta<\frac{1}{2\tau}$, which implies the number of iterations required to achieve the target precision is $\ceil{\log_{3/2}(4\tau\lambda/3)} $. Since $\lambda\leq 1$, the number is at most $\ceil{\log_{3/2}(4\tau/3)}$.
    
    The resource analysis proceeds as follows. By Theorem~\ref{thm:sample upper bound}, each estimation $\measureloss{\lambda}$ requires $L$ circuits that uses one ancilla qubit initialized in the zero state, $\BigTO{\tau}$ controlled-$U_H$ queries per circuit (Theorem~\ref{thm:qpp converge}), and $8L\Lambda^2\tau^3 \explb^{-2}$ measurement shots in average, with individual failure probability bounded by $e^{-\tau}$. Given the iteration count of at most $\ceil{\log_{3/2}(4\tau/3)}$ for the while loop, we require $2\ceil{\log_{3/2}(4\tau/3)} L$ circuits. The overall success probability follows from the union bound:
\begin{equation}
    (1 - e^{-\tau})^{2\ceil{\log_{3/2}(4\tau/3)}} \approx 1 - 2\ceil{\log_{3/2}(4\tau/3)} e^{-\tau},
\end{equation}
    where the approximation holds via first-order Taylor expansion when $\tau \gg 0$.

    Combining with Algorithm~\ref{alg:binary search}, one requires at most
\begin{equation}
    \BigO{2L\ceil{\log_{3/2}(4\tau/3)} + L\ceil{1+\log_2 \tau}}
    = \BigO{L\log\tau}
\end{equation}
    quantum circuits, and the overall success probability is
\begin{equation}
    (1 - e^{-\tau})^{2\ceil{\log_{3/2}(4\tau/3)} + \ceil{1+\log_2 \tau}} \approx 1 - \BigO{e^{-\tau} \log\tau}
.\end{equation}
\end{proof}

\section{Details and proofs of open system simulation}~\label{appendix:lindbladian}

The evolution of an $n$-qubit quantum state in an open quantum system at time $t > 0$ can be characterized by the Lindblad Master equation~\cite{manzano2020short}
\begin{equation}~\label{eqn:gksl}
    \ddx{t}\rho(t) = -\frac{i}{\hslash}[\Hsys, \rho(t)] + \sum_{j} \gamma_j \left( D_j \rho D_j^\dag - \frac{1}{2} \{ D_j^\dag D_j, \rho(t) \} \right)
,\end{equation}
where $\Hsys \in \textrm{Herm}(2^n)$ is the system Hamiltonian, $D_j \in \CC^{d \times d}$ is a jump operator as environment effect with coefficient $\gamma_j$.
In practice, $\hslash$ and $\gamma_j$ are absorbed into $H$ and $D_j$, respectively. By further short-handing the RHS as the effect of a \emph{Lindbladian} $\cL$, the equation reduces to a simpler form
\begin{equation}~\label{eqn:lindbladian}
    \ddx{t} \rho(t) = \cL[\rho(t)]
    \textrm{,\,\, where } \cL[\rho] \coloneqq -i[\Hsys, \rho] + \sum_{j}  D_j \rho D_j^\dag - \frac{1}{2} \{ D_j^\dag D_j, \rho \}
.\end{equation}
Then the solution of the Lindblad Master equation is $\rho(t) = e^{\cL t}[\rho(0)]$. 
The problem of preparing this solution can be approximately related to imaginary time evolution, by the following lemma.

\begin{lemma}[\cite{havel2003robust}]~\label{lem:lindbladian schrodinger}
    Equation~\eqref{eqn:lindbladian} can be expressed as the Schr\"{o}dinger form
\begin{equation}
    \ddx{t} \kett{\rho(t)} = -i( H_c - iH) \kett{\rho(t)}
,\end{equation}
    where $H_c, H$ are Hermitian given as
\begin{align}
    H_c &= I \ox \Hsys - \Hsys^T \ox I + \frac{i}{2} \sum_j (\overline{D_j} \ox D_j - D_j^T \ox D_j^\dag), \\
    H &= \frac{1}{2} \sum_j I \otimes D_j^\dag D_j + D_j^T \overline{D_j} \ox I - \overline{D_j} \ox D_j - D_j^T \ox D_j^\dag
.\end{align}
\end{lemma}
\begin{proof}
    Consider three identities for multiplications in vectorization:
\begin{equation}
    \kett{AX} = (I \ox A) \kett{X}, \quad
    \kett{XB} = (B^T \ox I) \kett{X}, \quad
    \kett{AXB} = (B^T \ox A) \kett{X}
.\end{equation}
    Then Equation~\eqref{eqn:lindbladian} can be translated to
\begin{align}
    \ddx{t} \rho(t)
    &= -i\left( \Hsys \rho(t) - \rho(t) \Hsys \right) + \sum_{j}  D_j \rho(t) D_j^\dag - \frac{1}{2} \left(D_j^\dag D_j \rho(t) + \rho(t) D_j^\dag D_j \right)\\
    \implies \ddx{t} \kett{\rho(t)} &= -i (\Hsys \ox I - \Hsys^T \ox I) \kett{\rho(t)} + \sum_j  \left[ (D_j^\dag)^T \ox D_j - \frac{1}{2} \left(I \ox D_j^\dag D_j + (D_j^\dag D_j)^T \ox I \right)  \right] \kett{\rho(t)} \\
    &= -i \left[ (I \ox \Hsys - \Hsys^T \ox I) - i \sum_j \frac{1}{2} \left(I \ox D_j^\dag D_j + D_j^T \overline{D_j} \ox I \right) - \overline{D_j} \ox D_j  \right] \kett{\rho(t)}
.\end{align}
    The result follows by splitting the non-Hermitian term $\overline{D_j} \ox D_j$ into one anti-Hermitian and one Hermitian term that goes to $H_c$ and $H$, respectively:
\begin{align}
    \overline{D_j} \ox D_j 
    &= \frac{1}{2}\left[ \left( \overline{D_j} \ox D_j + \left(\overline{D_j} \ox D_j \right)^\dagger \right) + \left( \overline{D_j} \ox D_j - \left(\overline{D_j} \ox D_j \right)^\dagger \right) \right] \\
    &= \frac{1}{2} \left( \overline{D_j} \ox D_j + D_j^T \ox D_j^\dag \right) + \frac{1}{2} \left( \overline{D_j} \ox D_j - D_j^T \ox D_j^\dag \right)
.\end{align}
\end{proof}

Now we begin to prove Theorem~\ref{thm:lindbladian}. To align the analysis with imaginary-time evolution, we slightly abuse the notations such that $\ket{\phi(t)} = \kett{\rho(t)} / \norm{\kett{\rho(t)}}$ and $\ket{\phi} = \ket{\phi(t)}$, $\gamma = \abs{\braket{\psi_0}{\phi(t)}}$, and $\ket{\psi_0}$ remains to be the unique ground state of $H$ with energy $\lambda_0$ and energy gap $\Delta$. Denote $\ket{\zeta_N}$ as the ideal output state (i.e., $\UITE{\tau}{H}$ precisely simulates $e^{-\tau H}$ in every step) of Algorithm~\ref{alg:lindbladian}. We have the following approximation bound.

\begin{lemma}~\label{lem:approx normalize}
    $\norm{\ket{\phi(t)} - \ket{\zeta_N}} \leq \BigO{\tau^2 N} / (\norm{\kett{\rho(t)}} + \norm{\kett{\rho(0)}})$.
\end{lemma}
\begin{proof}
    Equation~\eqref{eqn:lindbladian decomp} gives 
\begin{equation}
    \norm{\kett{\rho(t)} - \left( e^{-iH_c\tau} e^{-H\tau} \right)^N \kett{\rho(0)}} 
    \leq \BigO{\tau^2 N}
,\end{equation}
    while the normalized state of $\left( e^{-iH_c\tau} e^{-H\tau} \right)^N \kett{\rho(0)}$ is exactly the ideal output state $\ket{\zeta_N}$ of Algorithm~\ref{alg:lindbladian}, since
\begin{equation}
    \frac{\left( e^{-iH_c\tau} e^{-H\tau} \right)^2 \kett{\rho(0)}}{ \norm{\left( e^{-iH_c\tau} e^{-H\tau} \right)^2 \kett{\rho(0)}} }
    = \frac{e^{-iH_c\tau} e^{-iH\tau} \left( e^{-iH_c\tau} e^{-H\tau} \kett{\rho(0)} / \norm{e^{-iH_c\tau} e^{-H\tau} \kett{\rho(0)}} \right)}{ \norm{ e^{-iH_c\tau} e^{-H\tau} \left( e^{-iH_c\tau} e^{-H\tau} \kett{\rho(0)} / \norm{e^{-iH_c\tau} e^{-H\tau} \kett{\rho(0)}} \right)} }
\end{equation}
    and similar reasoning holds for $N > 2$. Let $A = \norm{\kett{\rho(t)}}$ and $B = \norm{\left( e^{-iH_c\tau} e^{-iH\tau} \right)^N \kett{\rho(0)}}$. Then one has 
\begin{align}
    \norm{\ket{\phi(t)} - \ket{\zeta_N}} 
    &= \norm{\frac{\kett{\rho(t)}}{A} - \frac{\left( e^{-iH_c\tau} e^{-H\tau} \right)^N \kett{\rho(0)}}{B}} \\
    &= \frac{\norm{\kett{\rho(t)} - \left( e^{-iH_c\tau} e^{-H\tau} \right)^N \kett{\rho(0)}} - (A - B)^2}{AB} \\
    &\leq \frac{2\norm{\kett{\rho(t)} - \left( e^{-iH_c\tau} e^{-H\tau} \right)^N \kett{\rho(0)}}}{A + B}
    = \frac{\BigO{\tau^2 N}}{\norm{\kett{\rho(t)}} + \norm{\kett{\rho(0)}}}
.\end{align}
\end{proof}

\renewcommand\thetheorem{\ref{thm:lindbladian}}
\begin{theorem}
    Let $N > 0$. Under Assumptions~(\nref{assum:normalize},\nref{assum:oracle},\nref{assum:long evolution}), the state fidelity between the output state of Algorithm~\ref{alg:lindbladian} and the normalized state of $\kett{\rho(t)}$ is approximately lower bounded as
\begin{equation}
    1 - \BigO{e^{1/N} N \epsilon + t^2/\mu N}
,\end{equation}
    where $\mu = \norm{\rho(0)}_2 + \norm{\rho(t)}_2$. 
    Moreover, taking $\epsilon = \BigO{\poly(t^{-1})}$, Algorithm~\ref{alg:lindbladian} uses the following cost,
\begin{enumerate}[leftmargin=1em]
    \item [-] $N$ queries to $U_{H_c}$,
    \item [-] $\BigTO{N t}$ queries to controlled-$U_H$ and its inverse,
    \item [-] a maximal total of $\BigTO{N t}$ query depth of $U_{H}$ and $N$ query depth of $U_{H_c}$, and
    \item [-] one ancilla qubit initialized in the zero state.
\end{enumerate}
\end{theorem}
\renewcommand{\thetheorem}{S\arabic{proposition}}
\begin{proof}
    Denote $U = U_{H_c}$ and $\UITEs = \UITE{\tau}{H}$ for convenience. At the $k$-th step, let $\ket{\varphi_k}$ be the output state in Algorithm~\ref{alg:lindbladian} with respect to the ideal output state $\ket{\zeta_k}$. Let $\gamma_k$ be the overlap between ground state $\ket{\psi_0}$ of $H$ and $\ket{\varphi_k}$.
    
    For error scaling, at the $(k + 1)$-th step, one can derive that
\begin{align}
    \norm{\ket{\varphi_{k + 1}} - \ket{\zeta_{k + 1}}}
    &= \norm{ U \UITEs [\ket{\varphi_k}] - \frac{U e^{-H\tau} \ket{\zeta_k}}{\norm{U e^{-H\tau} \ket{\zeta_k}}} }
    = \norm{ \UITEs [\ket{\varphi_k}] - \frac{e^{-H\tau} \ket{\zeta_k}}{\norm{e^{-H\tau} \ket{\zeta_k}}} } \\
    &\leq \norm{ \UITEs [\ket{\varphi_k}] - \UITEs[\ket{\zeta_k}] } + \norm{ \UITEs[\ket{\zeta_k}] - \frac{e^{-H\tau} \ket{\zeta_k}}{\norm{e^{-H\tau} \ket{\zeta_k}}} } 
.\end{align}
    In the RHS of inequality, the second term is bounded by $\BigO{\alpha^{-1} \epsilon}$ as proven in Corollary~\ref{coro:ite short}. As for the first term, note that error in this step is inherently the approximation error to the exponential function, which mainly occur near the singular point of energy spectrum, i.e., within the ground-state subspace. Therefore, by assuming $\ket{\varphi_k}$ and $\ket{\zeta_k}$ are mainly differed by their components in ground-state subspace of $H$, one can simplify the second term as $\norm{ \UITEs [\ket{\varphi_k}] - \UITEs[\ket{\zeta_k}] } \approx \norm{\ket{\varphi_k} - \ket{\zeta_k}}$, and hence
\begin{align}
     \norm{\ket{\varphi_k} - \ket{\zeta_k}} + \BigO{\epsilon} \gtrsim \norm{\ket{\varphi_{k + 1}} - \ket{\zeta_{k + 1}}}
.\end{align}
    Note that $\ket{\varphi_0} = \ket{\zeta_0} = \kett{\rho(0)} / \norm{\kett{\rho(0)}}$, the overall error between $\ket{\varphi_{k + 1}}$ and $\ket{\zeta_{k + 1}}$ is bounded as
\begin{equation}
    \norm{\ket{\varphi_{k + 1}} - \ket{\zeta_{k + 1}}} 
    \lesssim k \BigO{\alpha^{-1} \epsilon} + \norm{\ket{\varphi_0} - \ket{\zeta_0}} 
    = \BigO{k \alpha^{-1} \epsilon}
.\end{equation}
    The final error bound is given by Lemma~\ref{lem:approx normalize} as
\begin{align}
    \norm{\ket{\varphi_N} - \ket{\phi(t)}} 
    &\leq \norm{\ket{\varphi_N} - \ket{\zeta_N}} + \norm{\ket{\zeta_N} - \ket{\varphi_N}} \\
    &\lesssim \frac{ \BigO{\tau^2 N} }{\norm{\kett{\rho(t)}} + \norm{\kett{\rho(0)}}} + \BigO{N \epsilon}
    = \BigO{\frac{t^2}{\mu N}  + N\alpha^{-1} \epsilon}
.\end{align} 
    The result follows by $\norm{\kett{\rho}} = \sqrt{\trace{\rho^2}} = \norm{\rho}_2$ and $\alpha = e^{-1/N}$ in Algorithm~\ref{alg:lindbladian}.

    As for the resource cost, Algorithm~\ref{alg:lindbladian} consumes $N$ times of $U$ and $\UITEs$. Taking $\epsilon = \BigO{\poly(t^{-1})}$, Theorem~\ref{thm:qpp converge} implies query complexity of controlled-$U_H$ and its inverse for each $\UITEs$ is $\BigTO{t}$. The maximal query depth is the sum of $\BigTO{Nt}$ query depth of $U_H$ and $N$ query depth of $U$.
\end{proof}
\vspace{2em}

One can further adapt the setting in Ref.~\cite{cleve2017efficient} that there exists a matrix $\beta \in \CC^{(m + 1)\times L}$ and $2n$-qubit Pauli operators $\set{\sigma_j}_{j=1}^L$ such that
\begin{equation}
    \Hsys = \sum_{k=1}^L \beta_{0k} \sigma_k
    \textrm{\quad and \quad }
    D_j = \sum_{k=1}^L \beta_{jk} \sigma_k
\end{equation}
and $\beta_{0k} \in \RR$. Here $L$ is called the Pauli sparsity. Then $H_c$ and $H$ in Lemma~\ref{lem:lindbladian schrodinger} can be expressed as two linear combinations of $\BigO{L^2}$ Pauli operators, as shown in the following lemma.

\begin{lemma}~\label{lem:pauli count}
    If $\abs{\beta_{jk}} > 0$ for all $j, k$, then $H_c$ is a linear combination of at most $L^2 + 2L$ Pauli terms, and $H$ is a linear combination of at most $L^2 + 2L + 1$ Pauli terms. These two linear combinations can be obtained with $\BigO{(m+n)L^2}$ classical operations and $\BigO{L^2}$ classical memories.
\end{lemma}
\begin{proof}
For $H_c$, expand the three pieces:
\begin{align}
    I\ox \Hsys - \Hsys^T\ox I 
    &= \sum_{k=1}^L \beta_{0k}\,\big(I\ox\sigma_k - \sigma_k^T\ox I\big), \\
    \overline{D_j}\ox D_j 
    &= \sum_{k,\ell=1}^L \overline{\beta_{jk}}\,\beta_{j\ell}\,(\overline{\sigma_k}\ox \sigma_\ell), \\
    D_j^T\ox D_j^\dag 
    &= \sum_{k,\ell=1}^L \beta_{jk}\,\overline{\beta_{j\ell}}\,(\sigma_k^T\ox \sigma_\ell^\dag).
\end{align}
Transpose and conjugation only affect coefficients (and possibly signs), not the Pauli labels. Hence,
from $I\ox \Hsys - \Hsys^T\ox I$, we get at most $2L$ single-sided terms of form $I\ox\sigma_k$ or $\sigma_k\ox I$.
And, using $\abs{\beta_{jk}}>0$ for all $j,k$, from $\overline{D_j}\ox D_j$ and $D_j^T\ox D_j^\dag$, we get pairs of form $\sigma_k\ox\sigma_\ell$ for $1 \leq k, \ell \leq L$,
whose union over all $j$ contains at most $L^2$ distinct terms. Therefore,
\begin{equation}
    \text{number of Pauli terms in }H_c \;\leq\; L^2 + 2L.
\end{equation}

For $H$, expand similarly and use Pauli closure $\sigma_k\sigma_\ell=\phi_{k\ell}\,\sigma_{k\star \ell}$ with $\phi_{k\ell}\in\{\pm1,\pm i\}$:
\begin{align}
    I\ox D_j^\dag D_j 
    &= \sum_{k,\ell=1}^L \overline{\beta_{jk}}\,\beta_{j\ell}\,\phi_{k\ell}\; I\ox \sigma_{k\star \ell}, \\
    D_j^T\overline{D_j}\ox I 
    &= \sum_{k,\ell=1}^L \beta_{jk}\,\overline{\beta_{j\ell}}\,\phi_{k\ell}\; \sigma_{k\star \ell}\ox I, \\
    \overline{D_j}\ox D_j,\; D_j^T\ox D_j^\dag 
    &\Rightarrow \text{ pairs } \set{\sigma_k\ox\sigma_\ell:\,1\leq k,\ell\leq L}.
\end{align}
The products $D_j^\dag D_j$ and $D_j^T\overline{D_j}$ range over at most $L$ distinct single-Pauli labels $\sigma_{k\star \ell}$, so these two lines contribute at most $2L$ single-sided terms of form $I\ox\sigma_u$ and $\sigma_u\ox I$. The pair terms again contribute at most $L^2$ distinct $\sigma_k\ox\sigma_\ell$. In addition, since $\sigma_k\sigma_k=I$ for any Pauli $\sigma_k$, there can be an $I\ox I$ term, contributing at most one extra distinct Pauli string. Hence,
\begin{equation}
    \text{number of Pauli terms in }H \;\leq\; L^2 + 2L + 1.
\end{equation}
This establishes the stated worst-case bounds.

For the classical cost to compute the coefficients of these linear combinations from $\Hsys$ and $\{D_j\}_{j=1}^m$, proceed as follows. First, form the $L\times L$ Gram matrix $G$ with entries
\begin{equation}
    G_{k\ell} \coloneqq \sum_{j=1}^m \overline{\beta_{jk}}\,\beta_{j\ell},
\end{equation}
which costs $\BigO{mL^2}$ complex multiply-adds. The $2L$ single-sided contributions from $I\ox \Hsys - \Hsys^T\ox I$ use $\{\beta_{0k}\}_{k=1}^L$ and take $\BigO{L}$ additional operations. The pair terms for $H_c$ and $H$ coming from $\overline{D_j}\ox D_j$ and $D_j^T\ox D_j^\dag$ are read directly from $G$ and $G^T$ and accumulated into a dictionary of Pauli labels in $\BigO{L^2}$ updates.
For the single-sided terms in $H$ arising from $D_j^\dag D_j$ and $D_j^T\overline{D_j}$, one must map each pair $(k,\ell)$ to the product label $u=k\star \ell$ and phase $\phi_{k\ell}$. Using the standard binary symplectic representation of $2n$-qubit Paulis, computing $(u,\phi_{k\ell})$ takes $\BigO{n}$ bit operations per pair, so processing all pairs costs $\BigO{nL^2}$. 
Summing these contributions into the coefficient dictionaries is $\BigO{L^2}$ additional updates.
Collecting terms, the total classical time to obtain the two linear combinations is
\begin{equation}
    \BigO{mL^2} + \BigO{nL^2} \;=\; \BigO{(m+n)L^2},
\end{equation}
with $\BigO{L^2}$ memory to store intermediate accumulators and the final coefficient maps.
\end{proof}

\vspace{2em}
We are now ready to analyze the Trotter case.

\renewcommand\thetheorem{\ref{thm:lindbladian trotter}}
\begin{theorem}
    Under Assumptions~(\nref{assum:normalize},\nref{assum:pauli},\nref{assum:long evolution}), if $\rho(0)$ is a pure state,  Algorithm~\ref{alg:lindbladian} can prepare the normalized state of $\kett{\rho(t)}$ up to fidelity $1 - \BigO{\varepsilon}$ using the following cost:
\begin{enumerate}[leftmargin=1em]
    \item [-] $\BigTO{L^6 \Lambda^2 t^3 / \varepsilon}$ queries to (controlled) Pauli rotations, and
    \item [-] one ancilla qubit initialized in the zero state,
\end{enumerate}
    where $L$ is the number of Pauli terms and $\Lambda$ is the maximum absolute Pauli coefficient in the decompositions of $H_c$ and $H$.
\end{theorem}
\renewcommand{\thetheorem}{S\arabic{proposition}}
\begin{proof}
    Choose $N$ such that $N = \BigOmega{t^3}$ and $\BigO{N^{-1}} = \BigO{\poly(t^{-1})}$.
    By taking $\epsilon = \BigO{\poly(t^{-1})}$ and $\rho(0)$ as pure state, then $e^{1/N}N\epsilon$ and $t^2/\mu N$ in Theorem~\ref{thm:lindbladian} can be simplified as $\BigO{N\poly(t^{-1})}$ and $\BigO{t^2/N}$, respectively. For the Trotter case, since $\tau = t /N \leq  1 / t^2$ is small by Assumption~(\nref{assum:long evolution}), the key is to use the following approximated inequality
\begin{equation}
    \norm{e^{-iH} - e^{-iH^{\approx}}}_\infty \leq \eta \implies \norm{e^{-H\tau} - e^{-H^{\approx}\tau}}_\infty \leq \eta \tau \cdot  e^{(H - H^\approx) / t^2} \approx  \eta \tau
.\end{equation}
    We choose the Trotter steps for simulating $e^{-iH_c \tau}$ and $e^{-i H}$ to be $N_{\text{tr}}$.  By Lemma~\ref{lem:pauli count}, $H_c$ and $H$ can be represented by $\BigO{L^2}$ Pauli operators. Then total queries to Pauli rotations and controlled Pauli rotations, and maximal query depth are 
\begin{equation}
    \BigO{N_{\text{tr}} L^2 \cdot N} + \BigO{N_{\text{tr}} L^2 \cdot Nt} = \BigO{N \cdot N_{\text{tr}}  L^2 t}
.\end{equation}
    Denote $\Lambda$ as the maximum absolute coefficient of Pauli terms in these two representations.
    By Theorem~\ref{thm:1st order product formula} and above equation, the extra error introduced for $e^{-iH_c \tau}$ and $e^{-H\tau}$ is 
\begin{equation}
    \BigO{\frac{(L^2 \Lambda \tau)^2}{N_{\text{tr}}} \, \exp(\frac{L^2 \Lambda \tau }{N_{\text{tr}}})}
    \textrm{\quad and \quad } \BigO{\frac{(L^2 \Lambda \tau)^2}{N_{\text{tr}}} \, \exp(\frac{L^2 \Lambda \tau }{N_{\text{tr}}}) \cdot  \tau}
,\end{equation}
    respectively. By a similar reasoning in Lemma~\ref{lem:approx normalize}, the extra error introduced for the map $\ket{\psi}\mapsto e^{-H\tau} \ket{\psi} / \norm{e^{-H\tau} \ket{\psi}}$ remains the same magnitude (input states are pure states). Also, the proof of Theorem~\ref{thm:lindbladian} shows that the error propagates linearly. Since $\tau < 1$ and hence $\tau^3 < \tau^2$, the extra error introduced by Trotter decomposition over the algorithm is approximately 
\begin{equation}
    \BigO{N \cdot \frac{L^4 \Lambda^2 (\tau^2 + \tau^3)}{N_{\text{tr}}} \, \exp(\frac{L^2 \Lambda \tau }{N_{\text{tr}}})}
    = \BigO{\frac{t^2}{N} \cdot \frac{L^4 \Lambda^2}{N_{\text{tr}}} \, \exp(\frac{L^2 \Lambda }{t^2  N_{\text{tr}} })}
    \xrightarrow{t\gg 0}  \BigO{\frac{t^2}{N} \cdot \frac{L^4 \Lambda^2}{N_{\text{tr}}} }
.\end{equation}
    Taking $N_{\text{tr}} = L^4 \Lambda^2$, combining all three error terms gives the overall error bound
\begin{equation}
    \BigO{N\poly(t^{-1})} + \BigO{\frac{t^2}{N}} + \BigO{\frac{t^2}{N}} = \BigO{\frac{t^2}{N}}
.\end{equation}
    Let $\varepsilon = t^2 / N$ and the statement follows.
\end{proof}

\end{document}